\documentclass[defaultstyle,11pt]{thesis}

\renewcommand{\baselinestretch}{1.0}

\usepackage{amsmath,amsfonts,amssymb,amsthm}
\usepackage{physics}
\usepackage{mathtools}
\usepackage{stmaryrd}
\usepackage{expdlist}
\usepackage{ytableau}
\usepackage[all]{xy}
\usepackage{bbold}

\providecommand{\onecolumn}{} 

\usepackage{graphicx,xcolor}
\usepackage{hyperref}
\usepackage{url}

\usepackage{verbatim}
\usepackage{alltt}
\usepackage{moreverb}

\usepackage{xr} 

\usepackage{enumitem}

\usepackage{relsize}

\usepackage{bbm}

\usepackage{geometry}
\usepackage{tikz}
\usetikzlibrary{matrix}
\usetikzlibrary{calc}
\usetikzlibrary{fit}

\usepackage{scalerel}

\providecommand{\ignore}[1]{}

\newif\ifcmnt
 \cmntfalse
\ifdefined\cmntsoff\cmntfalse\fi
\ifcmnt
    \providecommand{\aucmnt}[1]{#1}
    \providecommand{\jrvcolor}[1]{\textcolor{red}{#1}}
    
\else
    \providecommand{\aucmnt}[1]{}
    \providecommand{\jrvcolor}[1]{{#1}}
    
\fi

\newcommand{\jrvc}[1]{\aucmnt{\jrvcolor{[#1]}}}

\providecommand{\ignore}[1]{}

\newif\ifcmnt
\cmnttrue
\ifdefined\cmntsoff\cmntfalse\fi

\ifcmnt
    \providecommand{\aucmnt}[1]{#1}

\else
    \providecommand{\aucmnt}[1]{}

\fi

\newcommand{\rls}{\mathbb{R}}
\newcommand{\cmplx}{\mathbb{C}}

\newcommand{\nats}{\mathbb{N}}

\newcommand{\cH}{\mathcal{H}}

\newcommand{\cP}{\mathcal{P}}
\newcommand{\cQ}{\mathcal{Q}}

\newtheorem{theorem}{Theorem}
\newtheorem{lemma}[theorem]{Lemma}
\newtheorem{corollary}[theorem]{Corollary}

\newcommand{\setydiagramtext}{\ytableausetup{nosmalltableaux}\ytableausetup{boxsize=0.4em,centertableaux}}
\newcommand{\setydiagrameq}{\ytableausetup{nosmalltableaux}\ytableausetup{smalltableaux,centertableaux}}
\newcommand{\subydiagramtext}[1]{{\ytableausetup{boxsize=0.2em,nocentertableaux}\ydiagram{#1}\setydiagramtext}}
\newcommand{\subydiagram}[1]{{\ytableausetup{boxsize=0.2em,nocentertableaux}\ydiagram{#1}\setydiagramtext}}
\newcommand{\subydiagrameq}[1]{{\ytableausetup{boxsize=0.2em,nocentertableaux}\ydiagram{#1}\setydiagrameq}}

\let\s\shortstack
\let\ytab\ytableaushort

\newcommand{\lambdap}{\lambda'\,}
\newcommand{\mup}{\mu'\,}
\newcommand{\nup}{\nu'\,}
\newcommand{\bzero}{\mathbf{0}}
\newcommand{\tphi}{\tilde{\phi}}
\newcommand{\mus}{{\mu^{(1)}\cdots\mu^{(N)}}}
\newcommand{\lambdas}{{\lambda^{(1)}\cdots\lambda^{(N)}}}
\newcommand{\musk}{{\mu^{(1)}\cdots\mu^{(k)}}}
\newcommand{\musp}{{\mu^{(1)\prime}\cdots\mu^{(N)\prime}}}
\newcommand{\chimn}{\chi_\mus^\nu}

\newcommand{\chimnp}{\chi_\mus^{\nu'}}

\newcommand{\bepsilon}{\boldsymbol{\epsilon}}

\newcommand{\blambda}{\boldsymbol{\lambda}}
\newcommand{\bmu}{\boldsymbol{\mu}}
\newcommand{\bnu}{\boldsymbol{\nu}}
\newcommand{\bigboxtimes}{\bigotimes}

\newcommand{\tp}{\intercal}
\newcommand{\cR}{\mathcal{R}}
\newcommand{\one}{\mathbb{1}}
\newcommand{\JM}{\text{Jucys--Murphy }}
\DeclareMathOperator{\End}{End}
\DeclareMathOperator{\Aut}{Aut}
\newcommand{\Ind}{\text{Ind}}

\DeclareMathOperator{\Hom}{Hom}
\newcommand{\Span}{\text{span}}

\newcommand{\SU}{\text{SU}}

\newcommand{\Id}{\text{Id}}
\newcommand{\h}{\mathfrak{h}}
\newcommand{\gl}{\mathfrak{gl}}
\newcommand{\g}{\mathfrak{g}}
\newcommand{\osp}{\mathfrak{osp}}
\newcommand{\sla}{\mathfrak{sl}}
\newcommand{\proper}{\text{ proper}}
\newcommand{\so}{\mathfrak{so}}
\newcommand{\sort}{\text{sort}}
\newcommand{\spa}{\mathfrak{sp}}

\newcommand{\su}{\mathfrak{su}}
\newcommand{\GL}{\text{GL}}
\newcommand{\U}{\text{U}}
\newcommand{\SL}{\text{SL}}
\newcommand{\diag}{\text{diag}}

\newtheorem{prop}[theorem]{Proposition}

\theoremstyle{definition}
\newtheorem{exmp}[theorem]{Example}
\newtheorem{Definition}[theorem]{Definition}

\makeatletter
\def\renewtheorem#1{%
  \expandafter\let\csname#1\endcsname\relax
  \expandafter\let\csname c@#1\endcsname\relax
  \gdef\renewtheorem@envname{#1}
  \renewtheorem@secpar
}
\def\renewtheorem@secpar{\@ifnextchar[{\renewtheorem@numberedlike}{\renewtheorem@nonumberedlike}}
\def\renewtheorem@numberedlike[#1]#2{\newtheorem{\renewtheorem@envname}[#1]{#2}}
\def\renewtheorem@nonumberedlike#1{  
\def\renewtheorem@caption{#1}
\edef\renewtheorem@nowithin{\noexpand\newtheorem{\renewtheorem@envname}{\renewtheorem@caption}}
\renewtheorem@thirdpar
}
\def\renewtheorem@thirdpar{\@ifnextchar[{\renewtheorem@within}{\renewtheorem@nowithin}}
\def\renewtheorem@within[#1]{\renewtheorem@nowithin[#1]}
\makeatother

\title{Universality of swap for qudits: a representation theory approach}

\author{James~R.}{van Meter}

\otherdegrees{Ph.D. physics, University of California at Davis, 2003 \\
              B.S. mathematics, University of California at Santa Cruz, 1994 }

\degree{Doctor of Philosophy}           
        {Ph.D., mathematics}         

\dept{Department of}                    
        {Mathematics}                

\advisor{}                         
        {Emanuel Knill}                      

\reader{Richard M. Green}           

\abstract{  \OnePageChapter     

An open problem of quantum information theory has been to determine under what conditions universal exchange-only computation is possible for qudits encoded on 
$d$-state systems for $d>2$.  
This problem can be posed in terms of representation theory by recognizing that each quantum mechanical swap, generated by exchange-interaction, can be identified with a transposition in a symmetric group, each $d$-state system can be identified with the fundamental representation of $\SU(d)$, and each encoded qudit can be identified with an irreducible representation of a Lie algebra generated by transpositions.   Towards this end we first give a mathematical definition of exchange-only universality
in terms of a map from the special unitary algebra on the product of qudits into a representation of a Lie algebra  generated by transpositions.   We show that this definition is consistent with quantum computing requirements.
We then proceed with the task of characterizing universal families of qudits, that is families of encoded qudits admitting exchange-only universality.
This endeavor is aided by the fact that the irreducible representations corresponding to qudits are canonically labeled by partitions.  In particular we derive necessary and sufficient conditions for universality on one or two such
qudits, in terms of simple arithmetic conditions on the associated partitions.  We also derive necessary and
sufficient conditions for universality on arbitrarily many such qudits, in terms of Littlewood--Richardson coefficients.
Among other results, we prove that universal families of multiple qudits are upward closed,
that universality is guaranteed for sufficiently many qudits, and that any family that is not universal can be made so by simply adding at most five ancillae.
We also obtain results for 2-state systems as a special case.

        }

\dedication[Dedication]{        
\begin{center}
To my parents.
\end{center}
        }

\acknowledgements{      \OnePageChapter 
I thank Emanuel Knill for the principal motivation and guidance behind this thesis,
and Richard Green for critical help and instruction along the way.
        }

\begin{document}

\chapter{Introduction}

Although still in their infancy, quantum computers promise to speed up various computations by many orders of magnitude, relative to classical computers, and perform tasks that are presently unfeasible.
Such capabilities are made possible by taking advantage of unique features of quantum mechanics.
For example, the composition of multiple quantum systems is described by the tensor product of the same number of Hilbert spaces.  As a consequence, the resources required to simulate such a composite system
grow exponentially with the number of components on a classical computer, but only polynomially on a quantum computer.  This is an auspicious fact for the future of computational chemistry, and also the principle
behind Google's recent experimental demonstration of a quantum speedup \cite{google}.  Key to such success is that rather than operating on classical bits, each equal to 0 or 1, a quantum computer effectively operates on superpositions of 0s and 1s (and maybe additional values), which are vectors in two-dimensional Hilbert spaces called qubits (or qudits for arbitrary dimensions).  Various clever ways have been found to exploit such superposed logical values, and particularly the cross-terms, also known as interference terms, that result from products of multiple such qudits.
Applications range from efficiently factoring large integers to searching unstructured databases faster than any classical computer. 

There are myriad possibilities for quantum computing architecture.  Despite the early success of Google's prototype in which qubits are implemented with superconducting circuits, there is no consensus on the best approach to building a quantum computer.  Indeed there are a great many proposals in regards to which specific materials to use, ways to encode information into physical systems, and how to translate computer operations into physical processes. 
Probably the most popular choice for the physical representation of a qubit, in theoretical discussions, is the spin of a particle.  

Here we consider a less common alternative, that of encoding information in the permutations of particles. 
Such encoding has the advantage of being impervious to certain kinds of noise that affect only spins \cite{Zanardi}.  It also has the advantage that, in certain physical systems, permutations of particles are conveniently enacted by the quantum
mechanical exchange interaction, which generates swaps between particle
states \cite{Loss}.
This operation is perhaps the purest physical realization of a transposition in a symmetric group.  Unsurprisingly, then, symmetric group representation theory plays a central role in the mathematical description of this approach to computation.
Indeed a $D$-dimensional qudit encoded on $n$ particles, as described above, may be identified with 
a $D$-dimensional irreducible representation of the symmetric group $S_n$.  
It may also be identified with a $D$-dimensional irreducible representation
of the Lie algebra generated by the transpositions in $S_n$.   
The elements of this ``algebra of transpositions", which physically correspond to Hamiltonians, 
may then be exponentiated to form operators serving as logical gates on the qudits.

This approach to computation is not merely theoretical.
Precise control of the exchange-interaction has been experimentally demonstrated between
specially confined particles in silicon semiconductors known as quantum dots \cite{QuTech,HRLa,HRLb}.
Silicon-based quantum dots have even been used to physically realize the two-dimensional irreducible
representation of the symmetric group on three elements \cite{HRLa}.
Such an approach is potentially competitive: 
although superconducting quantum computers benefit from optimal
conductivity, a quantum dot approach may be more efficient since it involves
fewer particles over shorter distance scales.

To live up to its name, exchange-only computation requires that every desired logical gate,
or unitary operator, be implementable by exchange interactions.
In the experiment mentioned above, it was shown that every unitary operator could be implemented by exchange-only means on the encoded qubit \cite{HRLa}.
More generally, the capability to implement every unitary operator on given qudits is called
{\it universality}.  Kempe et al. have shown that exchange-only universality is always achievable
for qudits that are encoded in permutations of particles corresponding to 2-state systems \cite{PhysRevA.63.042307}.
The number of states has the effect of constraining the irreducible 
representations considered.
The main effort of the present work generalizes the theorem of 
Kempe et al. to $d$-state systems of arbitrary $d$, on which universality does not always hold.   
Their result for 2-state systems is obtained here as a special case.
We also investigate in some detail the relative merits of qudits of various dimensions
encoded in 2-state systems.

The bulk of this thesis considers $d$-state systems for $d>2$.
Such facilitates efficient construction of qudits of arbitrarily large dimension,
which may have computational advantages.
To the author's knowledge, exchange-only computation with $d$-state systems for $d>2$
has never before been attempted experimentally, nor prior to this work investigated
theoretically.
Therefore, an objective of this thesis is to lay the theoretical foundation for this novel kind
of computation. 

A first step is to define with mathematical precision the universality of exchange-only computation on the encoded qudits.  
An essential requirement is that the algebra of transpositions generate, in some way,
every desired unitary operator on the product of the qudits.  We develop such a definition of universality in terms of the existence of an injective linear 
``universality witnessing" map from the special unitary algebra on the product of the qudits onto a physically motivated representation $\rho$ of the algebra of 
transpositions.  By ``physically motivated representation" we mean, first, that the vectors spanning the representation space correspond to states of the 
physical system, and second that the action of transpositions on these vectors describe the physical effect of the exchange-interaction on the corresponding 
states.  Physical consistency further requires compatibility of this universality-witnessing map with symmetric group intertwiners from the product of qudits to 
the representation space of $\rho$, reflecting the fact that the objects being permuted have a physical existence independent of the choice of representation.

There is however no unique way to define universality with the above properties.
For example, different definitions may result in universality-witnessing maps differing by scalars.
Such scalars are considered ``phase factors" in the relevant physical context, and have
no bearing on the desired computation.  Indeed multiple, seemingly disparate definitions of universality arise naturally from various physical considerations.  
One contribution of this thesis is to prove the mathematical equivalence of such definitions.

For the purpose of defining and proving exchange-only universality, it is helpful to note that each irreducible representation of the symmetric group $S_n$, 
and thus each qudit encoded on $n$ particles, is canonically labeled by a partition of the number $n$.  
More specifically, if the qudit is encoded on $d$-state systems,
it is labeled by a partition having $d$ parts.
This labeling is such that properties of the 
representation can be deduced from properties of the associated partition.  For example, a representation is one-dimensional
if and only if its labeling partition is trivial, that is consisting either of a single part of size $n$, or $n$ parts of size 1.  
For another example, 
if two partitions are conjugates of each other then
the associated representations are also related to each other, and 
in particular have the same dimension.  

Given that each qudit is labeled by a partition,
we find it convenient to define universality as a property of a family of partitions.
We then show that universality can be determined by properties of the partitions alone, without explicit reference to the associated group or 
Lie algebra representations.
One important result derived here is that a single-partition family, or singleton, 
is universal if and only if simple conditions are satisfied by the partition.
For arbitrarily many, nontrivial partitions, necessary and sufficient conditions
for universality are again obtained in terms of the partitions.

Here we should distinguish between the goal of characterizing universal families of partitions,
and the more ambitious goal of efficiently characterizing universal families of partitions.
The conditions referred to above for the universality of arbitrarily many nontrivial partitions
are not efficient, because they involve a number of computational steps that grows faster than a polynomial in 
the total number of particles.  On the other hand, we derive efficient conditions that are both necessary and sufficient for the universality of 
two nontrivial partitions, in a simple arithmetic form.
We also derive efficient necessary and sufficient conditions for the universality of arbitrarily many nontrivial partitions provided at least one is not self-conjugate.

The task of efficiently characterizing universal families of partitions is further simplified by another important result of this thesis, that universal families of multiple nontrivial partitions are upward closed.
This means that if such a family is contained in a larger family, the larger family is also universal.
Therefore the problem of characterizing universal families of nontrivial partitions reduces to the problem of characterizing minimal such families.
This problem is made tractable by another useful result: that for a sufficiently large but finite number of nontrivial partitions, universality is guaranteed.  
This last result further implies that as an exchange-only quantum computer is scaled up, universality is inevitable.

Finally, given a family of partitions that is not universal, we prove that it can be made so via only slight modification, 
which in principle can be realized by the corresponding physical system.
In particular we show that inclusion of the simple $3+2$ partition of the number 5 
-- and thus the addition of at most 5 ancillary particles to the system --  ensures the universality of any family of partitions.  
Alternatively the physical system can be constrained so as to also constrain the physically relevant representation of the algebra of transpositions, which we show simplifies the task of achieving universality.

The results of this thesis are enabled by two fields of study, in particular. 
One is symmetric group representation theory,
from Young's seminal formalism through Marin's recent development of the
representation theory of the algebra of transpositions \cite{MARIN2007742}.  The other is the 
interdisciplinary topic of Littlewood--Richardson coefficients, beginning naturally with 
Littlewood and Richardson and culminating with Knutson and Tao's proof that
Horn's inequalities are necessary and sufficient to determine whether each coefficient 
is nonvanishing \cite{1998math......7160K}.

It may be of interest to note that this thesis depends crucially on mathematical 
discoveries that are 
relatively recent.  Any attempt to obtain the results of this thesis without Marin's 
2003 work would require, in effect, that the latter be independently derived.
Meanwhile Knutson and Tao's results on necessary and sufficient conditions for
nonvanishing Littlewood--Richardson coefficients facilitated proof of one of the main
theorems of this thesis, that being necessary and sufficient conditions for 
exchange-only universality on two qudits.
While the latter may have been achievable without Knutson and Tao's breakthrough,
the combinatorics involved would have been formidable. 

This thesis is organized as follows.
In Chapter~2 we present the needed mathematical background mentioned above.
In Chapter~3 we give a brief overview of exchange-only quantum computation and discuss the relative merits of 
higher dimensional encoded qudits.  We note in particular that a fixed collection of $d$-state systems can be used to
encode various possible qudits of various dimensions, depending on how they are initialized. 
Then we expand on an observation made in \cite{PhysRevA.63.042307} 
of a feature unique to exchange-only quantum computation, namely that the higher the dimension of the encoded qudit,
the higher its ``coding efficiency".
We also give design specifications for exchange-only universality, as motivated by quantum computation, before defining it mathematically in terms of partitions in Chapter~4.  We then proceed in the rest of Chapter~4 to prove a number of original theorems, alluded to above, on characterizing universal families of partitions.  
Finally, conclusions and open problems are discussed in Chapter~5.

\chapter{Mathematical preliminaries}

Below we summarize the mathematical background needed for subsequent chapters.
Section~2.1 reviews representation theory and introduces Young's formalism for
symmetric group representation theory,
defines the algebra of transpositions which plays a central role in this work, and 
describes its ``Jucys--Murphy elements"  which prove to be valuable tools.
Section~2.2 summarizes Marin's pioneering work on the algebra of transpositions,
Section~2.3 reviews Schur--Weyl duality and the role played by Littlewood--Richardson
coefficients in representation theory, and Section~2.4 highlights relevant results
on calculating Littlewood--Richardson coefficients.
Notation and terminology is introduced as needed throughout.

\section{Representation theory}
We require some representation theory of groups and Lie algebras.
The former concerns group homomorphisms of the form  
 $\rho:G\rightarrow\GL(V)$, or equivalently the
$G$-module $V$, and the latter concerns the Lie algebra homomorphisms of the 
form,
$\rho:\g\rightarrow\gl(V)$, or equivalently the $\g$-module $V$.
Of course transpositions, as both group elements and Lie algebra elements,
figure prominently in this discussion. 
The material for this section is mostly drawn from 
Young \cite{young1977collected}, Fulton and Harris \cite{fulton1991representation}, and Marin \cite{MARIN2007742}. 

\subsection{Symmetric and unitary groups}
\setydiagrameq

Representations of symmetric groups are of central importance to this thesis,
so we now introduce conventions in that regard.
The symmetric group denoted by $S_n$ is the group of permutations of  
$\{1,\ldots,n\}$.  
When there is a need to be explicit, the elements of the symmetric group $S_n$ are given by permutations of $\{1,\ldots,n\}$ in the form of products of $k$-cycles, where a $k$-cycle is denoted in the usual way: $(i_1 i_2 \ldots i_k)$.
A subgroup of $S_n$ consisting of 
permutations of a subset $A$ of $\{1,\ldots,n\}$ is signified by $S_A$.  
For example, 
$S_{\{4,5,6\}}$ is a subgroup of $S_6$. 
Also we treat $S_{\{1,\ldots,\ell\}}\times S_{\{\ell+1,\ldots,n\}}$ as the obvious
subset of $S_n$, with elements $(s,t)$ of the former identified with 
elements $s \circ t$ of the latter, and therefore write $S_{\{1,\ldots,\ell\}}\times S_{\{\ell+1,\ldots,n\}}\subset S_n$.
(Note we overload the symbol ``$\subset$" to mean ``subset", 
``subgroup", or ``subalgebra", depending on the context, and allow equality in every case.)
Furthermore, in the isomorphic Cartesian product $S_\ell\times S_{n-\ell}$,
we assume the elements of the sets being permuted are labeled so as to allow the identification: $S_\ell\times S_{n-\ell}=S_{\{1,\ldots,\ell\}}\times S_{\{\ell+1,\ldots,n\}}$.
More generally, we write
$$
\prod_{i=1}^N S_{m_i}=\prod_{i=1}^N S_{\{n_{i-1}+1,\ldots,n_i\}}
$$
where $n_i=\sum_{j=1}^im_j$.

Recall that the number of irreducible representations of a finite group equals the 
number of conjugacy classes of that group, which motivates determination of the conjugacy 
classes of $S_n$.  For that purpose we observe that every element can be expressed as the 
product of disjoint cycles.  Further, because every element is generated by 
transpositions, and conjugation by a transposition only transposes numbers but does not 
change the length of a cycle, it follows that conjugation by any element does not change 
the lengths of cycles but only the numbers within.  For example: $(1 2 3)(4 5)$ 
conjugated by $(3 4)$ is $(1 2 4)(3 5)$.  Therefore, each conjugacy class is characterized by the multiset of lengths of the disjoint cycles that compose each of its representative elements.
Since the lengths of these disjoint cycles add up to $n$ (including trivial 1-cycles), 
each conjugacy class is uniquely characterized by a partition of $n$.  So the irreducible 
representations of $S_n$ are in bijection with the partitions of $n$.

There is, in particular, a canonical bijection between the irreducible representations of
$S_n$ and the partitions of $n$, by which it proves useful to label the former by the latter.
Therefore we introduce special notation for that purpose.
We denote partitions by certain Greek letters, typically $\kappa$, $\lambda$, $\mu$, and $\nu$,
or by a Greek letter with a superscripted number in parenthesis in the case of a family of partitions, e.g. $\mu^{(i)}$ where $i\in\nats$ .   
These are partitions of numbers sometimes denoted by corresponding Roman letters, e.g.
$\kappa \vdash k$,
$\lambda \vdash \ell$,
$\mu \vdash m$,
$\nu \vdash n$,
$\mu^{(i)} \vdash m_i$,
where ``$\vdash$" 
indicates the number being partitioned. 
We may also refer to the size of a partition $\lambda$, denoted as $|\lambda|$, which 
equals the number being partitioned, so in particular 
$\lambda\vdash|\lambda|$.
A part of a partition is indicated by a subscript, so that $\lambda_i$ is the 
$i$th part, where by convention it is assumed the parts are in weakly descending order so 
that $\lambda_i\geq\lambda_{i+1}$ for all $i\in\nats$. 
When there is a need to explicitly specify a partition, it may be given by its 
nonzero parts listed in square brackets: $[\lambda_1,\lambda_2,\ldots,\lambda_d]$, 
where $\lambda_i=0$ for $i>d$.  Further, in this notation, multiple parts with the same 
length may be denoted by the length of that part superscripted 
by its repetition, so for example $[1,1,1]=[1^3]$.

A pictorial notation for a partition is the {\it Young diagram}, 
in which parts 
are depicted by rows of squares arranged by descending order vertically.  For example the 
Young diagram for $[3,2,2]$ is:
$$
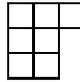
\ydiagram{3,2,2}
$$
An advantage of this presentation is that the squares may be filled in with numbers, 
for various purposes.  In particular a {\it Young tableau} is a Young diagram $\lambda$ 
in which the squares have been filled in with the numbers from 1 to $|\lambda|$, with no 
repetition, such that they increase along each row and down each column.  For example:
$$
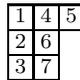
\ytableaushort{145,26,37}
$$
\setydiagramtext
The number of Young tableaux for a given partition $\lambda$ equals the 
dimension of the irreducible representation of $S_{|\lambda|}$ canonically labeled by 
$\lambda$ and denoted by $V_\lambda$ or, when there is no confusion, simply by $\blambda$.
For example, if $\lambda=\ydiagram{3,1}$ then $\dim(\blambda)=3$ because there are
three Young tableaux of this shape:
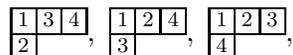
\setydiagrameq
$$
\ytableaushort{134,2},\ \ytableaushort{124,3},\ \ytableaushort{123,4},
$$
\setydiagramtext
while if $\lambda=\ydiagram{3,2}$ then $\dim(\lambda)=5$ because there are five Young
tableaux of this shape:
\setydiagrameq
$$
\ytableaushort{135,24},\ \ytableaushort{134,25},\ \ytableaushort{125,34},\ \ytableaushort{124,35},\ \ytableaushort{123,45}
$$
These tableaux can serve as basis vectors for the associated $S_{|\lambda|}$-module.
In Section~2.1.3 we show how the action of $S_{|\lambda|}$ on $\blambda$ can be specified in terms of the tableaux.

\setydiagramtext
Young diagrams also conveniently label irreducible representations of certain other
groups, notably certain subgroups of the general linear group $\GL(d,\cmplx)$.  Although we are primarily concerned
with the symmetric group, the special unitary group also plays an important role in this
work.  Finite dimensional irreducible representations of $\SU(d)$
are canonically labeled by Young diagrams of at most $d$ rows.
The dimension of the representation thus labeled is given by the number of ways in which
the integers between 1 and $d$ can fill the squares of the Young diagram, with repetition,
such that they are weakly increasing along rows and strictly increasing down columns.
A Young diagram filled in according to such rules of ordering is called a 
{\it semistandard tableau}, in general.
When specifically filled in with 
integers between 1 and $d$, particularly when considered as a basis vector of
an irreducible representation of $\GL(d,\cmplx)$ or its subgroup, 
such a semistandard tableau is sometimes called a {\it Weyl tableau}.
For example, the representation of $\SU(2)$ labeled by $\ydiagram{3,1}$ is 3-dimensional
because there are three Weyl tableaux of this shape:
\setydiagrameq
$$
\ytableaushort{111,2},\ \ytableaushort{112,2},\ \ytableaushort{122,2},
$$
\setydiagramtext
while the representation of $\SU(2)$ labeled by $\ydiagram{3,2}$ is 2-dimensional because
there are two Weyl tableaux of this shape:
\setydiagrameq
$$
\ytableaushort{111,22},\ \ytableaushort{112,22}
$$
Such a representation of $\SU(d)$ labeled by a partition $\lambda$ is denoted
by $U_\lambda^{(d)}$.

We also consider tensor products of representations.
A tensor product of $\SU(d)$-modules is always assumed to be an ``inner" tensor product,
meaning that a product $M\otimes N$ between $G$-modules $M$ and $N$ is taken to be 
a $G$-module with the action of $g\in G$ on $M\otimes N$ given by 
$g.(m\otimes n)=g.m\otimes g.n$ for all $m\in M$ and $n\in N$.
On the other hand,
unless otherwise noted, a tensor product of symmetric group modules is assumed to be an
``outer" tensor product, meaning that a product $M\otimes N$ between $G$-module $M$ and 
$H$-module $N$
is taken to be a $G\times H$-module with the action of $(g,h)\in G\otimes H$ given by 
$(g,h).m\otimes n = g.m\otimes h.n$.
In other cases, we specify the resulting module and its action as needed.

\subsection{Hooks and conjugates}

Another advantage of Young diagrams is that various features of their shapes
correspond to properties of the associated representations,
and certain relationships between the shapes correspond to maps between the
representations.  Of particular interest are {\it hooks}, which are (roughly) 
hook-shaped diagrams,
e.g.
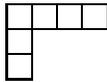
$$
\ydiagram{4,1,1}
$$
Hooks are identifiable by the length of the second row and column equaling one, 
whereas if the second row is greater than one then the partition is called {\it proper}
(while an {\it improper} partition is either a hook or trivial).
In terms of the bracket notation, a hook has the form $[n-r,1^r]$, for some $0<r<n-1$.
It proves useful to distinguish between {\it shallow hooks}, for which $r=1$ or $r=n-2$,
and {\it deep hooks}, for which $1<r<n-2$.
The correspondence between hooks and symmetric group representations is such that each 
irreducible $S_n$-module associated
with a hook of size $n$ is an exterior power of $V_{[n-1,1]}$; that is
$V_{[n-r,1^r]}=\bigwedge^rV_{[n-1,1]}$.  
The dimension of such a representations is given by the following formula \cite{MARIN2007742}:
\begin{lemma}
\label{lem:hookdimension}
We have
$$
\dim(V_{[n-r,1^r]})=\left(\begin{array}{c} n-1 \\ r \end{array}\right)=\frac{(n-1)!}{r!(n-1-r)!}.
$$
\end{lemma}
\noindent So for example a shallow hook yields dimension $n-1$, which for $n>4$ is the lowest possible
dimension of a nontrivial irreducible representation of $S_n$ \cite{fulton1991representation}.

Another noteworthy feature of a Young diagram's shape is the symmetry, or lack thereof,
with respect to the diagonal axis passing through the upper left corner.
A diagram that is flipped around this axis is called the {\it conjugate} of the original,
as for example:
$$
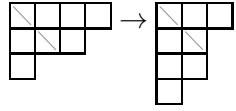
\ytableausetup{aligntableaux=top}
\ytableaushort{{\color{gray}\diagdown}\none\none\none,\none{\color{gray}\diagdown}\none,\none}*{4,3,1}\rightarrow
\ytableaushort{{\color{gray}\diagdown}\none\none,\none{\color{gray}\diagdown},\none\none,\none}*{3,2,2,1}
\ytableausetup{centertableaux}
$$
This operation may also be described as interchanging rows and columns.
Denoting the conjugate of a partition $\lambda$ by 
$\lambda'$, it also satisfies
$\lambda'_{\ i}=\max\{j\geq 1|\lambda_j\geq i\}$. 

We are interested in how the representations associated with conjugate partitions
relate to each other.  To begin with, they have the same dimension; that is, 
$\dim(\blambda')=\dim(\blambda)$, for all $\lambda$.
More specifically, there is an $S_{|\lambda|}$-module isomorphism
between $\boldsymbol{\lambda'}$ and the inner tensor product $\boldsymbol{\lambda}\otimes\boldsymbol{\epsilon}$, 
where $\bepsilon$ is the one-dimensional sign representation (which may also be
associated with the partition
$[1^{|\lambda|}]$) \cite[Exercise~4.4]{fulton1991representation}.
For this reason conjugate pairs and self-conjugate partitions, which we refer to collectively as
{\it conjugates}, are of special significance. 
We return to this point later.

\subsection{Young's orthogonal form}

The symmetric group action is conveniently formulated in terms of Young tableaux.  
With that formalism, and a given irreducible representation, there are many common ways 
to explicitly represent the symmetric group as matrices.  One popular choice is the representation afforded by the Specht module \cite{zbMATH03016521},
also known as Young's natural representation. Another is Young's seminormal form \cite{young1977collected}.  
However, the matrices representing the transpositions, as given by either of the above 
two representations,  are not Hermitian. 

For consistency with the quantum physical realization of the transposition, 
that being the exchange interaction, we require the matrix representation of each
transposition to be Hermitan and unitary.
For this purpose, a suitable choice is {\it Young's orthogonal form}, defined as follows.
Given a Young tableau $T$, let $r_T(k)$ and $c_T(k)$ respectively denote the row and column of $T$ 
containing the number $k$, and for $1\leq i,j\leq \ell$ let
$$
d_T(i,j)\equiv(c_T(j)-r_T(j))-(c_T(i)-r_T(i)),
$$
where $c_T(i)-r_T(i)$ is called the {\it content} of $i$ and $d_T(i,j)$ is called the
{\it axial distance} between $i$ and $j$.
Denoting Young's orthogonal form of the irreducible representation associated 
with partition $\lambda$ by $\rho_\lambda:S_\ell\rightarrow\GL(\blambda)$, we have
\begin{equation*}
\rho_\lambda((i\ i+1))T=\frac{1}{d_T(i,i+1)}T+\sqrt{1-\frac{1}{d_T(i,i+1)^2}}(i\ i+1)T,
\end{equation*}
where
$(i\ i+1)T$ denotes direct action on $T$ via exchange of the positions
of $i$ and $i+1$.
Note that $(i\ i+1)T$ results in an invalid tableau (only) when
$i$ and $i+1$ are in the same row or column, but then
$d_T(i,i+1)=\pm 1$ and thus the coefficient of the invalid term vanishes.
Since transpositions of the form $(i\ i+1)$ generate $S_\ell$,
the above definition is sufficient to
determine the representation of every permutation.

It is useful to define an inner product on the representation space, by 
$\langle S,T\rangle=\delta_{ST}$, 
for any tableau basis vectors $S$ and $T$ of the space.  
This inner product may then be
sesquilinearly extended to the full space. 
As the name suggests, Young's orthogonal form of each permutation is 
an orthogonal matrix in the tableau basis, 
which we verify below:
\begin{lemma}
\label{lem:orthogonalform}
Young's orthogonal form is orthogonal with respect
to the inner product defined on the associated representation space 
by $\langle S,T\rangle=\delta_{ST}$,
for any tableaux $S$ and $T$.
\end{lemma}
\begin{proof}
We have
\begin{eqnarray*}
\langle\rho_\lambda((i\ i+1))S,\rho_\lambda((i\ i+1))T\rangle &=&
\left(\frac{1}{d_T(i,i+1)d_S(i,i+1)}+\sqrt{\left(1-\frac{1}{d_S(i,i+1)^2}\right)\left(1-\frac{1}{d_T(i,i+1)^2}\right)}\right)\langle S,T\rangle\\
&&+\frac{1}{d_S(i,i+1)}\sqrt{1-\frac{1}{d_T(i,i+1)^2}}\langle S,(i\ i+1)T\rangle\\
&&+\frac{1}{d_T(i,i+1)}\sqrt{1-\frac{1}{d_S(i,i+1)^2}}\langle T,(i\ i+1)S\rangle. 
\end{eqnarray*}
Thus if $S=T$ then $\langle\rho_\lambda((i\ i+1))S,\rho_\lambda((i\ i+1))T\rangle=1$.  If $S=(i\ i+1)T$ or equivalently $T=(i\ i+1)S$, then 
$d_T(i,i+1)=-d_S(i,i+1)$ and all terms cancel.
Otherwise, $\langle\rho_\lambda((i\ i+1))S,\rho_\lambda((i\ i+1))T\rangle=0$.
Therefore $\langle\rho_\lambda((i\ i+1))S,\rho_\lambda((i\ i+1))T\rangle=\langle S,T\rangle$.  
This result extends to all permutations since again the nearest neighbor
transpositions generate the entire symmetric group $S_{|\lambda|}$.
\end{proof}
\noindent This inner product is therefore $S_{|\lambda|}$-invariant.
Since the permutations act as real orthogonal matrices, they are also unitary.
Since the transpositions are self-inverse, they are furthermore Hermitian.

Henceforth, if there is a need to consider the explicit action of $S_\ell$ on 
the irreducible module $\lambda$, we assume $S_\ell$ acts by Young's orthogonal
form.  Further, given $s\in S_\ell$ and $v\in\blambda$ we may denote 
$\rho_\lambda(s)v$ by $s.v$, since the action should be unambiguously 
understood.

\subsection{Lie algebras}

Although we are interested in elements of unitary groups, our focus is mostly on Lie 
algebras.  In their quantum mechanical applications, transpositions correspond to terms in 
Hamiltonians which can be considered elements of a Lie algebra, which in turn generate 
unitary operators by exponentiation.  The focus on Lie algebras is also convenient in taking 
advantage of Marin's results on the algebra of transpositions.

Therefore, starting with the group algebra 
$$
\cmplx S_{|\lambda|}\equiv\left\{\left.\sum_{s\in S_{|\lambda|}}a_ss\right|a_s\in\cmplx\right\}
$$
we construct a complex Lie algebra $(\cmplx S_{|\lambda|},[\cdot,\cdot])$ 
by using the multiplication and addition operators of $\cmplx S_{|\lambda|}$
to define a Lie bracket equal to the commutator $[\cdot,\cdot]$.
Of particular interest is the {\it algebra of transpositions} $\g_{|\lambda|}$,
that being the Lie subalgebra of $\cmplx S_{|\lambda|}$ generated
by the transpositions in $S_{|\lambda|}$.
Additional Lie subalgebras of $\cmplx S_{|\lambda|}$ with significant roles to play are the derived
algebra of $\g_{|\lambda|}$, denoted by $\g'_{|\lambda|}\equiv[\g_{|\lambda|},\g_{|\lambda|}]$, which consists of the span of the commutators of all elements of $\g_{|\lambda|}$,
and a compact real form of $\g'_{|\lambda|}$, denoted by $\Re(\g'_{|\lambda|})$,
which up to isomorphism is the real subalgebra of $\g'_{|\lambda|}$ with complexification equal to
$\g'_{|\lambda|}$ that generates a compact Lie group.  

We also consider irreducible representations of the Lie algebra $\cmplx S_{|\lambda|}$ by linearly extending $\rho_\lambda$ from $\rho_\lambda:S_{|\lambda|}\rightarrow\GL(\blambda)$ to $\rho_\lambda:\cmplx S_{|\lambda|}\rightarrow\gl(\blambda)$,
where $\gl(\blambda)$ is the Lie algebra of all complex matrices on the vector
space associated with $\lambda$, i.e. $\gl(\blambda)\cong\gl(\dim(\blambda),\cmplx)$.
The use of these particular representations for $\g_{|\lambda|}$ and $\g'_{|\lambda|}$
is justified by the following observation \cite[Proposition~2]{MARIN2007742}:
\begin{lemma}
The following statements are equivalent.
\begin{enumerate}[label=(\roman{*})]
\item The group representations $\rho_\lambda$ and $\rho_\mu$ of $S_m$ are equivalent.
\item The Lie algebra representations $\rho_\lambda$ and $\rho_\mu$ of
$\g_m$ are equivalent.
\item The Lie algebra representations $\rho_\lambda$ and $\rho_\mu$ of
$\g'_m$ are equivalent.
\end{enumerate}
\end{lemma}

We also have need of various subalgebras of $\gl(\blambda)$.
In particular, $\sla(\blambda)$ is the special linear algebra of traceless matrices on $\blambda$, isomorphic to $\sla(\dim(\blambda),\cmplx)$.
Its compact real form is the special unitary algebra $\su(\blambda)$, isomorphic to
$\su(\dim(\blambda),\cmplx)$.
Additional subalgebras of $\sla(\blambda)$, namely the special orthogonal algebra
$\so(\blambda)$, the symplectic algebra $\spa(\blambda)$, and the orthosymplectic algebra
$\osp(\blambda)$, can be defined in terms of a bilinear form on $\blambda$, to be discussed
in Sections~2.1--2.2.
Importantly, if $V$ is a subspace of vector space $W$, our conventions are such that $\su(V)$ denotes
the subalgebra of $\su(W)$ consisting of operators on $W$ that act as $\su(V)$ on $V$ and as 0 on
the orthogonal complement (with respect to the Euclidean inner product) of $V$ in $W$.

\subsection{Jucys--Murphy elements}
Certain sums of transpositions in $\g_n$ prove to be particularly useful.
\begin{Definition}
The elements $X_1,\ldots,X_n\in\cmplx S_n$
given by
$ X_k\equiv\sum_{i=1}^{k-1}(i\ k) $ 
are called {\it Jucys--Murphy elements}.
\end{Definition}
\noindent 
According to \cite{zbMATH02213760}, these special elements were originally 
noticed by Young, then independently rediscovered by Jucys in 1966 \cite{Jucys} and again
by Murphy in 1981 \cite{MURPHY1981287}.
These elements all commute with each other, and therefore
have joint eigenvectors.
We have in particular \cite[Equation~(12)]{JUCYS1974107},\cite[Equation~(3.8)]{MURPHY1981287}:
\begin{lemma}[Jucys--Murphy]
\label{lem:JucysMurphy}
Every Young tableau $T$ of size $n\geq k$ is an eigenvector of
Jucys--Murphy element $X_k$, with eigenvalue $x_k=c_T(k)-r_T(k)$ equal to the content of $k$.
\end{lemma}

From this follows the utility of the Jucys--Murphy elements, as their eigenvalues characterize the Young tableaux
\cite[Lemma~2.2]{MURPHY1981287}:
\begin{lemma}[Jucys--Murphy]
\label{lem:JMunique}
The $k$-tuple $(x_1,\ldots,x_k)$ of Jucys--Murphy eigenvalues
(called a {\it content vector}) of Jucys--Murphy elements $X_1,\ldots,X_k$
acting on a Young tableau $T$ of size $k$, uniquely determines $T$.
\end{lemma}
\noindent A proof can be sketched by recovering the tableau eigenvector from the eigenvalues $x_i$ as follows:
Given the subset $\{x_{i_0}\in\{x_i\}|x_{i_0}=0\}$ of the eigenvalues, by Lemma~\ref{lem:JucysMurphy} we can fill the cells of the main diagonal
(from upper left to lower right) of the tableau with the values $\{i_0\}$ 
in ascending order.  Labeling the remaining diagonals of the tableau
by $k$ such that the diagonal is offset from the first diagonal by $k$ cells (positive $k$ being above
and negative $k$ being below), 
we can similarly fill the cells of the $k$th diagonal with 
the eigenvalues $\{x_{i_k}\in\{x_i\}|x_{i_k}=k\}$.
Again, by definition of tableau, the values along
each diagonal are to be put in ascending order.  In this way the tableau is uniquely specified.

\begin{exmp}
The Jucys--Murphy element $X_k$ acting on a particular Young tableau results in multiplication
by an eigenvalue $x_k$:
$$
X_k\ytableaushort{1 3 {9} {10} {11}, 2 4 {12} {13}, 5 6 {14} {15},7 8 {16}}=x_k\ytableaushort{1 3 {9} {10} {11}, 2 4 {12} {13}, 5 6 {14} {15},7 8 {16}}
$$
Filling in the corresponding Young diagram with all such eigenvalues such that $x_k$
is in the cell containing $k$ in the original Young tableau results in cells along the
$j$-diagonal filled in with the value $j$:
\ytableausetup{boxsize=1.5em}
$$
\ytableaushort{{x_1} {x_3} {x_{9}} {x_{10}} {x_{11}}, {x_2} {x_4} {x_{12}} {x_{13}}, {x_5} {x_6} {x_{14}} {x_{15}},{x_7} {x_{8}} {x_{16}}}=\ytableaushort{0 1 2 3 4, {-1} 0 1 2, {-2} {-1} 0 1, {-3} {-2} {-1}}
$$
This process can be reversed:  Starting from the 16-tuple of eigenvalues:
$$
(0,-1,1,0,-2,-1,-3,-2,2,3,4,1,2,0,1,-1), 
$$
the number of eigenvalues equal to $j$
determines the length of the $j$-diagonal, which in turn uniquely determines the shape of
the Young diagram.  Then filling in the cells along each $j$-diagonal with the index
$k$ of each $x_k$ equal to $j$, in numerical order, uniquely determines the Young tableau.
For example, the Jucys--Murphy eigenvalues equal to 1 are  $x_3$, $x_{12}$, and $x_{15}$.
Therefore the 1-diagonal of the resulting tableau is $(3,12,15)$.
\end{exmp}
\setydiagramtext

As should be evident from the above discussion, when it is said
that the Jucys--Murphy elements uniquely determine a tableau, that includes the
shape of the tableau.
This important point is discussed in \cite{zbMATH02213760}, wherein the shape
of tableaux defines an equivalence class.
See in particular \cite[Proposition~5.3]{zbMATH02213760} and its proof.

Another useful feature follows immediately from the fact that 
Young's orthogonal form of each transposition is an orthogonal operator
by Lemma~\ref{lem:orthogonalform}.
Since each Jucys--Murphy element is a sum of such orthogonal operators,
it is also orthogonal.  
Thus by property of orthogonal operators, we have
\begin{lemma}
\label{lem:JMorthogonal}
Eigenvectors of a Jucys--Murphy element having different eigenvalues are orthogonal.
\end{lemma}
\noindent This has nontrivial applications to cases in which the
eigenvectors are linear combinations of
tableaux. 

In the present work we often make use of the Cartesian product of symmetric groups,
which motivates a slight variation of the definition considered above.
For the following, recall the identification
$$
\prod_{i=1}^N S_{m_i}=\prod_{i=1}^N S_{\{n_{i-1}+1,\ldots,n_i\}}
$$
where $n_i=\sum_{j=1}^im_j$.
\begin{Definition}
The element $X^{(i)}_k\in\cmplx S_{m_1}\times\cdots\times S_{m_N}$
given by $X^{(i)}_k=\sum_{j=1}^{k-1}(n_{i-1}+j\ n_{i-1}+k)$,
where $1\leq i\leq N$ and $1\leq k\leq m_i$,
is called a {\it local Jucys--Murphy element}.
\end{Definition}
\noindent As should be clear from this definition, given an irreducible representation
$\bigotimes_{i=1}^N\bmu^{(i)}$ of $\prod_{i=1}^N S_{m_i}$,
the local Jucys--Murphy element $X^{(i)}_k$ acts as the identity on every factor $\bmu^{(j)}$ for $j\neq i$,
and acts as an ordinary Jucys--Murphy element on $\bmu^{(i)}$.

\section{Marin's analysis of the algebra of transpositions}

In this section we summarize some of Marin's work on the algebra of transpositions. 
Primarily, we draw upon the first six sections of \cite{MARIN2007742}, which
are concerned with decomposing the algebra of transpositions into a direct sum of simple algebras.
Marin's strategy is to find a minimal set of irreducible representations whose direct sum is 
a faithful representation, while showing that the image of each such irreducible representation is a simple algebra.
Along the way Marin derives a number of results pertaining to irreducible representations that we can use.

To begin with, 
the algebra of transpositions is found to be isomorphic to 
the direct sum of an abelian Lie algebra and a semisimple Lie algebra,
which is to say it is reductive. 
More specifically by \cite[Proposition~1]{MARIN2007742}, we have
\begin{lemma}[Marin]
\label{lem:reductive}
The Lie algebra
$\g_{|\lambda|}$ is reductive, with its one-dimensional center generated by
the sum of all its transpositions. 
It follows that $\g_{|\lambda|}\cong\cmplx\oplus\g'_{|\lambda|}$,
and $\rho_\lambda(\g_{|\lambda|})=\rho_\lambda(\g'_{|\lambda|})$ when the
sum of transpositions acts as 0 on $\blambda$ 
and $\rho_\lambda(\g_{|\lambda|})=\cmplx\oplus\rho_\lambda(\g'_{|\lambda|})$ otherwise.
\end{lemma}
\noindent As we are ultimately interested in constructing unitary operators from the 
transpositions, and the 
center of $\g_{|\lambda|}$ amounts to nothing 
more than a ``global phase factor" in the parlance of quantum mechanics, we, like Marin,
focus on the structure of the derived algebra $\g'_{|\lambda|}$.

\subsection{Dependence between representations}

Returning to our quantum computing motivation for a moment, for a given operator we would 
like to find an element in $\g_{|\lambda|}$ that acts as that operator on a particular 
subspace of every irreducible representation associated with a partition in a set 
$\Lambda$.  A 
sufficient condition 
for this to be possible is as follows:
for each $\lambda\in\Lambda$ and each $X_\lambda\in\rho_\lambda(\g_{|\lambda|})$, there exists an element $x\in\g_{|\lambda|}$ such that
$(\bigoplus_{\lambda\in\Lambda}\rho_\lambda)(x)=\bigoplus_{\lambda\in\Lambda}X_\lambda$.
Equivalently, Kempe et al. describe representations as ``independent" if, for a given
operator on a given irreducible representation $\blambda$, there exists an element in $\g_{|\lambda|}$
that acts as that operator on $\blambda$ while acting as zero on all other irreducible
representations \cite{PhysRevA.63.042307}.

A consequence of Marin's work, summarized below, is that,
given representations $\rho_\lambda$ and $\rho_\mu$ of $\g'_n$,
we have $\ker\rho_\lambda = \ker\rho_\mu$
if $\lambda=\mu'$ or if $\lambda$ and $\mu$ are both hooks.
Therefore, such representations cannot be independent of each other in the
sense described above.
To evaluate whether or not the resulting actions are too constrained for our purposes,
we consider these cases in some detail.

We have already seen that $S_\ell$-modules associated with hooks are related to each
other, as exterior powers of $V_{[\ell-1,1]}$.
Marin has shown that the same is true of $\g'_\ell$-modules \cite[Lemma~12]{MARIN2007742}:
\begin{lemma}[Marin]
As a representation of $\g'_\ell$, we have $V_{[\ell-r,1^r]}\cong\wedge^rV_{[\ell-1,1]}$.
\end{lemma}
\noindent As Marin has observed of the exterior power of the $\g'_\ell$-module $V_{[\ell-1,1]}$ \cite[Section~5.3]{MARIN2007742}, 
accompanying the exterior power of the vector
space $V_{[\ell-1,1]}$ is a linear map $\Delta_{r}$ from $\rho_{[\ell-1,1]}(\g'_\ell)$ to
$\rho_{[\ell-r,1^r]}(\g'_\ell)$ defined such that for all $X\in\rho_{[\ell-1,1]}(\g'_\ell)$ and
$v_1,\ldots,v_r\in V_{[\ell-1,1]}$, 
$$
\Delta_r(X)v_1\wedge\cdots\wedge v_r=Xv_1\wedge\cdots\wedge v_r+v_1\wedge Xv_2\wedge\cdots\wedge v_r+v_1\wedge v_2\wedge\cdots\wedge Xv_r
$$
The map $\Delta_r$ defines a Lie algebra homomorphism from $\rho_{[\ell-1,1]}(\g'_\ell)$ to $\rho_{[\ell-r,1^r]}(\g'_\ell)$.  Since for all $x\in\g'_\ell$, 
if $\rho_{[\ell-1,1]}(x)$ is nonzero then it is nonscalar, by Lemma~\ref{lem:reductive}, $\Delta_r$ is also
bijective for $0<r<\ell-1$.  Thus we have \cite[Section~6.6]{MARIN2007742}:
\begin{lemma}[Marin]
\label{lem:marinisohooks}
Given hooks $\lambda,\mu\vdash\ell$, we have
$\rho_\lambda(\g'_\ell)\cong\rho_\mu(\g'_\ell)$.
\end{lemma}

The other important relationship in this context is the one between any two representations 
corresponding to a conjugate pair of partitions.  As discovered in 1900 by
Frobenius \cite[p.148--166]{frobenius1968gesammelte}, there is an  $S_{|\lambda|}$-module isomorphism between 
$\blambda'$ and the inner tensor product $\blambda\otimes\bepsilon$ \cite[Exercise~4.4]{fulton1991representation}. 
Originally formulated for any irreducible $S_{|\lambda|}$-module $\blambda$, 
such an isomorphism extends linearly to irreducible $\g_{|\lambda|}$-modules.
However, instead of considering a map between $\blambda'$ and $\blambda\otimes\bepsilon$,
it proves more convenient to consider the map induced directly between
$\blambda$ and $\blambda'$.

\subsection{The alternating intertwiner}

As discussed by Marin, the existence of the $S_{|\lambda|}$-intertwiner between 
$\blambda'$ and $\blambda\otimes\bepsilon$ implies a map, albeit not an 
$S_{|\lambda|}$-intertwiner,
between $\blambda$ and $\blambda'$:
\begin{Definition}
Given irreducible $\cmplx S_\ell$-module $\blambda$ and its conjugate $\blambda'$,
an {\it alternating intertwiner} $M_\lambda:\blambda\rightarrow\blambda'$ is a 
vector space isomorphism that further satisfies the condition
$$M_\lambda (s.v)=\epsilon(s)s.(M_\lambda v)$$ 
for all $s\in S_\ell$ and $v\in\blambda$,
where $\epsilon(s)$ is the sign of permutation $s$.
\end{Definition}
\noindent We call such a map an ``alternating intertwiner" because it is an intertwiner with 
respect to elements of the alternating group.
In Marin's notation this map is signified by $M$ when considering the case $\lambda=\lambda'$ \cite[Section~4.2]{MARIN2007742}, 
and by $P_{\lambda'}$ when considering the case $\lambda\neq\lambda'$ \cite[Section~5.1]{MARIN2007742}.
However Marin's definitions for $M$ and $P_{\lambda'}$ are consistent with each other, such
that each may be considered a special case of an alternating intertwiner as defined above.

Marin further points out a valuable relationship between an alternating intertwiner $M_\lambda$ and the action
of $\g_{|\lambda|}$ \cite[Sections~5.1--5.2]{MARIN2007742}.  In the following, ``$\tp$" denotes
the matrix transpose in the tableau basis:
\begin{lemma}[Marin]
\label{lem:isoMorphism}
For all $x\in\g_{|\lambda|}$, we have $M_\lambda\rho_\lambda(x)M_\lambda^{-1}=-\rho_{\lambda'}(x)^\tp$.  As a consequence, we have
$\rho_\lambda(\g'_{|\lambda|})\cong\rho_{\lambda'}(\g'_{|\lambda|})$.
\end{lemma}

It is helpful to give an explicit expression of $M_\lambda$ in the tableau basis \cite[Lemma~7]{MARIN2007742}: 
\begin{lemma}[Marin]
The map
$M_\lambda:\blambda\rightarrow\blambda'$ given by 
$$
M_\lambda T=w(T)T'
$$ 
for all tableaux $T\in\blambda$, where
$$
w(T)\equiv(-1)^{|\{i<j|r_T(i)>r_T(j)\}|},
$$
 is an alternating intertwiner.
\end{lemma}
\begin{proof}
Clearly $M_\lambda$ defined as above is a vector space isomorphism, since
the map $T\mapsto T'$ is a bijection between the tableau basis of $\blambda$
and the tableau basis of $\blambda'$.
It remains to show that $M_\lambda s.v=\epsilon(s)s.(M_\lambda v)$ for all $s\in S_{|\lambda|}$ and
$v\in\blambda$. 
Since $S_{|\lambda|}$ is generated by transpositions of neighboring elements,
and $\blambda$ admits as a basis its tableaux, it suffices to show that
$M_\lambda s.T=-s.(M_\lambda T)$ for any $s=(j\ j+1)$ such that $1\leq j<|\lambda|$
and any tableau $T\in\blambda$.
Let $T_s$ equal $T$ with its $j$ and $j+1$ entries transposed, 
let $d=d_T(j\ j+1)$, let $d'=d_{T'}(j\ j+1)$, and note that $d'=-d$.  
We also need the crucial observation that $w(T)=-w(T_s)$, which is the purpose of this
factor.
Then using the ``two-column proof" style for clarity, 
we have:
\begin{align*}
M_\lambda s.T &= M_\lambda\left(\frac{1}{d}T+\sqrt{1-\frac{1}{d^2}}T_s\right) & \text{by Young's orthogonal form},\\
&= \frac{1}{d}w(T)T'+\sqrt{1-\frac{1}{d^2}}w(T_s)T'_s & \text{by definition of }M_\lambda, \\
&= -w(T)\left(-\frac{1}{d}T'+\sqrt{1-\frac{1}{d^2}}T'_s\right) & \text{by }w(T_s)=-w(T),\\
&= -w(T)\left(\frac{1}{d'}T'+\sqrt{1-\frac{1}{(d')^2}}T'_s\right) & \text{by }d'=-d,\\
&= -w(T)s.T' & \text{by Young's orthogonal form}, \\
&= -s.(M_\lambda T) & \text{by definition of }M_\lambda. 
\end{align*}
\end{proof}

\subsection{Images of the algebra of transpositions}

\setydiagramtext
For the case of $\lambda=\lambda'$, the condition of Lemma~\ref{lem:isoMorphism} significantly
constrains 
$\rho_\lambda(\g'_{|\lambda|})$.
This condition can be written
$$\rho_\lambda(x)^\tp M_\lambda+M_\lambda\rho_\lambda(x)=0,$$
which is a common form of the condition for membership in the special orthogonal
or symplectic Lie algebra, depending on $M_\lambda$ \cite[Section~8.2]{fulton1991representation}.
It is convenient to define, as Marin does, a Lie algebra in terms of the above condition
for any $M_\lambda$ \cite[Definition~2]{MARIN2007742}:
\begin{Definition}
The Lie algebra $\osp(\blambda)$ is the subalgebra of $\sla(\blambda)$ such that for all
$X\in\sla(\blambda)$, $X$ is in $\osp(\blambda)$ if and only if $X^\tp M_\lambda+M_\lambda X=0$.
\end{Definition}
\noindent Elements satisfying the above definition can be verified to form a
subalgebra by observing that for all $X,Y\in\sla(\blambda)$, we have
$M_\lambda[X,Y]M_\lambda^{-1}=-[X,Y]^\tp$.
Then we have $\osp(\blambda)\cong\so(\dim(\blambda),\cmplx)$ or $\osp(\blambda)\cong\spa(\dim(\blambda),\cmplx)$ depending on whether $M_\lambda$ is symmetric or antisymmetric \cite{MARIN2007742}.  
In turn, $M_\lambda$ is found to be symmetric or antisymmetric depending on whether
4 divides $|\lambda|-b(\lambda)$, where $b(\lambda)=\max\{i|\lambda_i=\lambdap_i\}$ \cite[Section~4.2]{MARIN2007742}.
Therefore, the $\sla(\blambda)$ subalgebras $\so(\blambda)$ and $\spa(\blambda)$ can be defined \cite{MARIN2007742},
isomorphic to $\so(\dim(\blambda),\cmplx)$ and $\spa(\dim(\blambda),\cmplx)$ respectively,
such that 
$$
\osp(\blambda)=\left\{\begin{array}{ll} \so(\blambda),& 4|(|\lambda|-b(\lambda)),\\
\spa(\blambda),& \text{otherwise}.\end{array}\right.
$$
Returning to $\rho_\lambda(\g_{|\lambda|})$ we have \cite[Lemma~11]{MARIN2007742}: 
\begin{lemma}[Marin]
\label{lem:osp}
If $\lambda=\lambda'$, then we have $\rho_\lambda(\g_{|\lambda|})\subset\osp(\blambda)$.
\end{lemma}

More generally, we seek to understand the structure of the image of every 
irreducible representation $\rho_\lambda$ of $\g_{|\lambda|}$. 
A good starting place is \cite[Proposition~4]{MARIN2007742},
which states:
\begin{lemma}[Marin]
\label{lem:marindeq2}
If $\lambdap_1=2$, then we have $\rho_\lambda(\g'_{|\lambda|})=\sla(\blambda)$.
\end{lemma}
\noindent But there are two such partitions that are also self-conjugate,
namely $\ydiagram{2,1}$ and $\ydiagram{2,2}$.  
In those cases we have $\rho_\lambda(\g_{|\lambda|})=\sla(\blambda)$,
yet as implied above we have $\rho_\lambda(\g_{|\lambda|})\subset\osp(\blambda)$.
To verify the consistency between these statements, 
first note that $\dim(V_\subydiagramtext{2,1})=\dim(V_\subydiagramtext{2,2})=2$.
Calculation of $M_\lambda$ in the tableau basis then yields
$$
M_\subydiagrameq{2,1}=M_\subydiagrameq{2,2}=\left(\begin{array}{cc}0&1\\-1&0\end{array}\right).
$$
The resulting condition $\rho_\lambda(x)^\tp M_\lambda+M_\lambda\rho_\lambda(x)=0$
is satisfied by each of the Pauli matrices,
$$
\left(\begin{array}{cc}0&1\\1&0\end{array}\right),\ \left(\begin{array}{cc}0&i\\-i&0\end{array}\right),\ \left(\begin{array}{cc}1&0\\0&-1\end{array}\right),
$$
\setydiagramtext
which constitute a basis for $\sla(2,\cmplx)$, and therefore of the equivalent
algebras $\sla(V_\subydiagramtext{2,1})$ and $\sla(V_\subydiagramtext{2,2})$.
Thus we must have that $\sla(\blambda)\subset\osp(\blambda)$ in this case.
Yet, for such a $\lambda$, we also have that $\osp(\blambda)=\spa(\blambda)$,
which is a subalgebra of $\sla(\blambda)$.
These statements are consistent because
$$
\spa(\blambda)\cong\spa(2,\cmplx)\cong\sla(2,\cmplx)\cong\sla(\blambda),
$$
from which we conclude that $\rho_\lambda(\g_{|\lambda|})=\sla(\blambda)=\osp(\blambda)$ for $\lambda\in\{\ydiagram{2,1},\ydiagram{2,2}\}$.

Lemma~\ref{lem:marindeq2} has considerably more implications for more general cases. 
Since $\rho_{[\ell-1,1]}(\g'_\ell)=\sla(V_{[\ell-1,1]})$, it follows
by Lemma~\ref{lem:marinisohooks} that $\rho_\lambda(\g'_\ell)\cong\sla(V_{[\ell-1,1]})$ 
for every hook $\lambda\vdash\ell$.
Further, for all proper, non-self-conjugate $\lambda$ with two rows ($\lambdap_1=2$), it follows that
$\rho_\lambda(\g'_{|\lambda|})=\sla(\blambda)$ and thus, by Lemma~\ref{lem:isoMorphism}, 
all proper, non-self-conjugate $\lambda$ with two columns ($\lambda_1=2$)
also satisfy $\rho_\lambda(\g'_{|\lambda|})=\sla(\blambda)$.
In particular, this implies that $\rho_\lambda(\g'_{|\lambda|})=\sla(\blambda)$ for all proper, non-self-conjugate $\lambda$ 
such that $|\lambda|<7$. 

For the case of $|\lambda|\geq 7$, Marin proved the following \cite[Proposition~7]{MARIN2007742}:
\begin{lemma}
Given a proper partition $\lambda\vdash\ell\geq 7$, we have
$$
\rho_\lambda(\g'_\ell)=\left\{\begin{array}{ll} \sla(\blambda),& \lambda\neq\lambda'\\
\osp(\blambda),& \lambda=\lambda'.\end{array}\right.
$$
\end{lemma}
\noindent But note this and the foregoing analysis exclude the case $\lambda=\ydiagram{3,2,1}$.  
To address this case, Marin has shown that \cite[Proposition~13]{Marin2003} 
$$
\g'_6\cong\sla(V_\subydiagram{5,1})\oplus\sla(V_\subydiagram{3,3})\oplus\sla(V_\subydiagram{4,2})\oplus\osp\left(V_\subydiagram{3,2,1}\right).
$$
Further, he established that the homomorphism
$$
(\rho_{\subydiagramtext{5,1}}\oplus\rho_{\subydiagramtext{3,3}}\oplus\rho_{\subydiagramtext{4,2}}\oplus\rho_{\subydiagramtext{3,2,1}})(\g'_6)\rightarrow\sla(V_\subydiagram{5,1})\oplus\sla(V_\subydiagram{3,3})\oplus\sla(V_\subydiagram{4,2})\oplus\osp\left(V_\subydiagram{3,2,1}\right)
$$
is injective \cite[Section~5]{MARIN2007742}.
From these two equations, and from the previous results for cases when $\rho_\lambda(\g'_{|\lambda|})=\sla(\blambda)$, and from \cite[Lemma~11]{MARIN2007742} which implies $\rho_{\subydiagramtext{3,2,1}}(\g'_6)\subset\osp\left(V_\subydiagram{3,2,1}\right)$, we conclude that 
\begin{lemma}[Marin]
The image of $\rho_{\subydiagramtext{3,2,1}}$ is
$
\rho_{\subydiagramtext{3,2,1}}(\g'_6)=\osp\left(V_\subydiagram{3,2,1}\right).
$
\end{lemma}

Consolidating Marin's results for the image of each irreducible representation we 
finally arrive at: 
\begin{lemma}[Marin]
\label{lem:marin}
The images of the irreducible representations associated with $\lambda\vdash\ell$ of 
$\g'_\ell$ are as follows:
\begin{equation*}
\rho_\lambda(\g'_\ell)\cong\left\{\begin{array}{ll}
\sla(V_{[\ell-1,1]}), & \lambda=[\ell-r,1^r],0< r<\ell-1,\\
\sla(\blambda), & \lambda\neq\lambda'\text{ proper},\\
\osp(\blambda), & \lambda=\lambda'\text{ proper},
\end{array}\right.
\end{equation*}
where this Lie algebra isomorphism is an equality except when 
$\lambda=[\ell-r,1^r]$ for $r>1$.
\end{lemma}

\subsection{Semisimple decomposition of the algebra of transpositions}

Marin proceeds to use results summarized in the previous sections to explicitly
decompose $\g'_{|\lambda|}$ as a direct sum of simple algebras
(Theorem~\ref{thm:marin} below).  
In the course of this proof,
the following useful fact is derived \cite[Section~6.6]{MARIN2007742}:
\begin{lemma}[Marin]
\label{lem:iffimagesiso}
Given distinct partitions $\lambda,\mu\vdash\ell$ such that $\dim(\blambda)>1$ and
$\dim(\bmu)>1$, 
the map given for all $x\in\g'_\ell$ by $\rho_\lambda(x)\mapsto\rho_\mu(x)$
is a Lie algebra isomorphism
if and only if
$\mu=\lambda'$ or $\lambda$ and $\mu$ are both hooks.
\end{lemma}

Below we state one of the main theorems of \cite{MARIN2007742}, that being an isomorphism 
between $\g’_\ell$ and the direct sum of
its simple ideals, in the form of a commuting diagram which indicates how the isomorphism was constructed.  In a series of commuting diagrams, Marin illustrated how
the image of each irreducible representation $\rho_\lambda(\g'_\ell)$ 
relates to the image of the sum over all partitions of the associated
irreducible representations, ultimately demonstrating how a subset of such irreducible
representations uniquely determines the image of the full sum.
The purpose in doing so, accepting the hypothesis that $\g'_\ell$ is isomorphic to
a direct sum of such images, was to determine 
which representations need constitute that sum.
Here we compose those commuting diagrams together \cite[Theorem~A]{MARIN2007742}:
\begin{theorem}[Marin]
\label{thm:marin}
The following diagram of Lie algebra homomorphisms commutes: 
\begin{displaymath}
\xymatrix
{
  && \bigoplus\limits_{\lambda\vdash \ell}\gl(\blambda) \\
  && \\
  && \bigoplus\limits_{\lambda\vdash \ell}\sla(\blambda)\ar[uu]_-\iota \\
  && \\
  && \bigoplus\limits_{\substack{\lambda\vdash \ell\\\lambda<\lambda'}}\sla(\blambda)\oplus\bigoplus\limits_{\substack{\lambda\vdash \ell\\\lambda=\lambda'}}\sla(\blambda)\ar[uu]_-{\bigoplus\limits_{\substack{\lambda\vdash \ell\\\lambda<\lambda'}}(\Id_{\sla(\blambda)}\oplus\tilde{M}_\lambda)\oplus\bigoplus\limits_{\substack{\lambda\vdash \ell\\\lambda=\lambda'}}\Id_{\sla(\blambda)}} \\
  && \\
  && \bigoplus\limits_{\substack{\lambda\vdash \ell\\\lambda<\lambda'}}\sla(\blambda)\oplus\bigoplus\limits_{\substack{\lambda\vdash \ell\\\lambda=\lambda'}}\osp(\blambda)\ar[uu]_-\iota\\
  && \\
\g'_\ell\ar@<-.3ex>[rr]_-{\bigoplus\limits_{\lambda\in\Lambda}\rho_\lambda}\ar[uuuuuuuurr]^{\bigoplus\limits_{\lambda\vdash \ell}\rho_\lambda}  && \sla(V_{[\ell-1,1]})\oplus\bigoplus\limits_{\substack{\lambda\vdash \ell\\\lambda<\lambda'\\\proper}}\sla(\blambda)\oplus\bigoplus\limits_{\substack{\lambda\vdash \ell\\\lambda=\lambda'\\\proper}}\osp(\blambda)\ar[uu]_-{\bigoplus\limits_{r=1}^{\ell-2}\Delta_r\oplus\bigoplus\limits_{\substack{\lambda\vdash \ell\\\lambda<\lambda'\\\proper}}\Id_{\sla(\blambda)}\oplus\bigoplus\limits_{\substack{\lambda\vdash \ell\\\lambda=\lambda'\\\proper}}\Id_{\osp(\blambda)}}\ar@<-.3ex>[ll]_-{\left({\bigoplus\limits_{\lambda\in\Lambda}\rho_\lambda}\right)^{-1}}
}
\end{displaymath}
where $\Lambda\equiv\{[\ell-1,1]\}\cup\{\lambda\vdash \ell|\lambda_2>1,\lambda\leq\lambda'\}$,  $\iota$ denotes inclusion, $\tilde{M}_\lambda$ is the automorphism
of $\gl(\blambda)$ defined by $X\mapsto -M_\lambda X^\tp M_\lambda^{-1}$,
and
\begin{equation*}
\Delta_r:\sla(V_{[\ell-1,1]})\rightarrow \sla(V_{[\ell-r,1^r]})
\end{equation*}
is defined such that
\begin{equation*}
\Delta_r(X)(v_1\wedge\cdots\wedge v_r)=Xv_1\wedge\cdots\wedge v_r+v_1\wedge Xv_2\wedge\cdots\wedge v_r+v_1\wedge v_2\wedge\cdots\wedge Xv_r
\end{equation*}
for all $v_1,\ldots,v_r\in V_{[\ell-1,1]}$. 
\end{theorem}

To review the maps depicted in the diagram, starting from the bottom right,
the first map determines all hook representations from wedge products of $[\ell-1,1]$ 
via the isomorphisms $\Delta_r$, the second map is an inclusion map indicating that 
self-conjugate representations are confined to the subalgebra $\osp(\blambda)$ of 
$\sla(\blambda)$, the third map determines the conjugate representation for each 
non-self-conjugate proper partition via the isomorphism induced by the alternating 
intertwiner $M_\lambda$, and the fourth map is another inclusion map indicating that 
representations of $\g'_\ell$ necessarily map to traceless elements of $\gl(\blambda)$.
Also note that, in the isomorphism
$$
\bigoplus_{\lambda\in\Lambda}\rho_\lambda:\g_\ell'\rightarrow \sla(V_{[\ell-1,1]})\oplus\bigoplus\limits_{\substack{\lambda\vdash \ell\\\lambda<\lambda'\\\proper}}\sla(\blambda)\oplus\bigoplus\limits_{\substack{\lambda\vdash \ell\\\lambda=\lambda'\\\proper}}\osp(\blambda),
$$
by Lemma~\ref{lem:iffimagesiso} the hook in $[\ell-1,1]$ in $\Lambda$ can be replaced by any other hook
$[\ell-r,1^r]$ of the same size $\ell$, and also any partition can
be replaced by its conjugate.  Thus
$$
\left(\bigoplus_{\lambda\in\Lambda}\rho_\lambda\right)(\g_\ell')\cong\sla(V_{[\ell-1,1]})\oplus\bigoplus\limits_{\substack{\lambda\vdash \ell\\\lambda<\lambda'\\\proper}}\sla(\blambda)\oplus\bigoplus\limits_{\substack{\lambda\vdash \ell\\\lambda=\lambda'\\\proper}}\osp(\blambda),
$$
where $\Lambda$ can be generalized such that it contains any one hook, one of each
conjugate pair of proper partitions, and all proper self-conjugate partitions.

The above isomorphism implies an important corollary regarding the interdependency of
the representations.  Using Lemma~\ref{lem:marin} the isomorphism can be written
$$
\left(\bigoplus_{\lambda\in\Lambda}\rho_\lambda\right)(\g_\ell')\cong\bigoplus_{\lambda\in\Lambda}\rho_\lambda(\g_\ell').
$$
As this remains true when we remove partitions from both sides, we obtain a slightly
more general statement:
\begin{corollary}
\label{cor:marin}
Given $\Lambda\subset\{\lambda\vdash\ell\}$, we have
$\left(\bigoplus_{\lambda\in\Lambda}\rho_\lambda\right)(\g_\ell')\cong\bigoplus_{\lambda\in\Lambda}(\rho_\lambda(\g_\ell'))$ if and only if $\Lambda$ is free of nontrivial 
conjugate pairs and multiple hooks.
\end{corollary}
\noindent Trivial representations are not a concern since $\rho_{[\ell]}(\g'_\ell)=\rho_{[1^\ell]}(\g'_\ell)=0$.
By further application of Lemma~\ref{lem:iffimagesiso}, conjugates and hooks can be added to the sums on both sides
of the isomorphism of Corollary~\ref{cor:marin}, as follows.
\begin{corollary}
\label{cor:marin2}
Given a set $\cP$ of any partitions of $n$, we have
$$
\left(\bigoplus\limits_{\substack{\text{hooks}\\\nu\in\cP}}\rho_\nu\right)(\g_n)\oplus\bigoplus\limits_{\substack{\text{conjugate pairs}\\\nu\neq\nu'\in\cP}}(\rho_\nu\oplus\rho_{\nu'})(\g_n)\oplus\bigoplus\limits_{\substack{\text{rest of}\\\nu\in\cP}}\rho_\nu(\g_n)\cong\left(\bigoplus_{\nu\in\cP}\rho_\nu\right)(\g_n).
$$
\end{corollary}

\section{Relevant intertwiners}

As representations of the symmetric group 
and unitary group are critical to this work, maps that are equivariant 
with respect to such actions are also of primary concern.   We use the terms 
{\it $G$-intertwiner} or {\it $G$-module homomorphism} interchangeably for linear maps 
that are equivariant with respect to action of the group $G$, and sometimes indicate 
such a map with $G$ above the arrow: ``$\overset{G}{\rightarrow}$", or isomorphism symbol: ``$\overset{G}{\cong}$".  The space of such maps 
between vector spaces $V$ and $W$ is itself a vector space, denoted $\Hom_G(V,W)$.  
Typically we are interested in equivariance with respect to a Cartesian product of 
symmetric groups, irreducibly represented by the first argument of $\Hom$, and when this 
equivariance 
is understood the group subscript is suppressed.  The material for this section is 
primarily drawn from Fulton and Harris \cite{fulton1991representation}, Goodman and Wallach \cite{goodman1998representations}, and Harrow
\cite{10.5555/1195346}.

\subsection{Schur--Weyl duality}

A fundamental intertwiner which makes exchange-only encoding possible is called Schur--Weyl duality.  For our purposes, it may be stated as an $S_\ell\times\SU(d)$-module isomorphism between the Hilbert space $(\cmplx^d)^{\otimes\ell}$ and a certain sum of 
irreducible $S_\ell\times\SU(d)$-representations. 
More commonly, and more generally, it is expressed in terms of an arbitrary vector
space $V$ as a $S_\ell\times\GL(V)$-module isomorphism, as follows \cite[Exercise~6.30]{fulton1991representation}:
$$
V^{\otimes\ell}\overset{S_\ell\times\GL(V)}{\cong}\bigoplus_{\lambda\vdash\ell,\lambdap_1\leq d}V_\lambda\otimes U^{(d)}_\lambda.
$$
Here, $U^{(d)}_\lambda$ is the irreducible representation of $GL(V)$, 
or ``Weyl module", canonically 
associated with Young diagram $\lambda$ \cite[Section~6.1]{fulton1991representation}, $\GL(V)$ acts on $V^{\otimes\ell}$ by
simultaneous multiplication on every vector $v_i\in V$ of the product $v_1\otimes\cdots\otimes v_\ell\in V^{\otimes\ell}$, and $S_\ell$ acts on $V^{\otimes\ell}$ by permuting the vectors in the product
$v_1\otimes\cdots\otimes v_\ell$. 

In this regard, it may be clarifying to remark that every finite dimensional 
irreducible representation of $\GL(d,\cmplx)$ restricts to an irreducible representation of 
$\SU(d)$.  
In fact, referring to Fig.~(\ref{fig:GLsubgroups}), the restriction of an irreducible
representation of any of the depicted groups to any of the indicated subgroups
(connected by a line segment in the figure) is also irreducible. 
An irreducible representation of $\GL(d,\cmplx)$ 
restricted to $\SL(d,\cmplx)$ remains irreducible because each element of the former is a 
scalar multiple of an element of the latter \cite[p.222--223]{fulton1991representation}.
By a similar argument, an irreducible representation of $\U(d)$ restricted to $\SU(d)$
likewise remains irreducible.
Furthermore, an irreducible representation of $\GL(d,\cmplx)$ or 
$\SL(d,\cmplx)$ restricted to $\U(d)$ or $\SU(d)$, respectively, remains irreducible by an argument 
due to Weyl known as the ``unitary trick" \cite[Section~9.3]{fulton1991representation}.  In brief, 
in such a case when complexification of the Lie algebra of the subgroup generates the 
original group, it can be deduced that reducibility of a representation restricted to the 
former would imply reducibility of the representation of the latter.

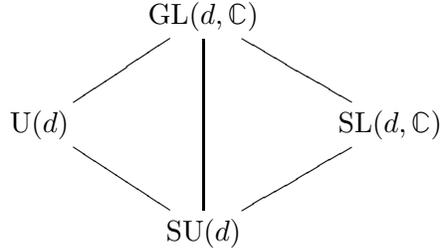
\begin{figure}
\begin{displaymath}
\xymatrix{
& \GL(d,\cmplx)\ar@{-}[dl]\ar@{-}[dd]\ar@{-}[dr] & \\
\U(d)\ar@{-}[dr] & & \SL(d,\cmplx)\ar@{-}[dl] \\
& \SU(d)
}
\end{displaymath}
\caption{Subgroups of the general linear group $\GL(d)$ that share irreducible representations.}  
\label{fig:GLsubgroups}
\end{figure}

So basis vectors of irreducible representations of both $\GL(d,\cmplx)$ and $\SU(d)$
can be taken to be the Weyl tableaux described in Section~2.1.1.  As an example
of the use of both Weyl tableaux and Young tableaux in the context of Schur--Weyl duality,
consider
\setydiagrameq
\begin{equation*}
(\cmplx^2)^{\otimes 3}\overset{S_3\times SU(2)}{\cong}(V_\subydiagrameq{3}\otimes U^{(2)}_\subydiagrameq{3})\oplus(V_\subydiagrameq{2,1}\otimes U^{(2)}_\subydiagrameq{2,1}).
\end{equation*}
\setydiagramtext
where
\setydiagrameq
\begin{eqnarray*}
V_\subydiagrameq{3}\otimes U^{(2)}_\subydiagrameq{3} &\cong&\text{span}\left\{\ytableaushort{1 2 3}\otimes\ytableaushort{1 1 1},\ytableaushort{1 2 3}\otimes\ytableaushort{1 1 2},\ytableaushort{1 2 3}\otimes\ytableaushort{1 2 2},\ytableaushort{1 2 3}\otimes\ytableaushort{2 2 2}\right\},\\
V_\subydiagrameq{2,1}\otimes U^{(2)}_\subydiagrameq{2,1}&\cong&\text{span}\left\{\ytableaushort{1 3,2}\otimes\ytableaushort{1 1, 2},\ytableaushort{1 2,3}\otimes\ytableaushort{1 1,2},\ytableaushort{1 3,2}\otimes\ytableaushort{1 2,2},\ytableaushort{1 2,3}\otimes\ytableaushort{1 2,2}\right\}.
\end{eqnarray*}.
\setydiagramtext
For more details on this particular example, see \cite{PhysRevA.99.042331}.

\subsection{Inducing from a product of symmetric group representations}

Another relevant intertwiner is that between a tensor product of symmetric group
representations, representing the Cartesian product of their symmetric groups,
and the induced representation of the symmetric group containing the factor groups.
For the tensor product of two irreducible representations, this induced representation
is given by the formula \cite[Equation~4.41]{fulton1991representation}
$$
\Ind_{S_{|\lambda|}\times S_{|\mu|}}^{S_{|\lambda|+|\mu|}}\blambda\otimes\bmu=\bigoplus_{\nu\vdash|\lambda|+|\mu|}c_{\lambda\mu}^\nu\bnu,
$$
where the symbol $c_{\lambda\mu}^\nu$ denotes the Littlewood--Richardson coefficient,
which here gives the multiplicity of each irreducible representation $\bnu$ in
the representation induced from $\blambda\otimes\bmu$.  
For the tensor product $\blambda\otimes\bmu$ 
above, and henceforth for products of symmetric group
representations, it is assumed that the group
acting on the resulting module is the Cartesian product of the groups acting on
each factor module.
The above formula may be 
considered a definition of the Littlewood--Richardson coefficient, although it has 
many other equally fundamental applications, some of which will be demonstrated shortly.
We also define it in Section~4.1 in terms of how it is calculated.

As Fulton and Harris remark \cite[p.58]{fulton1991representation}, the above operation of taking a tensor product of 
representations and inducing a new representation can be generalized for a tensor
product of arbitrarily many representations.  Since inducing representations is
transitive \cite[Exercise~3.16]{fulton1991representation}, the above formula immediately generalizes by iteration:
\begin{lemma}
Given a family of partitions $\{\mu^{(i)}\vdash m_i\}_{i=1}^k$, and $n=\sum_{i=1}^km_i$, we have
$$\Ind_{S_{m_1}\times\cdots\times S_{m_k}}^{S_n}{\bmu^{(1)}}\otimes\cdots\otimes {\bmu^{(k)}}=\bigoplus_{\nu\vdash n} c_{\mu^{(1)}\mu^{(2)}}^{\nu^{(2)}}c_{\nu^{(2)}\mu^{(3)}}^{\nu^{(3)}}\cdots c_{\nu^{(k-2)}\mu^{(k-1)}}^{\nu^{(k-1)}}c_{\nu^{(k-1)}\mu^{(k)}}^{\nu}\bnu,$$
where the Einstein summation convention is assumed, so there is an implicit sum
over each repeated index. 
\end{lemma}
The above operation of tensor multiplication 
and representation induction is both commutative and associative \cite[p.58]{fulton1991representation}, 
in the sense that the direct sum over $\nu$ of vector spaces $\bnu$, including their
multiplicities, is 
independent of the ordering of the tensor
product and of the particular Littlewood--Richardson coefficients that result.
Therefore we can unambiguously define the coefficient
$$
c_{\mu^{(1)}\cdots\mu^{(k)}}^\nu\equiv c_{\mu^{(1)}\mu^{(2)}}^{\nu^{(2)}}c_{\nu^{(2)}\mu^{(3)}}^{\nu^{(3)}}\cdots c_{\nu^{(k-2)}\mu^{(k-1)}}^{\nu^{(k-1)}}c_{\nu^{(k-1)}\mu^{(k)}}^{\nu}
$$
with the understanding that its value is invariant with respect to the order of its lower indices.

\subsection{Restricting to a product of symmetric group representations}

The Littlewood--Richardson coefficient also gives the multiplicity of the product 
$\blambda\otimes\bmu$ in the restriction of $\bnu$ from $S_{|\lambda|+|\mu|}$ to 
$S_{|\lambda|}\times S_{|\mu|}$ \cite[Exercise~4.43]{fulton1991representation}.  By Schur's lemma this can be 
expressed as \cite[p.16]{fulton1991representation}
$$
\dim(\Hom(\Res^{S_{|\lambda|+|\mu|}}_{S_{|\lambda|}\times S_{|\mu|}}\bnu,\blambda\otimes\bmu))=c_{\lambda\mu}^\nu.
$$
As with induced representations this formula can be generalized for products of
arbitrarily many representations.  In fact the latter may be derived from the former
by Frobenius reciprocity:  
\begin{lemma}
\label{lem:LRres}
The multiplicity of $\bmu^{(1)}\otimes\cdots\otimes\bmu^{(k)}$ in the restriction of
$\bnu$ from $S_n$ to
$S_{m_1}\times\cdots\times S_{m_k}$ is given by $c_{\mu^{(1)}\cdots\mu^{(k)}}^\nu$.
\end{lemma}
\begin{proof}
Using Frobenius reciprocity in the form $\Hom_G(V,\Ind_H^G W)\cong\Hom_H(V,W)$, we have
\begin{eqnarray*}
\Hom_{S_{m_1}\times\cdots\times S_{m_k}}(\bmu^{(1)}\otimes\cdots\otimes\bmu^{(k)},\bnu) &\cong& \Hom_{S_{m_1}\times\cdots S_{m_k}}(\bnu,\bmu^{(1)}\otimes\cdots\otimes\bmu^{(k)}),\\
&\cong& \Hom_{S_n}(\bnu,\Ind_{S_{m_1}\times\cdots\times S_{m_k}}^{S_n}\bmu^{(1)}\otimes\cdots\otimes\bmu^{(k)}),\\
&\cong& \Hom_{S_n}\left(\bnu,\bigoplus_\nu c_{\mu^{(1)}\cdots\mu^{(k)}}^\nu\bnu\right).
\end{eqnarray*}
Hence we have
$
\dim(\Hom_{S_{m_1}\times\cdots\times S_{m_k}}(\bmu^{(1)}\otimes\cdots\otimes\bmu^{(k)},\bnu))=c_{\mu^{(1)}\cdots\mu^{(k)}}^\nu.
$
\end{proof}

\subsection{Isotypical decomposition} 

Lemma~\ref{lem:LRres} implies that the maximum number of pairwise orthogonal subspaces of
$\bnu$ that are $S_{m_1}\times\cdots\times S_{m_N}$-module isomorphic to 
$\bigboxtimes_{i=1}^N\bmu^{(i)}$ is $c_{\mu^{(1)}\cdots\mu^{(N)}}^\nu$.
We call the direct sum
of a maximal number of orthogonal such subspaces a $\bigboxtimes_{i=1}^N\bmu^{(i)}$-isotypical subspace,
denoted by $V_\mus^\nu$.  
Alternatively, $V_\mus^\nu$ may be defined as the span of all subspaces of $\bnu$
that are $S_{m_1}\times\cdots\times S_{m_N}$-module isomorphic to
$\bigboxtimes_{i=1}^N\bmu^{(i)}$, which by Lemma~\ref{lem:LRres} must have a total dimension
of $\prod_{i=1}^N\dim(\bmu^{(i)})c_{\mu^{(1)}\cdots\mu^{(N)}}^\nu$.
Of course, many such isotypical subspaces 
typically comprise $\bnu$.  Since by Maschke's Theorem the restriction of $\bnu$ to $S_{m_1}\times\cdots\times S_{m_N}$ can, like any finite group representation, be decomposed into irreducible representations,
we have $\bnu=\bigoplus_{\mu^{(i)}\vdash m_i}V_\mus^\nu$.
The orthogonality of these isotypical subspaces comes as a consequence of Schur's lemma, and can be 
stated more precisely as follows \cite[Lemma~I.5.3]{panyushev}:
\begin{lemma}
\label{lem:isotypicalorthogonality}
Given isotypical subspaces $V_\lambdas^\nu$ and $V_\mus^\nu$ of $\bnu$,
if $\bigotimes_{i=1}^N\blambda^{(i)}$ and $\bigotimes_{i=1}^N\bmu^{(i)}$
are inequivalent representations of $S_{m_1}\times\cdots\times S_{m_N}$
then $V_\lambdas^\nu$ and $V_\mus^\nu$ are orthogonal with respect to any
$S_{m_1}\times\cdots\times S_{m_N}$-invariant inner product on $\bnu$.
\end{lemma}

Note that there exists a canonical map $\chi_\mus^\nu$ onto each isotypical subspace $V_\mus^\nu$ from $\bigboxtimes_{i=1}^N\bmu^{(i)}\otimes\Hom(\bigboxtimes_{i=1}^N\bmu^{(i)},\bnu)$, given by 
$\chi_\mus^\nu(\psi\otimes\iota)=\iota(\psi)$.
Considering $\bigboxtimes_{i=1}^N\bmu^{(i)}\otimes\Hom(\bigboxtimes_{i=1}^N\bmu^{(i)},\bnu)$ as an 
$S_{m_1}\times\cdots\times S_{m_N}$-module such that
$s.(\psi\otimes\iota)=(s.\psi)\otimes\iota$, it follows that \cite[Proposition~3.1.6]{goodman1998representations}
\begin{lemma}
\label{lem:canonicaliso}
The linear map 
\begin{eqnarray*}
\chi_\mus^\nu:\bigboxtimes_{i=1}^N\bmu^{(i)}\otimes\Hom(\bigboxtimes_{i=1}^N\bmu^{(i)},\bnu) &\rightarrow& V_\mus^\nu\\
\psi\otimes\iota &\mapsto& \iota(\psi)
\end{eqnarray*}
gives an $S_{m_1}\times\cdots\times S_{m_N}$-module isomorphism. 
\end{lemma}
\noindent For the purpose of studying operators on the irreducible components of $\bnu$,
it is often convenient to use the above isomorphism to consider the
corresponding operators on
$\bigboxtimes_{i=1}^N\bmu^{(i)}\otimes\Hom(\bigboxtimes_{i=1}^N\bmu^{(i)},\bnu)$.

It is useful to note that $\chi_\mus^\nu$ is an orthogonal transformation with respect 
to $S_{m_1}\times\cdots\times S_{m_N}$-invariant inner products naturally defined on its domain and range, as 
follows.  We already defined the Euclidean inner product applicable to
$\bigotimes_{i=1}^N\bmu^{(i)}$ and $\bnu$ in terms of tableau bases.
An inner product is then naturally induced on $\Hom(\bigotimes_{i=1}^N\bmu^{(i)},\bnu)$, up to scaling, by the condition that intertwiners are
orthogonal if and only if their images are orthogonal.
The scaling is then fixed by the further condition that 
$\langle\iota,\iota\rangle=1$ if and only if 
$\iota:\bigotimes_{i=1}^N\bmu^{(i)}\rightarrow\bnu$  is an isometry.
An inner product is then defined on $\bigotimes_{i=1}^N\bmu^{(i)}\otimes\Hom(\bigotimes_{i=1}^N\bmu^{(i)},\bnu)$ by factoring: 
$\langle\psi\otimes\iota,\psi'\otimes\iota'\rangle=\langle\psi,\psi'\rangle\langle\iota,\iota'\rangle$.
It follows that $\langle\psi\otimes\iota,\psi'\otimes\iota'\rangle=\langle\iota(\psi),\iota'(\psi')\rangle$, 
and then immediately that:
\begin{lemma}
\label{lem:orthogonalchi}
The canonical isomorphism 
$$
\chi_\mus^\nu:\bigotimes_{i=1}^N\bmu^{(i)}\otimes\Hom(\bigotimes_{i=1}^N\bmu^{(i)},\bnu)\rightarrow V_\mus^\nu
$$
is an isometry with respect to the inner product defined above on
$\bigotimes_{i=1}^N\bmu^{(i)}\otimes\Hom(\bigotimes_{i=1}^N\bmu^{(i)},\bnu)$
and the Euclidean inner product on the tableau basis of $\bnu$.
\end{lemma}

\subsection{Isomorphism between a sum and product of unitary group representations}

Combining the isotypical decomposition of symmetric group representations with 
Schur--Weyl duality leads to an important intertwiner with respect to the special
unitary group.  Following Harrow \cite[Claim~7.1]{10.5555/1195346}, by Schur--Weyl duality we have
\begin{eqnarray*}
(\cmplx^d)^{\otimes(\ell+m)} &=& (\cmplx^d)^{\otimes\ell}\otimes(\cmplx^d)^{\otimes m}\\
 &\overset{S_\ell\times S_m\times\SU(d)}{\cong}&\bigoplus_{\substack{\lambda\vdash\ell,\mu\vdash m\\\lambdap_1\leq d,\mup_1\leq d}}(V_\lambda\otimes U^{(d)}_\lambda)\otimes(V_\mu\otimes U^{(d)}_\mu).
\end{eqnarray*}
By Schur--Weyl duality and Lemma~\ref{lem:LRres} we also have
\begin{eqnarray*}
(\cmplx^d)^{\otimes(\ell+m)} &\overset{S_{\ell+m}\times\SU(d)}{\cong}& \bigoplus_{\substack{\nu\vdash \ell+m\\\nup_1\leq d}}(V_\nu\otimes U^{(d)}_\nu)\\
&\overset{S_\ell\times S_m\times\SU(d)}{\cong}& \bigoplus_{\substack{\lambda\vdash\ell,\mu\vdash m,\nu\vdash \ell+m\\\lambdap_1\leq d,\mup_1\leq d,\nup_1\leq d}}c_{\lambda\mu}^\nu V_\lambda\otimes V_\mu\otimes U^{(d)}_\nu.
\end{eqnarray*}
Consistency then demands
$$
\bigoplus_{\substack{\lambda\vdash\ell,\mu\vdash m\\\lambdap_1\leq d,\mup_1\leq d}}V_\lambda\otimes V_\mu\otimes U^{(d)}_\lambda\otimes U^{(d)}_\mu\overset{S_\ell\times S_m\times\SU(d)}{\cong}\bigoplus_{\substack{\lambda\vdash\ell,\mu\vdash m\\\lambdap_1\leq d,\mup_1\leq d}} V_\lambda\otimes V_\mu\otimes \bigoplus_{\substack{\nu\vdash \ell+m\\\nup_1\leq d}}c_{\lambda\mu}^\nu U^{(d)}_\nu
$$
which is satisfied if
$$
U^{(d)}_\lambda\otimes U^{(d)}_\mu\overset{\SU(d)}{\cong}\bigoplus_{\substack{\nu\vdash\ell+m\\\nup_1\leq d}} c_{\lambda\mu}^\nu U^{(d)}_\nu.
$$
See \cite[Section~6.1]{fulton1991representation} or \cite[Section~9.2.2]{goodman1998representations} for more discussion and proofs of the above isomorphism.
As with the symmetric group representations this formula can be iterated for products
of arbitrarily many representations, yielding:
\begin{lemma}
\label{lem:LRunitary}
Given $n=\sum_{i=1}^N|\mu^{(i)}|$, we have
$
\bigotimes_{i=1}^N U^{(d)}_{\mu^{(i)}}\overset{\SU(d)}{\cong}\bigoplus_{\nu\vdash n}c_{\mu^{(1)}\cdots\mu^{(N)}}^\nu U^{(d)}_\nu.
$
\end{lemma}

For the purpose of describing physical implementations of exchange-only
quantum computers, we are interested in an $S_{m_1}\times\cdots\times S_{m_N}\times \SU(d)$-embedding of the subspace
$\bigotimes_{i=1}^N(V_{\mu^{(i)}}\otimes U^{(d)}_{\mu^{(i)}})$ of the Schur--Weyl decomposition of the product of
Hilbert spaces $\prod_{i=1}^N(\cmplx^d)^{\otimes m_i}$ into the 
Schur--Weyl decomposition $\bigoplus_{\nu\vdash n}V_\nu\otimes U^{(d)}_\nu$ of the Hilbert
space $(\cmplx^d)^{\otimes n}$, 
where $\mu^{(i)}\vdash m_i$ and $n=\sum_{i=1}^Nm_i$,
by way of the intertwiners: 
$$
\bigotimes_{i=1}^N(V_{\mu^{(i)}}\otimes U^{(d)}_{\mu^{(i)}})\overset{S_{m_1}\times\cdots\times S_{m_N}\times \SU(d)}{\cong}\bigoplus_{\substack{\nu\vdash n\\\nup_1\leq d}}c_{\mu^{(1)}\cdots\mu^{(N)}}^{\nu}\bigotimes_{i=1}^NV_{\mu^{(i)}}\otimes U^{(d)}_\nu\overset{S_{m_1}\times\cdots\times S_{m_N}\times \SU(d)}{\rightarrow}\bigoplus_{\substack{\nu\vdash n\\\nup_1\leq d}}V_\nu\otimes U^{(d)}_\nu,
$$
where the first is given by Lemma~\ref{lem:LRunitary} and the second is given term by term by Lemma~\ref{lem:LRres}. 
We return to this mapping in the next chapter. 

\section{Littlewood--Richardson rules}
Having introduced the Littlewood--Richardson coefficients by their applications to
representation theory, we now explain how to calculate them, and give some useful
relationships with the partitions upon which their values depend.
It should be noted that the applications of Littlewood--Richardson 
coefficients are many and varied \cite{2020arXiv200404995B}.
The century-spanning  literature pertinent to calculating or constraining their values
is also quite rich.
Here we only summarize the key results needed for our present purposes, 
but some references given below may provide a glimpse of this greater mathematical vista.

\subsection{Definitions}
The approach to calculating the Littlewood--Richardson coefficients given below is
consistent with that originally introduced by Littlewood and Richardson \cite{LR}
and elaborated upon in \cite{fulton1991representation} and elsewhere.
There are many alternative methods, which may have relative advantages.
But the traditional approach given below is sufficient for our purposes. 

To define the Littlewood--Richardson coefficients it is convenient to introduce the Littlewood--Richardson tableau.  First, we require some more basic concepts:
\begin{Definition}
Given Young diagrams $\nu$ and $\lambda$
such that $\nu_i\geq\lambda_i$ for all $i\in\nats$,
the {\it skew diagram} denoted by
$\nu\backslash\lambda$ is obtained by the removal of the
first $\nu_i$ cells from the $i$th row of $\lambda$, for all $i$, with the remaining
cells of $\nu$ left in place.
\end{Definition}
\begin{Definition}
A {\it Littlewood--Richardson sequence} is a sequence of natural numbers,
with repetition, such that within any initial truncation of the sequence a given number
$j$ appears at least as often as any number greater than $j$. 
\end{Definition}
\begin{Definition}
\label{def:LRtableau}
A {\it Littlewood--Richardson tableau} is a skew, semistandard tableau
such that when read from right to
left and top to bottom the entries form a Littlewood--Richardson sequence.
\end{Definition}

The following examples clarify what a Littlewood--Richardson tableau is and is not.
\setydiagrameq
\begin{exmp}
The skew tableau $\ytableaushort{\none\none\none 1,\none 112,2}$ is a Littlewood--Richardson tableau.  It is in fact the only Littlewood--Richardson tableau of this shape with three 1's and two 2's.
\end{exmp}
\begin{exmp}
The skew tableau $\ytableaushort{\none \none\none 1,\none 122,1}$ is not a Littlewood--Richardson tableau because $1,2,2,1,1$ is not a Littlewood--Richardson sequence. 
\end{exmp}
\begin{exmp}
The skew tableau $\ytableaushort{\none \none\none 1,\none 211,2}$ is not a Littlewood--Richardson tableau because it is not semistandard.
\end{exmp}
\setydiagramtext

Before proceeding to the Littlewood--Richardson coefficients we
require another definition:
\begin{Definition}
The ${\it weight}$ of a semistandard tableau is a partition, the $k$th part
of which equals the number of times $k$ appears in the tableau.
\end{Definition}
\noindent Finally, the following definition encapsulates the Littlewood--Richardson rules.
\begin{Definition}
\label{def:LRcoefficient}
Given partitions $\lambda$, $\mu$, and $\nu$,
a {\it Littlewood--Richardson coefficient} $c_{\lambda\mu}^\nu$ is the number
of Littlewood--Richardson tableaux of shape $\nu\backslash\lambda$ and weight $\mu$.
\end{Definition}

\subsection{Identities}

The Littlewood--Richardson coefficients satisfy a number of symmetries \cite{2014arXiv1409.8356G}.  Two symmetries that we need are as follows:
\begin{lemma}
\label{lem:lrsymmetries}
Given partitions $\lambda$, $\mu$, and $\nu$, we have
$c_{\lambda\mu}^\nu=c_{\mu\lambda}^\nu=c_{\lambda'\mu'}^{\nu'}$.
\end{lemma}
\noindent These symmetries can be demonstrated by 
algebraic applications
of the Littlewood--Richardson coefficients \cite{2014arXiv1409.8356G}.
That $c_{\lambda\mu}^\nu=c_{\mu\lambda}^\nu$ follows from the observation
that the operation of taking the exterior tensor product of $\lambda$ with $\mu$, and from it inducing
a representation of $S_{|\lambda|+|\mu|}$, is commutative \cite[p.58]{fulton1991representation}.  That $c_{\lambda\mu}^\nu=c_{\lambda'\mu'}^{\nu'}$ follows from the invariance
of Schur functions with respect to the involutive algebra automorphism that transforms
the associated partition into its conjugate \cite{HANLON19921}.

Yet the Littlewood--Richardson rules themselves are not manifestly symmetric in the ways 
indicated by Lemma~\ref{lem:lrsymmetries}, and proving these symmetries directly by such a combinatorical approach is notoriously nontrivial \cite{2014arXiv1409.8356G}. 
Nevertheless it can be done.
The symmetry $c_{\lambda\mu}^\nu=c_{\mu\lambda}^\nu$ is painstakingly shown to arise from Littlewood--Richardson tableaux in \cite{2016arXiv160305037A},
and $c_{\lambda\mu}^\nu=c_{\lambda'\mu'}^{\nu'}$ is proved via the observation
that there are always as many Littlewood--Richardson tableaux of shape $(\nu\backslash\lambda)'$ and weight $\mu'$ as there are Littlewood--Richardson tableaux of shape
$\nu\backslash\lambda$ and weight $\mu$ in \cite{HANLON19921}.

A useful albeit obvious identity that follows immediately from the preceding results is that each
row of $\nu$ must be greater than that of $\lambda$ or $\mu$ whenever $c_{\lambda\mu}^\nu>0$:
\begin{lemma}
\label{lem:nugeqmax}
Given partitions $\lambda$, $\mu$, and $\nu$,
if $c_{\lambda\mu}^\nu>0$ then we have $\nu_i\geq\max(\lambda_i,\mu_i)$, for all $i\in\nats$.
\end{lemma}
\begin{proof}
By definition of nonzero $c_{\lambda\mu}^\nu$, the skew diagram $\nu\backslash\lambda$ exists so $\nu_i\geq\lambda_i$ for all $i$.
Further, by symmetry of the Littlewood--Richardson coefficient,
Lemma~\ref{lem:lrsymmetries}, we have $c_{\mu\lambda}^\nu=c_{\lambda\mu}^\nu$,
so $\nu_i\geq\mu_i$ for all $i$.  Thus we have $\nu_i\geq\max(\lambda_i,\mu_i)$.
\end{proof}

Finally, another well-known consequence of the Littlewood--Richardson rules comes
in handy for the part-wise sum of partitions, 
$\lambda+\mu\equiv[\lambda_1+\mu_1,\ldots,\lambda_k+\mu_k,\ldots]$
\cite[Chapter~5, Exercise~2]{fulton_1996}:
\begin{lemma}
\label{lem:LRCartan}
Given any partitions $\lambda$ and $\mu$, it follows that $c_{\lambda\mu}^{\lambda+\mu}=1$.
\end{lemma}
\begin{proof}
By definition, each row of $(\lambda+\mu)\backslash\lambda$  is the same length
as that of $\mu$.  Therefore, a Littlewood--Richardson tableau of weight $\mu$
can be constructed by filling the $i$th row of $(\lambda+\mu)\backslash\lambda$ with
$\mu_i$ occurrences of $i$, so $c_{\lambda\mu}^{\lambda+\mu}>0$.  
That this is the only Littlewood--Richardson tableau of shape $(\lambda+\mu)\backslash\lambda$ and weight $\mu$ 
can be proved inductively.
The 1st row of $(\lambda+\mu)\backslash\lambda$ must be filled with
$\mu_1$ occurrences of 1, because if it is not then the higher numbers 
that must appear in the rightmost cells to satisfy the definition of a 
semistandard tableau violate the Littlewood--Richardson sequence.
For some $k\in\nats$ and each $1\leq i\leq k$, assume that the $i$th row of
$(\lambda+\mu)\backslash\lambda$ is filled with $\mu_i$ occurrences of $i$.
Then the $(k+1)$th row of $(\lambda+\mu)\backslash\lambda$ must be filled with
$\mu_{k+1}$ occurrences of $k+1$, because if it is not then the higher numbers 
that must appear in the rightmost cells to satisfy the definition of a 
semistandard tableau again violate the Littlewood--Richardson sequence.
\end{proof}

\subsection{Horn inequalities}

In 1962 Horn conjectured that, given $d\times d$ Hermitian matrices $A$ and $B$ with
eigenvalues $\{\alpha_i\}_{i=1}^d$ and $\{\beta_i\}_{i=1}^d$ respectively,
$A+B$ has eigenvalues $\{\gamma_i\}_{i=1}^d$ 
if and only if the latter satisfy
a certain set of 
inequalities with $\{\alpha_i\}_{i=1}^d$ and $\{\beta_i\}_{i=1}^d$ \cite{Horn:conj}.
Horn's conjecture was finally proved in the affirmative by Knutson and Tao in 1999 \cite{1998math......7160K}.
The reason this concerns us here is because, in the process of proving Horn's conjecture,
Knutson and Tao built on earlier work by Klyachko establishing a close relationship between Hermitian eigenvalues and Littlewood--Richardson coefficients \cite{30002131808}.  In particular we have \cite{1999math......8012F,anderson_richmond_yong_2013}
\begin{lemma}[Klyachko--Knutson--Tao]
\label{lem:KnutsonTao}
Given partitions $\lambda$, $\mu$, and $\nu\vdash|\lambda|+|\mu|$
with at most $d$ rows each, we have
$c_{\lambda\mu}^\nu>0$ if and only if there exist
$d\times d$ Hermitian matrices $A$ and $B$ such that
$A$, $B$, and $A+B$ have eigenvalues $\{\lambda_i\}_{i=1}^d$,
$\{\mu_i\}_{i=1}^d$ and $\{\nu_i\}_{i=1}^d$
respectively.
\end{lemma}
\noindent This result alone is quite powerful, since there are cases where
constructing the requisite Hermitian matrices with the desired eigenvalues may be easier than
direct application of the Littlewood--Richardson rules, as we demonstrate in Section~4.2.4.

In conjuction with Knutson and Tao's proof of Horn's conjecture, then,
we have the machinery of the Horn inequalities at our disposal. 
Applied to the Littlewood--Richardson coefficients, this takes the following form \cite[Theorem~11]{1999math......8012F}:  
Given partitions $\lambda$, $\mu$, and $\nu\vdash|\lambda|+|\mu|$ with at most
$d$ parts each, we have $c_{\lambda\mu}^\nu>0$ if and only if 
$$
\sum_{k\in K}\nu_k\leq\sum_{i\in I}\lambda_i+\sum_{j\in J}\mu_j
$$
for all $(I,J,K)$ in a set of triples of subsets of $\{1,\ldots,d\}$, as
prescribed in \cite[Theorem~1]{FULTON200023}.
Or equivalently \cite[Equation~(7)]{1999math......8012F}, we have
$$
\sum_{i\in I^c}\lambda_i+\sum_{j\in J^c}\mu_j\leq\sum_{k\in K^c}\nu_k
$$
where $I^c$, $J^c$, and $K^c$ are the complements of $I$, $J$, and $K$ in 
$\{1,\ldots,d\}$.
\jrvc{Give prescription and promote to theorem or lemma?}

In the present work however we will use only a subset of the Horn inequalities,
ones that conveniently constrain individual parts of $\nu$ rather than sums of parts.  Originally discovered by Weyl in the context of Hermitian eigenvalues
\cite{Weyl1912}, one such set of inequalities takes the form \cite{Fulton1997-1998}:
$$
\nu_{i+j-1}\leq\lambda_i+\mu_j,\ i+j-1\leq d.
$$
A related set of inequalities, derivable from the Weyl inequalities and
referred to as their dual \cite{taoblog}, may be considered a subset of
the complement of the Horn inequalities above  \cite[Equation~(11)]{1999math......8012F}:
$$
\lambda_i+\mu_j\leq\nu_{i+j-d},\ i+j-d\leq d.
$$
Setting $k=i+j-1$, the Weyl inequalities may be expressed
$$
\nu_k\leq\min_{i+j=k+1}\{\lambda_i+\mu_j\},
$$
while setting $k=i+j-d$, the dual Weyl inequalities may be expressed
$$
\max_{i+j=k+d}\{\lambda_i+\mu_j\}\leq\nu_k.
$$
In light of Lemma~\ref{lem:KnutsonTao}, the above results can be put in the 
following form:
\begin{lemma}[Weyl]
\label{lem:Weyl}
Given partitions $\lambda$, $\mu$, and $\nu$,
if $c_{\lambda\mu}^\nu>0$ then for all $1\leq k\leq d=\nup_1$, we have
$$\max_{i+j=k+d}\{\lambda_i+\mu_j\}\leq\nu_k\leq\min_{i+j=k+1}\{\lambda_i+\mu_j\},$$
where as indicated the maximum on the left hand side is over the set $\{(i,j)|1\leq i,j\leq d,i+j=k+d\}$ and the minimum on the right hand side is over the set $\{(i,j)|1\leq i,j\leq d,i+j=k+1\}$.
\end{lemma}

\chapter{Exchange-only quantum computing}

The purpose of this chapter is to provide some physics context for mathematicians, 
and perhaps vice-versa.
Our specific aim is to introduce those quantum computing concepts needed to motivate 
relevant mathematical problems.  Broader context can be found, for example, in \cite{mikeandike}.
For the purposes of the present work, quantum computation consists primarily of two kinds of
components: ``qudits", those being the fundamental units of quantum information that generalize
binary digits, and logical gates, those being the fundamental steps in a computation.

A $d$-dimensional qudit may be indentified with a $d$-dimensional Hilbert space. 
We also consider each qudit as a $d$-dimensional $G$-module for a group $G$, that is as a group representation. 
Multiple qudits can then be combined by taking the tensor product.
For example, $n$ qudits of $d$ dimensions each can be expressed as $(\cmplx^d)^{\otimes n}$.
Of course as a Hilbert space this product is isomorphic to $\cmplx^{d^n}$.  What distinguishes
the two expressions is the set of operators we consider on the product, as determined by
the tensor product of the component $G$-modules, which naturally respect the qudit structure.

A gate is a unitary operator on one or more qudits.
If the qudits are denoted by $\bmu^{(1)}\ldots\bmu^{(n)}$,
the relevant gates are elements of $\SU(\bmu^{(1)}\otimes\cdots\otimes\bmu^{(n)})$.
We focus on the special unitary group, rather than the unitary group, because
elements of the latter differ by the former only by scalar ``phase" factors which 
typically have no bearing on computations.
The reason for the irrelevance of such scalar factors is that, strictly speaking, 
quantum physical states, and thus the associated logical states, are given by rays in a projective Hilbert space.
Of vital importance, however, is that every element of the relevant special unitary group be
implementable on the given qudits, using whatever operations are permitted by the underlying physical system.  We call this essential property ``universality", and developing its mathematical 
definition occupies much of this and the next chapter.

We could instead ask for just ``approximate universality" \cite{mikeandike}, meaning the physically implementable 
operators are dense in the desired unitary operators.  But in our case, the implementable operators are 
generated by a representation of a Lie algebra,
which proves to be isomorphic to the special linear, special orthogonal, or symplectic algebra.
The implementable operators thus form a matrix Lie group $G$, which has the 
property that if any sequence of its elements converges to a matrix $A$,
then $A$ is either in $G$ or is not invertible \cite{Hall}.  Since every unitary operator
is invertible, we conclude in our case that approximate universality implies exact
universality, and so we need only consider the latter.  

\section{Logical qudits}

Often a qudit can be identified with a subspace of a Hilbert space corresponding to a quantum physical
system, as indicated above.  Such a qudit is sometimes called a ``physical qudit", and textbook examples include the spin of 
a particle such as an electron.  Mathematically, a $d$-dimensional physical qudit is also commonly identified with the 
defining representation of $\SU(d)$.

More generally, a qudit need not be identified with a subspace of such a Hilbert space, but rather
with a ``subsystem" of the Hilbert space.  A subsystem refers to a tensor factor of a
subspace of a Hilbert space: if $\cH=\bigoplus_i S_i\otimes B_i$, then each factor space $S_i$ or $B_i$ is a subsystem \cite{DFS}.
A measurement of the state of such a qudit, say $S_i$, would then be associated with a ``logical observable" of the 
form $O\otimes\one$, where $O$ is a Hermitian operator corresponding to an observable of interest and $\one$ is the 
identity on $B_i$.  A logical gate, however, is a unitary operator of the form $U\otimes A$, where 
$U$ is the gate of interest and $A$ is arbitrary.

\subsection{Noiseless subsystems}

While in principle any subsystem of the Hilbert space may serve as a qudit, some choices are better than others.  
For example, there are typically many physical sources of noise which may introduce errors into the encoded information.  
Ideally, then, a subsystem can be chosen on which some such noise acts as the identity.  That is, if $\cH=\bigoplus_i S_i\otimes B_i$, and the noise acts as $\bigoplus_i \one\otimes N_i$, then any $S_i$ is a good choice of qudit.  
Such a subsystem, or qudit, is called ``noiseless" or ``decoherence-free" \cite{DFS}.

A common type of noise is called ``collective" noise \cite{PhysRevA.63.042307}.  In this context, the relevant Hilbert space is assumed to be a product 
of $d$-dimensional physical qudits, which in order to distinguish from logical qudits we refer to as ``$d$-state systems".  
Collective noise is noise that affects many $d$-state systems at the same time, in the same way.  Typically, collective noise 
is further assumed to act as an operator in $\SU(d)$.  For example, if the $d$-state systems are particles with spin, 
collective noise might simultaneously flip all the spins.  This could be the case if the source of noise is electromagnetic 
radiation with wavelength long compared to the distance between the particles.  Therefore, we are motivated to find a 
collective-decoherence-free subsystem (CDFS).

Finding every physically relevant CDFS in a Hilbert space is readily achieved by Schur--Weyl decomposition.
Given $\ell$ copies of the $d$-state system $V$, the resulting Hilbert space is
$$
V^{\otimes\ell}\overset{S_\ell\times\SU(d)}{\cong}\bigoplus_{\lambda\vdash\ell,\lambdap_1\leq d}V_\lambda\otimes U^{(d)}_\lambda.
$$
As first proposed by Zanardi and Rasetti \cite{Zanardi},
each $V_\lambda$ can be used as a collective-decoherence-free logical qudit, which we call a CDFS qudit.
Operating on $d$-state systems by a swap then corresponds to acting on $V_\lambda$ by a transposition. 
Collective noise meanwhile acts on the $d$-state systems by simultaneous matrix multiplication,
and on $U_\lambda\otimes V_\lambda$ as a nontrivial
operation on $U_\lambda$ times the identity on $V_\lambda$.

It is assumed that the only operations on the $d$-state systems utilized for computation on the CDFS qudits
are those generated by quantum physical exchange-interactions, by which the states of
$d$-state systems are interchanged.  The exchange-interactions may be identified with transpositions
permuting a set of numbers labeling the $d$-state systems. 
Exchange-interactions among the $d$-state systems comprising a single logical qudit then generate
gates on this logical qudit, and exchange-interactions among $d$-state systems of multiple logical
qudits generate gates on these combined logical qudits.
In addition to the protection it provides from collective noise, a CDFS qudit may also be a convenient choice for ease of control in physical systems where the exchange interaction is readily implemented.

\subsection{Physical and logical bases}

Thinking in terms of certain bases may help clarify the above concepts, as well as bridge the gap to 
common physics formulations when considering practical applications.  For this purpose we use 
Dirac's bra-ket notation, whereby a vector labeled by $v$ is denoted by the ``ket" $\ket{v}$.  
A basis of a logical qudit may be 
identified with the quantum analogs of logical values in a many-valued logic: 0, 1, and (possibly) 
so forth.  The associated basis vectors are denoted $\ket{0_L}$, $\ket{1_L}$, $\ket{2_L}$, up to
$\ket{(D-1)_L}$ for a $D$-dimensional qudit.  Since the logical qudits at issue are symmetric group 
representations, which already have a convenient basis in the form of Young tableaux, it is natural 
to identify the logical basis vectors with the Young tableaux.   This identification is 
consistent with literature on the topic, for example \cite{PhysRevA.99.042331}. 

Physical qudit, or $d$-state system, basis vectors may be similarly labeled, without the $L$ subscript.  Tensor products of such 
basis vectors, e.g. $\ket{1}\otimes\ket{0}\otimes\ket{2}$, are commonly denoted by a string of such labels in a 
single ket, e.g. $\ket{102}$. These form a basis of the Hilbert space consisting of a product of such 
$d$-state sytems.  So the $d^n$ basis vectors of the Hilbert space $(\cmplx^d)^{\otimes \ell}$ have the form 
$\ket{b_1\cdots b_\ell}$ where $b_i\in\{0,\ldots,d-1\}$.  We call this the physical basis.  

Since the Hilbert space $(\cmplx^d)^{\otimes \ell}$ is also a representation of $S_\ell\times\SU(d)$, each 
physical basis vector can be acted upon by elements of $S_\ell\times\SU(d)$.  In particular, 
transpositions act on physical basis vectors by transposing $d$-state system labels:
$$
(i j)\ket{b_1\cdots b_i\cdots b_j\cdots b_\ell}=\ket{b_1\cdots b_j\cdots b_i\cdots b_\ell}.
$$
By Schur--Weyl duality, each tensor product of a Young tableau in $V_\lambda$ with a Weyl 
tableau in $U^{(d)}_\lambda$ maps by $S_\ell\times\SU(d)$-intertwiner to a linear combination of 
physical basis vectors in $(\cmplx^d)^{\otimes \ell}$ which is unique up to overall scalar factor.

In the examples below we consider 2-state systems and,
in order to make contact with relevant physics literature,
further associate the 2-state systems with spin-1/2 particles such as electrons.
For the purpose of illustrating the correspondence between the mathematical
and physical ideas, 
an operator of particular interest on the Hilbert space $(\cmplx^2)^{\otimes\ell}$ is 
the $z$-component of the collective spin operator, $S_z$.  
By convention, $S_z$ acts on the Hilbert space as a diagonal operator in the physical basis,
for which each eigenvalue equals half the number of 0's labeling the basis vector minus
half the number of 1's labeling the basis vector. So for example 
$S_z\ket{010}=\frac{1}{2}\ket{010}$ and $S_z\ket{111}=-\frac{3}{2}\ket{111}$.
By further convention $S_z$ is proportional, by a factor of $\frac{i}{2}$, 
to an element in the fundamental representation of $\su(2)$. 
Therefore by Schur--Weyl duality $S_z$ acts trivially on $V_\lambda$, 
for each $\lambda\vdash\ell$,
and nontrivially on $U_\lambda$.  
For each such $\lambda$, $S_z$ is also diagonal in the Weyl tableau
basis, for which each eigenvalue equals one half the number of 1's in the 
tableau minus one half
the number of 2's in the tableau. 
\setydiagrameq 
So for example $S_z\ytab{1 1,2}=\frac{1}{2}\ytab{1 1,2}$ and $S_z\ytab{2 2 2}=-\frac{3}{2}\ytab{2 2 2}$.
\setydiagramtext
Another physical quantity of note is the collective spin magnitude $S$, 
which for a given $U_\lambda$ equals the maximum eigenvalue of $S_z$.  
It follows that a collective spin magnitude is uniquely associated with each 
$U_\lambda$ by the equation $S=\frac{1}{2}(\lambda_1-\lambda_2)$.

Our goal is to map tableaux to physical basis vectors, in order to provide insight into noiseless 
subsystems, as well as to make connections with potential physical applications.
We consider three examples of noiseless subsystems: two logical qubits and a logical qutrit.
The existence of all three of these encodings have been previously noted \cite{PhysRevA.63.042307}.  
The two logical qubits have also been given previously in terms of the explicit 
Hilbert space basis used below, and in fact have been rather extensivey 
investigated \cite{divincenzo:qc2000b,PhysRevA.99.042331,Hsieh2003}.  In contrast, 
the logical qutrit and higher dimensional qudits have not received
as much attention in the literature, but have potential advantages suggested in the next subsection.

\begin{exmp}

Consider encoding a logical qubit on three 2-state systems.  The Schur--Weyl decomposition of
the corresponding Hilbert space is:
\setydiagrameq
\begin{equation}
(\cmplx^2)^{\otimes 3}\overset{S_3\times SU(2)}{\cong}(V_\subydiagrameq{3}\otimes U_\subydiagrameq{3})\oplus(V_\subydiagrameq{2,1}\otimes U_\subydiagrameq{2,1}).
\end{equation}
\setydiagramtext
The only term with a two-dimensional symmetric group representation is:
\setydiagrameq
$$
V_\subydiagrameq{2,1}\otimes U_\subydiagrameq{2,1} \cong \text{span}\left\{\ytableaushort{1 3,2}\otimes\ytableaushort{1 1,2},\ \ytableaushort{1 2,3}\otimes\ytableaushort{1 1,2},\ \ytableaushort{1 3,2}\otimes\ytableaushort{1 2,2},\ \ytableaushort{1 2,3}\otimes\ytableaushort{1 2,2}\right\}.
$$
\setydiagramtext
By Young's orthogonal form the action of neighboring transpositions on the supporting Young tableaux is given by:
\setydiagrameq
\begin{eqnarray*}
(1\ 2).\ytableaushort{1 3,2 } &=& -\ytableaushort{1 3,2 }\\
(2\ 3).\ytab{1 3,2 } &=& \frac{1}{2}\,\ytab{1 3,2 }+\frac{\sqrt{3}}{2}\,\ytab{1 2,3 }\\
(1\ 2).\ytab{1 2,3 } &=& \phantom{-}\ytab{1 2,3 }\\
(2\ 3).\ytab{1 2,3 } &=& -\frac{1}{2}\,\ytab{1 2,3 }+\frac{\sqrt{3}}{2}\,\ytab{1 3,2 }
\end{eqnarray*}

Our general strategy for mapping the tableau basis to the physical basis by the Schur--Weyl isomorphism is to first deduce the image of one of the tableau basis vectors.  The image of every other
tableau basis vector is then uniquely determined by consistency with the symmetric group action
as given by the equations above.  In the present case, antisymmetry with respect to 1 and 2 of the 
Young tableau $\ytab{1 3,2}$ requires that it map to a Hilbert space vector proportional to 
$\ket{01?}-\ket{10?}$, where the final bit is to be determined.  Equivariance with respect to the 
$S_z$ operator then requires that $\ytab{1 3,2}\otimes\ytab{1 1,2}$ map 
to $a(\ket{010}-\ket{100})$, where we are free to choose $a\in\cmplx$.
For quantum mechanical purposes we choose $a=\frac{1}{\sqrt{2}}$ in order to normalize the vector.
Then by the second equation above, we have
$$
\ytab{1 2,3}\otimes\ytab{1 1,2} \,\mapsto\, \frac{2}{\sqrt{3}}((2\ 3).(\frac{1}{\sqrt{2}}(\ket{010}-\ket{100}))-\frac{1}{2}(\frac{1}{\sqrt{2}}(\ket{010}-\ket{100}))).
$$
Proceeding in this way, we arrive at:
\setydiagrameq
\begin{eqnarray*}
|0_L\rangle|\tfrac{1}{2},+\tfrac{1}{2}\rangle &\equiv& \ytableaushort{1 3,2}\otimes\ytableaushort{1 1,2} \,\mapsto\, \frac{1}{\sqrt{2}}(|010\rangle-|100\rangle) \\
|1_L\rangle|\tfrac{1}{2},+\tfrac{1}{2}\rangle &\equiv& \ytableaushort{1 2,3}\otimes\ytableaushort{ 1  1, 2} \,\mapsto\, \frac{1}{\sqrt{6}}(2|001\rangle-|100\rangle-|010\rangle)\\
|0_L\rangle|\tfrac{1}{2},-\tfrac{1}{2} \rangle &\equiv& \ytableaushort{1 3,2}\otimes\ytableaushort{ 1  2, 2} \,\mapsto\, \frac{1}{\sqrt{2}}(|101\rangle-|011\rangle)\\
|1_L\rangle|\tfrac{1}{2},-\tfrac{1}{2}\rangle &\equiv& \ytableaushort{1 2,3}\otimes\ytableaushort{1 2,2} \,\mapsto\, \frac{1}{\sqrt{6}}(2|110\rangle-|011\rangle-|101\rangle),
\end{eqnarray*}
\setydiagramtext
where by convention we have identified each Young tableau with a logical basis vector,
and each Weyl tableau with an eigenvector of the spin labeled by the total spin (given by the maximum eigenvalue of $S_z$), as well as the relevant eigenvalue of $S_z$. 
We then observe that flipping all the bits on the right hand side, which is a simple example
of collective noise, preserves the corresponding logical basis vectors.
\end{exmp}

\begin{exmp}
Consider encoding a logical qubit on four 2-state systems.
By Schur--Weyl duality we have
\setydiagrameq
\begin{equation*}
(\cmplx^2)^{\otimes 4}\overset{S_4\times SU(2)}{\cong}(V_\subydiagrameq{4}\otimes U_\subydiagrameq{4})\oplus(V_\subydiagrameq{2,2}\otimes U_\subydiagrameq{2,2})\oplus(V_\subydiagrameq{3,1}\otimes U_\subydiagrameq{3,1}).
\end{equation*}
\setydiagramtext
To encode a logical qubit, we again use the only term with a two-dimensional symmetric group representation:
\setydiagrameq
\begin{equation*}
V_\subydiagrameq{2,2}\otimes U_\subydiagrameq{2,2} \cong \text{span}\left\{\ytableaushort{1 3,2 4}\otimes\ytableaushort{1 1,2 2},\ \ytableaushort{1 2,3 4}\otimes\ytableaushort{1 1,2 2}\right\}.
\end{equation*} 
The action of neighboring transpositions on the Young tableaux is given by:
\begin{eqnarray*}
(1\ 2).\ytableaushort{1 3,2 4} &=& -\ytableaushort{1 3,2 4}\\
(2\ 3).\ytab{1 3,2 4} &=& \frac{1}{2}\,\ytab{1 3,2 4}+\frac{\sqrt{3}}{2}\,\ytab{1 2,3 4}\\
(3\ 4).\ytab{1 3,2 4} &=& -\ytab{1 3,2 4}\\
(1\ 2).\ytab{1 2,3 4} &=& \phantom{-}\ytab{1 2,3 4}\\
(2\ 3).\ytab{1 2,3 4} &=& -\frac{1}{2}\,\ytab{1 2,3 4}+\frac{\sqrt{3}}{2}\,\ytab{1 3,2 4}\\
(3\ 4).\ytab{1 2,3 4} &=& \phantom{-}\ytab{1 2,3 4}
\end{eqnarray*}

Since $\ytab{1 3,2 4}$ is antisymmetric with respect to transpositions of 1 and 2 and of 3 and 4,
the only possible Hilbert space vector to which $\ytab{1 3,2 4}\otimes\ytab{1 1,2 2}$ can map
is, up to scalar factor, 
$$(\ket{01}-\ket{10})(\ket{01}-\ket{10}).$$
Choosing a scalar factor to normalize that vector,
and then deducing the image of $\ytab{1 2,3 4}\otimes\ytab{1 1,2 2}$,
we obtain
\begin{eqnarray*}
\ket{0_L}\ket{0,0} &\equiv& \ytab{1 3,2 4}\otimes\ytab{1 1,2 2}\,\mapsto\,\frac{1}{2}(\ket{0101}-\ket{0110}-\ket{1001}+\ket{1010})\\
\ket{1_L}\ket{0,0} &\equiv& \ytab{1 2,3 4}\otimes\ytab{1 1,2 2}\,\mapsto\,\frac{1}{\sqrt{12}}(2\ket{0011}+2\ket{1100}-\ket{0101}-\ket{0110}-\ket{1001}-\ket{1010})
\end{eqnarray*}
We observe that this time flipping the bits on the right hand side preserves each vector.
\end{exmp}

\begin{exmp}
Consider encoding a logical qutrit on four 2-state systems.
\setydiagramtext
Referring to the Schur--Weyl duality expression of Example~6,
the only term with a three-dimensional symmetric group representation is:
\begin{align*}
V_\subydiagrameq{3,1}\otimes\U_\subydiagrameq{3,1} \cong \Span\Big\{
& \ytab{1 3 4,2}\otimes\ytab{1 1 1,2},\ \ytab{1 3 4,2}\otimes\ytab{1 1 2,2},\ \ytab{1 3 4,2}\otimes\ytab{1 2 2,2},\\
& \ytab{1 2 4,3}\otimes\ytab{1 1 1,2},\ \ytab{1 2 4,3}\otimes\ytab{1 1 2,2},\ \ytab{1 2 4,3}\otimes\ytab{1 2 2,2},\\
& \ytab{1 2 3,4}\otimes\ytab{1 1 1,2},\ \ytab{1 2 3,4}\otimes\ytab{1 1 2,2},\ \ytab{1 2 3,4}\otimes\ytab{1 2 2,2}\Big\}
\end{align*}
The action of neighboring transpositions on the Young tableaux is:
\setydiagrameq
\begin{eqnarray*}
(12).\ytableaushort{1 3 4,2} &=& -\ytableaushort{1 3 4,2}\\
(23).\ytableaushort{1 3 4,2} &=& \frac{1}{2}\,\ytableaushort{1 3 4,2}+\frac{\sqrt{3}}{2}\,\ytableaushort{1 2 4,3}\\
(34).\ytab{1 3 4,2} &=& \ytab{1 3 4,2}\\
(12).\ytab{1 2 4,3} &=& \ytab{1 2 4,3}\\
(23).\ytab{1 2 4,3} &=& -\frac{1}{2}\,\ytab{1 2 4,3}+\frac{\sqrt{3}}{2}\,\ytab{1 3 4,2}\\
(34).\ytableaushort{1 2 4,3} &=& \frac{1}{3}\,\ytableaushort{1 2 4,3}+\frac{\sqrt{8}}{3}\,\ytableaushort{1 2 3,4}\\
(12).\ytab{1 2 3,4} &=& \ytab{1 2 3,4}\\
(23).\ytab{1 2 3,4} &=& \ytab{1 2 3,4}\\
(34).\ytab{1 2 3,4} &=& -\frac{1}{3}\,\ytableaushort{1 2 3,4}+\frac{\sqrt{8}}{3}\,\ytableaushort{1 2 4,3}.
\end{eqnarray*}

As in the previous examples, antisymmetry with respect to transposition of 1 and 2 demands that
$\ytab{1 3 4,2}$ map to a vector with first two physical qubit factors having the form 
$\ket{01}-\ket{10}$.  The remaining factors are determined by the symmetry with respect to 
transposition of 3 and 4, together with the spin as given by each Weyl tableau as before.
Then normalizing and solving for the remaining vectors in accordance with the above
transposition actions, we obtain:
\begin{eqnarray*}
|0_L\rangle|1,+1\rangle &\equiv& \ytableaushort{1 3 4,2}\otimes\ytab{1 1 1,2}\,\mapsto\,\frac{1}{\sqrt{2}}(|0100\rangle-|1000\rangle)\\
|1_L\rangle|1,+1\rangle &\equiv& \ytableaushort{1 2 4,3}\otimes\ytab{1 1 1,2}\,\mapsto\,\frac{1}{\sqrt{6}}(2|0010\rangle-|1000\rangle-|0100\rangle)\\
|2_L\rangle|1,+1\rangle &\equiv& \ytableaushort{1 2 3,4}\otimes\ytab{1 1 1,2}\,\mapsto\,\frac{1}{\sqrt{12}}(3|0001\rangle-|1000\rangle-|0100\rangle-|0010\rangle)\\
|0_L\rangle|1,\phantom{+}0\rangle &\equiv& \ytableaushort{1 3 4,2}\otimes\ytab{1 1 2,2}\,\mapsto\,\frac{1}{2}(\ket{0101}+\ket{0110}-\ket{1001}-\ket{1010})\\
|1_L\rangle|1,\phantom{+}0\rangle &\equiv& \ytableaushort{1 2 4,3}\otimes\ytab{1 1 2,2}\,\mapsto\,\frac{1}{\sqrt{12}}(2\ket{0011}-2\ket{1100}+\ket{0110}-\ket{1001}-\ket{0101}+\ket{1010})\\
|2_L\rangle|1,\phantom{+}0\rangle &\equiv& \ytableaushort{1 2 3,4}\otimes\ytab{1 1 2,2}\,\mapsto\,\frac{1}{\sqrt{6}}(\ket{0011}-\ket{1100}+\ket{0101}-\ket{1010}-\ket{0110}+\ket{1001})\\
|0_L\rangle|1,-1\rangle &\equiv& \ytableaushort{1 3 4,2}\otimes\ytab{1 2 2,2}\,\mapsto\,\frac{1}{\sqrt{2}}(|1011\rangle-|0111\rangle)\\
|1_L\rangle|1,-1\rangle &\equiv& \ytableaushort{1 2 4,3}\otimes\ytab{1 2 2,2}\,\mapsto\,\frac{1}{\sqrt{6}}(2|1101\rangle-|0111\rangle-|1011\rangle)\\
|2_L\rangle|1,-1\rangle &\equiv& \ytableaushort{1 2 3,4}\otimes\ytab{1 2 2,2}\,\mapsto\,\frac{1}{\sqrt{12}}(3|1110\rangle-|0111\rangle-|1011\rangle-|1101\rangle)
\end{eqnarray*}
Again, flipping the bits on the right hand side preserves the corresponding logical qudits.
\end{exmp}

Related to the above is another relatively simple encoding, that of a 
logical ``ququint" on five 2-state systems. This ququint may be considered as a logical qutrit direct-summed with a logical qubit, in the following sense.
There exists a natural $S_4$-module isomorphism from $V_\subydiagramtext{2,2}\oplus V_\subydiagramtext{3,1}$
to $V_\subydiagram{3,2}$ given by:
\setydiagrameq
\begin{eqnarray*}
\ytab{1 3,2 4} &\mapsto& \ytab{1 3 5,2 4}\\
\ytab{1 2,3 4} &\mapsto& \ytab{1 2 5,3 4}\\
\ytab{1 3 4,2} &\mapsto& \ytab{1 3 4,2 5}\\
\ytab{1 2 4,3} &\mapsto& \ytab{1 2 4,3 5}\\
\ytab{1 2 3,4} &\mapsto& \ytab{1 2 3,4 5}.
\end{eqnarray*}
In this way the $\ket{0_L}$ and $\ket{1_L}$ states of this logical ququint are related to the
$\ket{0_L}$ and $\ket{1_L}$ of the logical qubit of Example~6, and the $\ket{2_L}$, $\ket{3_L}$, and
$\ket{4_L}$ states of the logical ququint are related to the $\ket{0_L}$, $\ket{1_L}$, and $\ket{2_L}$
states of the logical qutrit of Example~7.

\setydiagramtext

\subsection{Coding efficiency}

\ytableausetup{smalltableaux}
\begin{table}[h!]
\label{tab:efficiency}
\centering
{
  \begin{tabular}{|c|c|c|c|}
  \hline
  $|\lambda|\backslash d$ & 2 & 3 & 4 \\
  \hline
  3 &  \s{ \\ \ydiagram{2,1}\\D=2, E=0.33}  &  \s{\ydiagram{2,1}\\D=2, E=0.21} &  \s{\ydiagram{2,1}\\D=2, E=0.17} \\
  \hline
  4 &  \s{ \\ \ydiagram{3,1}\\D=3, E=0.40} &  \s{\ydiagram{3,1}\\D=3, E=0.25} &  \s{\ydiagram{3,1}\\D=3, E=0.20} \\
  \hline
  5 &  \s{ \\ \ydiagram{3,2}\\D=5, E=0.46} &  \s{\ydiagram{3,2}\\D=5, E=0.29} &  \s{\ydiagram{3,2}\\D=5, E=0.23} \\
  \hline
  6 &  \s{ \\ \ydiagram{4,2}\\D=9, E=0.53} &  \s{\ydiagram{4,2}\\D=9, E=0.33} &  \s{\ydiagram{4,2}\\D=9, E=0.26} \\
  \hline
  7 &  \s{\ydiagram{4,3}\\D=14, E=0.54} &  \s{\ydiagram{4,2,1}\\D=35, E=0.46} &  \s{ \\ \ydiagram{4,2,1}\\D=35, E=0.37} \\
  \hline
  8 &  \s{\ydiagram{5,3}\\D=28, E=0.60} &  \s{\ydiagram{4,3,1}\\D=70, E=0.48} &  \s{ \\ \ydiagram{4,3,1}\\D=70, E=0.38} \\
  \hline
  9 &  \s{\ydiagram{6,3}\\D=48, E=0.62} &  \s{\ydiagram{4,3,2}\\D=168, E=0.52} &  \s{ \\ \ydiagram{4,3,1,1}\\D=216, E=0.43} \\
  \hline
 10 &  \s{\ydiagram{6,4}\\D=90, E=0.65} &  \s{\ydiagram{5,3,2}\\D=450, E=0.56} &  \s{ \\ \ydiagram{5,3,1,1}\\D=567, E=0.46} \\
  \hline
 11 &  \s{\ydiagram{7,4}\\D=165, E=0.67} &   \s{\ydiagram{5,4,2}\\D=990, E=0.57} &  \s{ \\ \ydiagram{5,3,2,1}\\D=2310, E=0.51} \\
  \hline
 12 &   \s{\ydiagram{7,5}\\D=297, E=0.68} &  \s{\ydiagram{6,4,2}\\D=2673, E=0.60}  &  \s{ \\ \ydiagram{5,4,2,1}\\D=5775, E=0.52} \\
  \hline
 13 &  \s{\ydiagram{8,5}\\D=572, E=0.70} &  \s{\ydiagram{6,4,3}\\D=6435, E=0.61} &  \s{ \\ \ydiagram{6,4,2,1}\\D=17160, E=0.54} \\
  \hline
  \end{tabular}
\caption{For $|\lambda|$ $d$-state systems,
the Young diagram labeling the irreducible representation of maximum dimension that admits
universality (by Theorem~\ref{thm:iff1partitionu}) is shown, along with $D=\dim(\blambda)$ (from \cite[p.414--419]{james_1984}) and the 
resulting coding efficiency $E$.
}
}
\end{table}
\ytableausetup{nosmalltableaux}

Clearly $\ydiagram{2,1}$ gives a more efficient encoding of a logical qubit
than $\ydiagram{2,2}$ since it requires one fewer physical qubit for the job.
Also $\ydiagram{3,1}$ encodes a logical qutrit more efficiently than 
$\ydiagram{2,2}$ encodes a logical qubit, 
since the former yields more dimensions for the cost of the same number of physical
qubits.  But, how to compare the logical qubit given by $\ydiagram{2,1}$ with the
logical qutrit given by $\ydiagram{3,1}$~?
One approach is to appeal to classical information theory, according to which
the information content of a trit, in bits, is $\log_2(3)$.
By analogy, this may also be considered a measure of the quantum information,
in qubits, associated with a qutrit.
More generally, the number of qudits of dimension $d$ associated with a quDit of
dimension $D$ is $\log_d(D)$.
By this measure, the information in physical qudits encoding a logical qudit given by
$\blambda$ is $\log_d(\dim(\blambda))$.  Dividing by the number of physical
qudits encoding the logical qudit, that is the size of $|\lambda|$,
then gives a measure of the efficiency of the encoding: $E(d,\lambda)=\log_d(\dim(\blambda))/|\lambda|$ \cite{PhysRevA.63.042307}.
Conveniently, this quantity is invariant under conversion to $d$-state systems of 
different dimensions, since the conversion factors cancel.

We can then ask how to optimize the coding efficiency.  For fixed $|\lambda|$
and $d$, $E(d,\lambda)$ is optimized by choosing the shape of $\lambda$ that maximizes
$\dim(\blambda)$.  Although, as we discuss in Section~4.1, the set of acceptable
$\lambda$ should be restricted to those that allow universal computation.
Since the lower bound of both the maximum dimension and typical dimension of $\blambda$ grows
superpolynomially in $|\lambda|$ \cite{Vershik1985}, for fixed $d$ we further expect the optimal
$E(d,\lambda)$ to trend upwards with the number of $d$-state systems, $|\lambda|$.
For fixed $|\lambda|$, meanwhile, we see in Table~\ref{tab:efficiency} that $d=2$ optimizes
$\max_{\lambda\vdash\ell}E(d,\lambda)$ up to $|\lambda|=13$. 

Of course, other factors may determine the best encoding for a given application.
An obvious constraint is that imposed by the available resources, such as the
number and dimension of $d$-state systems on hand.  
Another consideration is that quantum circuits taking advantage of qudits with dimension greater than two
are less well-studied than those designed for qubits. 
Still the latter can be applied to a $2^n$-dimensional subspace of the qudit. 
For related reasons, a large qudit dimension
may be relatively advantageous in terms of circuit complexity or algorithm 
efficiency, independent of the coding efficiency \cite{Lanyon2009,10.1145/3307650.3322253,2020arXiv200800959W}.
One might therefore consider optimizing $\dim(\blambda)$ or $\dim(\blambda)/|\lambda|$.
As evidenced by Table~\ref{tab:efficiency}, the physical dimension $d$ that optimizes these quantities
grows (weakly) with $|\lambda|$.

Having indicated that some qudits may be preferable to others,
we should remark how, technologically, such a choice might be physically implemented.
The general idea is to physically constrain the states of the $d$-state systems
so that, in effect, they map by Schur--Weyl duality to a vector in one representation space
and no other.
A general method touched upon in Subsection~4.4.2 of this thesis
is to constrain the physical system in such a way that its lowest energy state corresponds to 
the ``highest weight vector".  This can be used to effectively construct a single-row Young diagram one cell at a time
by ``Cartan product", which can then be added in the same way to the first row of a multi-row Young diagram.
For 2-state systems corresponding to particles with spin, the above constraint can be realized by a magnetic field.
Another approach, for particles with spin, is to prepare pairs of particles into
spin-0 ``singlet" states by various means \cite{divincenzo:qc2000b,HRLa},
which correspond to
vertically stacked pairs of cells in a 2-row Young diagram 
(as can be verified by considering the eigenvalues corresponding to spin for the associated
Weyl tableaux).

\section{Encoded universality}

The existence of a noiseless subsystem encoding a logical qudit as described previously does not guarantee the construction of every 
unitary operator on the logical qudit by exchange-only interactions.  
It happens to be the case that every unitary operator can be so constructed on the logical qubits and logical
qutrit considered above, and for example on the subsystem given by $\ydiagram{4,2,1}$.  However it is not possible to construct every unitary
operator from exchange-only interactions on the following subsystems:
\setydiagrameq
$$
\ydiagram{3,1,1},\ \ydiagram{3,2,1},\ \ydiagram{4,1,1}.
$$
\setydiagramtext
In Section~4.2.1 we derive conditions for an encoding to admit exchange-only unitary operators.
Here we explore what precisely we require of such constructions.

\subsection{Single logical qudit}

Given an encoded logical qudit represented by the subsystem $\blambda$,
we would like to construct any unitary operator in $\SU(\blambda)$ from
exchange-only operations.  By this we mean, operators of the form
$\exp(itH)$, where $t\in\rls$ and $iH\in\g_{|\lambda|}$ may be referred to as a Hamiltonian in the context of time-evolution.  This implies an operational 
definition of universality suggested by Kempe et al., that $\rho_\lambda(\g_{|\lambda|})$
has a subalgebra isomorphic to $\su(\blambda)$.  
Equivalently, we may say $\blambda$ is universal on itself if there exists
a monomorphism from $\su(\blambda)$ into $\rho_\lambda(\g_{|\lambda|})$.

\subsection{Multiple logical qudits}

In the case of multiple logical qudits, defining universality is more 
complicated because there are more choices to be made.  We can find guidance in
the properties of the underlying physical systems on the one hand, and our
computational requirements on the other.  Consider two logical qudits
associated with the representations $\blambda$ and $\bmu$.
Combining them corresponds to taking the tensor product 
$(V_\lambda\otimes U_\lambda)\otimes(V_\mu\otimes U_\mu)$.
In the quantum mechanical description of the corresponding physical system,
this product can be identified with a subspace of the Hilbert space
$\cH=(\cmplx^d)^{\otimes n}$, where $n=|\lambda|+|\mu|$.
Since both the product and $\cH$ are representations of $S_\ell\times S_m\times\SU(d)$, which corresponds to operations on the same $n$ physical particles
in either case, $S_\ell\times S_m\times\SU(d)$ should act ``the same way"
in both representations and the relevant map between them should be the
$S_\ell\times S_m\times\SU(d)$-intertwiner discussed in Subsection~2.3.5.
Then, in accordance with the above, the relevant subspace of $\cH$, to 
which the intertwiner maps, is given
by the terms in its Schur--Weyl decomposition determined by the nonzero
Littlewood--Richardson coefficients: $\bigoplus_{c_{\lambda\mu}^\nu>0}V_\nu\otimes U_\nu\subset\cH$.
As before, the relevant computational spaces are confined to the subsystems
given by the symmetric group representations $V_\nu$.

For the purposes of computation we require a correspondence between each vector of 
a subspace of some $\bnu$ and a unique vector of $\blambda\otimes\bmu$. 
We assume such correspondence is given by a linear map from $\blambda\otimes\bmu$ to $\bnu$.  We further require a correspondence between each element $U$
of $\SU(\blambda\otimes\bmu)$ and an element of $\exp(i\rho_\nu(\g_n))$, 
that transforms 
a vector of $\bnu$ corresponding to some $\psi\in\lambda\otimes\mu$ 
into a vector of $\bnu$ corresponding to $U\psi\in\lambda\otimes\mu$.  
Thus
given a set of relevant partitions $\cP\subset\{\nu|c_{\lambda\mu}^\nu>0\}$,
we require a map
$$
\Phi:\SU(\blambda\otimes\bmu)\rightarrow\exp(i(\bigoplus_{\nu\in\cP}\rho_\nu)(\g_n))
$$
such that for all $U\in\SU(\blambda\otimes\bmu)$ and $\nu\in\cP$, there exists
$\tau,\tau'\in L(\blambda\otimes\bmu,\bnu)$
such that for all 
$\psi\in\blambda\otimes\bmu$, we have
$$
\Phi(U)\tau(\psi)=\tau'(U\psi).
$$

Since, again, $\blambda\otimes\bmu$ and $\bnu$ are both representations of 
$S_{1,\ldots,\ell}\times S_{\ell+1,\ldots,n}$, corresponding to permutations
of $n$ physical particles, physical consistency demands that associated
observables have equal expectation values on both representations.
In particular, for every observable given by a permutation $a\in S_{1,\ldots,\ell}\times S_{\ell+1,\ldots,n}$,
we have
$$
\tr(a U\psi\psi^\dagger U^\dagger)=\tr(a \Phi(U)\tau(\psi)\tau(\psi)^\dagger\Phi(U)^\dagger)=\tr(a \tau'(U\psi)\tau'(U\psi)^\dagger).
$$
From this we deduce that $\tau$ and $\tau'$ must be equivariant
with respect to $S_\ell\times S_m$; that is,
$$
\tau,\tau'\in\Hom_{S_\ell\times S_m}(\blambda\otimes\bmu,\bnu).
$$

Having thus argued for the compatibility of $\Phi$ with certain $S_\ell\times S_m$-intertwiners, the question arises as to which intertwiners.
First there is the question of the $\bnu$ into which the intertwiner
maps.  Physically, distinct $\bnu$ correspond to different spins,
or analogous parameters quantified by eigenvalues of operators in $\SU(d)$.
Meanwhile distinct intertwiners in $\Hom(\blambda\otimes\bmu,\bnu)$
may correspond to different spins of individual qudits.
Subsection~4.4.2 discusses how such parameters might be constrained.
But in general such control may not be achievable, and one needs to allow
for all possibilities, including superposition states.
We are thus led to suggest the following condition for universality:
For all $\nu\in\cP$, $\iota\in\Hom(\blambda\otimes\bmu,\bnu)$, and $U\in\SU(\blambda\otimes\bmu)$, there exists $\iota'\in\Hom(\blambda\otimes\bmu,\bnu)$
such that for all $\psi\in\blambda\otimes\bmu$, we have
$$
\Phi(U)\iota(\psi)=\iota'(U\psi).
$$

For consistency with the definition of Kempe et al. \cite{PhysRevA.63.042307} this condition can be
expressed in Lie-algebraic terms. 
Using the canonical isomorphism, and defining $T_\nu\in\Aut(\Hom(\blambda\otimes\bmu,\bnu))$
such that $T\iota=\iota'$ and $u\in\su(\blambda\otimes\bmu)$ such that
$\exp(iu)=U$, and $t_\nu\in\End(\Hom(\blambda\otimes\bmu,\bnu))$ such that $e^{it_\nu}=T_\nu$, we have
\begin{eqnarray*}
\Phi(U) &=& \bigoplus_{\nu\in\cP}\chi^\nu(U\otimes T_\nu)(\chi^\nu)^{-1}\\
&=& \bigoplus_{\nu\in\cP}\chi^\nu(U\otimes\one)(\one\otimes T_\nu)(\chi^\nu)^{-1}\\
&=& \bigoplus_{\nu\in\cP}\chi^\nu(e^{iu}\otimes\one)(\one\otimes e^{it_\nu})(\chi^\nu)^{-1}\\
&=& \bigoplus_{\nu\in\cP}\chi^\nu(\exp(i(u\otimes\one)))\exp(\one\otimes t_\nu)(\chi^\nu)^{-1}\\
&=& \bigoplus_{\nu\in\cP}\chi^\nu(\exp(i(u\otimes\one+\one\otimes t_\nu)))(\chi^\nu)^{-1}\\
&=& \exp(i\phi(u)),
\end{eqnarray*} 
where $\chi^\nu=\chi_{\lambda\mu}^\nu$ and $\phi(u)=\bigoplus_{\nu\in\cP}\chi^\nu(u\otimes\one+\one\otimes t_\nu)(\chi^\nu)^{-1}$.
Therefore the above condition for universality is equivalent to the statement
that there exists a map $\phi:\su(\blambda\otimes\bmu)\rightarrow\bigoplus_{\nu\in\cP}\rho_\nu(\g_n)$ such that 
for all $\nu\in\cP$, $\iota\in\Hom(\blambda\otimes\bmu,\bnu)$, and $u\in\su(\blambda\otimes\bmu)$, there exists $\iota_0\in\Hom(\blambda\otimes\bmu,\bnu)$
such that for all $\psi\in\blambda\otimes\bmu$, we have
$$
\phi(u)\iota(\psi)=\iota(u\psi)+\iota_0(\psi).
$$
Similar arguments apply to the case of arbitrarily many logical qudits.

\chapter{Universality}

Having introduced and motivated the concept of universality for the purposes of exchange-only computation, here we give it a precise mathematical definition.  Proving universality, 
however, is nontrivial, and we spend much of this chapter deriving conditions for it.  
Following this we show how, in cases where universality fails, simple modifications can make 
it unavoidable.  Along the way we take an important detour to prove the equivalence of two 
alternative definitions of universality.

\section{Defining universality}

As previously discussed, encoded universality on logical qudits labeled by partitions
$\{\mu^{(i)}\}_{i=1}^N$ can be expressed by the existence of a linear injective map from the special unitary
algebra $\su(\bigotimes_{i=1}^N\bmu^{(i)})$ into the image of a representation $\bigoplus_\nu\rho_\nu$ 
of the algebra of transpositions $\g_n$.  We require that the map be compatible
with intertwiners between the representations $\bigotimes_{i=1}^N\bmu^{(i)}$ and $\left(\bigoplus_\nu\rho_\nu\right)(\g_n)$.  In this section we propose two such definitions, followed by some exploration
of equivalences.

For this purpose, it is convenient to define two relevant sets of partitions, 
given a family of partitions $\{\mu^{(i)}\}_{i=1}^N$, namely:
$$
\cP_\mus\equiv\{\nu|\Hom(\bigotimes_{i=1}^N\bmu^{(i)},\bnu)\neq\{0\}\}
$$
and
$$
\cP^{(d)}_\mus\equiv\{\nu|\Hom(\bigotimes_{i=1}^N\bmu^{(i)},\bnu)\neq\{0\},\nup_1\leq d\}.
$$
Then a minimal definition of universality is: 
\begin{Definition}
A family of partitions $\{\mu^{(i)}\}_{i=1}^N$
is 
{\it weakly universal} on a nonempty set 
$\cP\subset\cP_\mus$ if
there exists a map
$$\phi:\su\left(\bigotimes_{i=1}^N\bmu^{(i)}\right)\rightarrow\left(\bigoplus_{\nu\in \cP}\rho_\nu\right)(\g_n)$$
such that for every $\nu\in \cP$, $\iota\in\Hom(\bigotimes_{i=1}^N\bmu^{(i)},\bnu)$, and $u\in\su(\bigotimes_{i=1}^N\bmu^{(i)})$, there exists $\iota_0\in\Hom(\bigotimes_{i=1}^N\bmu^{(i)},\bnu)$ such that for all $\psi\in\bigotimes_{i=1}^N \bmu^{(i)}$ we have
$$\phi(u)\iota(\psi)=\iota(u\psi)+\iota_0(\psi).$$
\end{Definition}
\noindent A simpler and seemingly stronger definition for universality is:
\begin{Definition}
\label{def:strongu}
A family of partitions $\{\mu^{(i)}\}_{i=1}^N$
is {\it strongly universal}, or just
{\it universal}, on a nonempty set 
$\cP\subset\cP_\mus$ if
there exists a map
$$\phi:\su\left(\bigotimes_{i=1}^N\bmu^{(i)}\right)\rightarrow\left(\bigoplus_{\nu\in \cP}\rho_\nu\right)(\g_n)$$
such that for every $\nu\in \cP$, $\iota\in\Hom(\bigotimes_{i=1}^N\bmu^{(i)},\bnu)$, $\psi\in\bigotimes_{i=1}^N \bmu^{(i)}$, and $u\in\su(\bigotimes_{i=1}^N\bmu^{(i)})$, we have
$$\phi(u)\iota(\psi)=\iota(u\psi).$$
If this condition holds then $\cP$ is said to {\it admit universality} of $\{\mu^{(i)}\}_{i=1}^N$.
If $\cP_\mus^{(d)}\subset \cP$ then $\{\mu^{(i)}\}_{i=1}^N$ is 
said to be {\it $d$-universal}.
\end{Definition}
\noindent We show in Section~4.3 that the above two definitions are equivalent, thus justifying our
focus henceforth on the simpler definition of strong universality, and of referring to it simply as
``universality".
Note that universality is invariant with respect to the order of the partitions in the 
family $\{\mu^{(i)}\}_{i=1}^N$, and so could have been defined in terms of a multiset of
partitions instead.  However, for the purpose of proving various results we find it convenient
to order the partitions, albeit arbitrarily.

As discussed in Chapter~3, a desired property of universality is that the projection $\Pi$ of the representation onto the
relevant isotypical subspaces $\bigoplus_{\nu\in \cP}V_{\mu^{(1)}\cdots\mu^{(N)}}^\nu$  contain a subalgebra isomorphic to the relevant special unitary
algebra.  
Therefore the universality-witnessing map should be 
a Lie algebra monomorphism when so projected.
This property could have been specified in the definitions of
universality above 
but there is no need, as it can be inferred from the strong-universality
witnessing map as given:
\begin{lemma}
\label{lem:monomorphism}
Let $\phi$ be a map witnessing the universality of $\{\mu^{(i)}\}_{i=1}^N$ on a nonempty
$\cP\subset \cP_\mus$.  Then the map $\tilde{\phi}$ formed
by projecting $\phi$
onto the
isotypical subspaces $\bigoplus_{\nu\in \cP}V_{\mu^{(1)}\cdots\mu^{(N)}}^\nu$
is a Lie algebra monomorphism.
\end{lemma}
\begin{proof}
Given such a universality-witnessing map $\phi:\su(\bigotimes_{i=1}^N\bmu^{(i)})\rightarrow\left(\bigoplus_{\nu\in \cP}\rho_\nu\right)(\g_n)$, define:
\begin{eqnarray*}
\tilde{\phi}:\su\left(\bigotimes_{i=1}^N\bmu^{(i)}\right) &\rightarrow& \left(\bigoplus_{\nu\in \cP}\rho_\nu\right)(\g_n)\Pi,\\
u &\mapsto& \phi(u)\Pi.
\end{eqnarray*}
Note that $\tilde{\phi}$ is also a witnessing map since the action of the elements in its image
on the isotypical subspaces is unchanged. Thus for every
$\nu\in \cP$, for every $S_{m_1}\times\cdots\times S_{m_N}$-intertwiner $\iota:\bigotimes_{i=1}^N \bmu^{(i)}\rightarrow \bnu$, $\psi$ in $\bigotimes_{i=1}^N \bmu^{(i)}$, and every $u$ in $\su(\bigotimes_{i=1}^N\bmu^{(i)})$, $\tilde{\phi}$ still
satisfies
$\tphi(u)\iota(\psi)=\iota(u\psi)$.
Now we show that $\tphi(u)$ applied to $\iota(\psi)$ is linear in $u$.
For any $u,v\in\su(\bigotimes_{i=1}^N\bmu^{(i)})$ and $a,b\in\cmplx$ we have
\begin{eqnarray*}
\tphi(au+bv)\iota(\psi) &=& \iota((au+vb)\psi)\\
&=& a\iota(u\psi)+b\iota(v\psi)\\
&=& a\tphi(u)\iota(\psi)+b\tphi(v)\iota(\psi)\\
&=& (a\tphi(u)+b\tphi(v))\iota(\psi).
\end{eqnarray*}
Since by the universal property of the tensor product, we have
$$\bigotimes_{i=1}^N \bmu^{(i)}\otimes\Hom\left(\bigotimes_{i=1}^N \bmu^{(i)},\bnu\right)=\Span\left\{\psi\otimes\iota|\iota\in\Hom\left(\bigotimes_{i=1}^N \bmu^{(i)},\bnu\right),\psi\in \bigotimes_{i=1}^N \bmu^{(i)}\right\},$$
application of the canonical
isomorphism $\psi\otimes\iota\mapsto\iota(\psi)$ yields
$$V_{\mu^{(1)}\cdots\mu^{(N)}}^\nu=\Span\left\{\iota(\psi)|\iota\in\Hom\left(\bigotimes_{i=1}^N \bmu^{(i)},\bnu\right),\psi\in \bigotimes_{i=1}^N \bmu^{(i)}\right\}.$$ 
The linearity in $u$ of $\tphi(u)$ applied to $V_{\mu^{(1)}\cdots\mu^{(N)}}^\nu$ then follows from the
linearity in $u$ of $\tphi(u)$ applied to $\iota(\psi)$ proven above.
Next we establish that $\tphi$ 
preservers the Lie bracket.
For each $\iota\in\Hom(\bigotimes_{i=1}^N \bmu^{(i)},\bnu)$
and $\psi\in \bigotimes_{i=1}^N \bmu^{(i)}$, we have
\begin{eqnarray*}
\tphi([u,v])\iota(\psi) &=& \iota([u,v]\psi)\\
&=& \iota(uv\psi)-\iota(vu\psi)\\
&=& \tphi(u)\iota(v\psi)-\tphi(v)\iota(u\psi)\\
&=& \tphi(u)\tphi(v)\iota(\psi)-\tphi(v)\tphi(u)\iota(\psi)\\
&=& [\tphi(u),\tphi(v)]\iota(\psi).
\end{eqnarray*}
Again using the linearity
of the action of $\tphi(u)$
for all $u\in\su\left(\bigotimes_{i=1}^N\bmu^{(i)}\right)$,
we conclude that $\tphi([u,v])=[\tphi(u),\tphi(v)]$ on $\bigoplus_{\nu\in \cP}V_{\mus}^\nu$.
Finally, the homomorphism $\tphi$ must be a monomorphism since by Schur's Lemma every nonzero $\iota$  is
injective:
if $\tphi(u)=\tphi(u')$ and thus for all $\psi\in \bigotimes_{i=1}^N\bmu^{(i)}$,
$\tphi(u)\iota(\psi)=\tphi(u')\iota(\psi)$,
then by definition of $\tphi$, $\iota(u\psi)=\iota(u'\psi)$,
and by injectivity of $\iota$, we have $u\psi=u'\psi$ and thus $u=u'$.
\end{proof}

The nontrivial $\Hom$-set condition in the definitions of
universality can be written in terms of nonvanishing Littlewood--Richardson coefficients:
\begin{lemma}
\label{lem:frobenius}
Given partitions $\mu^{(1)},\ldots,\mu^{(N)}$, and $\nu$ we have
$$\left\{\nu\left|\Hom\left(\bigotimes_{i=1}^N\bmu^{(i)},\bnu\right)\neq\{0\}\right.\right\}=\{\nu|c_{\mu^{(1)}\cdots\mu^{(N)}}^\nu>0\}.$$
\end{lemma}
\begin{proof}
By Lemma~\ref{lem:LRres}, we have
$$\dim(\Hom(\bmu^{(1)}\otimes\cdots\otimes\bmu^{(N)},\bnu))=c_{\mu^{(1)}\cdots\mu^{(N)}}^\nu.$$
Therefore we have
$$\Hom(\bmu^{(1)}\otimes\cdots\otimes\bmu^{(N)},\bnu)\neq\{0\}$$
if and only if
$$c_{\mu^{(1)}\cdots\mu^{(N)}}^\nu>0.$$
\end{proof}

Universality and weak universality can be characterized in terms of the 
canonical isomorphism between each product space $\bigotimes_{i=1}^N\bmu^{(i)}\otimes\Hom(\bigotimes_{i=1}^N\bmu^{(i)},\bnu)$ and the isotypical subspace $V_{\mu^{(1)}\cdots\mu^{(N)}}^\nu$,
as formalized by the next two lemmas.
\begin{lemma}
\label{lem:canonicalweak}
A family of partitions $\{\mu^{(i)}\}_{i=1}^N$ is weakly universal on a non-empty set
$\cP\subset\cP_\mus$ 
if and only if
for all $u\in\su(\bigotimes_{i=1}^N\bmu^{(i)})$ and $\nu\in\cP$ there exists $t_{\nu,u}\in\End(\Hom(\bigotimes_{i=1}^N\bmu^{(i)},\bnu)$ such that
$$\bigoplus_{\nu\in\cP}\chi_{\mu^{(1)}\cdots\mu^{(N)}}^\nu(u\otimes\one+\one\otimes t_{\nu,u})(\chi_{\mu^{(1)}\cdots\mu^{(N)}}^\nu)^{-1}\in\bigoplus_{\nu\in\cP}\rho_\nu(\g_n)\Pi.$$
\end{lemma}
\begin{proof}
First observe that for each $\nu\in \cP$, $\iota\in\Hom(\bigotimes_{i=1}^N\bmu^{(i)},\bnu)$, 
a given 
$$t_{\nu,u}\in\End(\Hom(\bigotimes_{i=1}^N\bmu^{(i)},\bnu),$$
and for all
$\psi\in\bigotimes_{i=1}^N \bmu^{(i)}$ and $u\in\su(\bigotimes_{i=1}^N\bmu^{(i)})$, we have
\begin{align}
&\bigoplus_{\nu\in\cP}\chi_{\mus}^\nu(u\otimes\one+\one\otimes t_{\nu,u})(\chi_{\mus}^\nu)^{-1}\iota(\psi)\\
&= \chi_{\mus}^\nu(u\otimes\one+\one\otimes t_{\nu,u})(\chi_{\mus}^\nu)^{-1}\iota(\psi)\notag\\
&=  \chi_{\mus}^\nu(u\otimes\one+\one\otimes t_{\nu,u})(\chi_{\mus}^\nu)^{-1}\chi_\mus^\nu(\psi\otimes\iota)\notag\\
&= \chi_{\mus}^\nu(u\otimes\one+\one\otimes t_{\nu,u})(\psi\otimes\iota)\notag\\
&= \chi_{\mus}^\nu(u\psi\otimes\iota+\psi\otimes t_{\nu,u}\iota)\notag\\
&= \iota(u\psi)+t_{\nu,u}\iota(\psi).\tag{*}
\end{align}

If $\{\mu^{(i)}\}_{i=1}^N$ is weakly universal on $\cP$,
then there exists a weak-universality-witnessing
function
$$\phi:\su\left(\bigotimes_{i=1}^N\bmu^{(i)}\right)\rightarrow\bigoplus_{\nu\in\cP}\rho_\nu(\g_n)$$
such that $\phi(u)\iota(\psi)=\iota(u\psi)+\iota_0(\psi)$.
Since $\iota(\psi)\in V_{\mu^{(i)}\cdots\mu^{(N)}}^\nu$,
and therefore $\iota(\psi)=\Pi\iota(\psi)$,
we also have
that $\phi(u)\Pi\iota(\psi)=\iota(u\psi)+\iota_0(\psi)$.
Now noting that $\iota$ depends only on $\nu$, $u$, and $\iota$, we can choose $t_{\nu,u}$ such that
$\iota_0=t_{\nu,u}\iota$.
Since furthermore a basis
for $\bigoplus_{\nu\in\cP}V_{\mu^{(i)}\cdots\mu^{(N)}}^\nu$ can be constructed in terms of vectors of the form $\iota(\psi)$,
it follows from Equation~(*) 
that
$$\bigoplus_{\nu\in\cP}\chi_{\mus}^\nu(u\otimes\one+\one\otimes t_{\nu,u})(\chi_{\mus}^\nu)^{-1}=\phi(u)\Pi$$
and thus $\bigoplus_{\nu\in\cP}\chi_{\mus}^\nu(u\otimes\one+\one\otimes t_{\nu,u})(\chi_{\mus}^\nu)^{-1}\in\left(\bigoplus_{\nu\in\cP}\rho_\nu\right)(\g_n)\Pi$.

Conversely,
if for all $u\in\su(\bigotimes_{i=1}^N\bmu^{(i)})$, we have
$$\bigoplus_{\nu\in\cP}\chi_{\mus}^\nu(u\otimes\one+\one\otimes t_{\nu,u})(\chi_{\mus}^\nu)^{-1}\in\bigoplus_{\nu\in\cP}\rho_\nu(\g_n)\Pi,$$ 
then there exists $x_u\in\g_n$ such that
$$\left(\bigoplus_{\nu\in\cP}\rho_\nu\right)(x_u)\Pi=\bigoplus_{\nu\in\cP}\chi_{\mus}^\nu(u\otimes\one+\one\otimes t_{\nu,u})(\chi_{\mus}^\nu)^{-1}.$$
Define
\begin{eqnarray*}
\phi:\su\left(\bigotimes_{i=1}^N\bmu^{(i)}\right) &\rightarrow& \left(\bigoplus_{\nu\in\cP}\rho_\nu\right)(\g_n),\\
u &\mapsto& \left(\bigoplus_{\nu\in\cP}\rho_\nu\right)(x_u).
\end{eqnarray*}
Then for all $\iota\in\Hom(\bigotimes_{i=1}^N\bmu^{(i)},\bnu)$ and $\psi\in\bigotimes_{i=1}^N \bmu^{(i)}$, we have
\begin{eqnarray*}
\phi(u)\iota(\psi) &=&  \left(\bigoplus_{\nu\in\cP}\rho_\nu\right)(x_u)\iota(\psi)\\
&=&  \left(\bigoplus_{\nu\in\cP}\rho_\nu\right)(x_u)\Pi\iota(\psi)\\
&=& \bigoplus_{\nu\in\cP}\chi_{\mus}^\nu(u\otimes\one+\one\otimes t_{\nu,u})(\chi_{\mus}^\nu)^{-1}\iota(\psi)\\
&=& \iota(u\psi)+\iota_0(\psi),
\end{eqnarray*}
where $\iota_0=t_{\nu,u}\iota$.
Thus $\{\mu^{(i)}\}_{i=1}^N$ is weakly universal on $\cP$.
\end{proof}

\begin{lemma}
\label{lem:canonicalstrong}
A family of partitions $\{\mu^{(i)}\}_{i=1}^N$ is universal on a set
$$\cP\subset\cP_\mus$$
if and only if
for all $u\in\su(\bigotimes_{i=1}^N\bmu^{(i)})$, we have $\bigoplus_{\nu\in\cP}\chi_{\mus}^\nu(u\otimes\one)(\chi_{\mus}^\nu)^{-1}\in\bigoplus_{\nu\in\cP}\rho_\nu(\g_n)\Pi$,
where $\Pi$ is the projection onto $\bigoplus_{\nu\in\cP}V_{\mu^{(i)}\cdots\mu^{(N)}}^\nu$.
\end{lemma}
\begin{proof} 
The proof proceeds identically to that of Lemma~\ref{lem:canonicalweak} but in the special case of
$\iota_{0}=0$ and $t_{\nu,u}=0$.
\end{proof}

\section{Proving universality}

In this section we derive various conditions for
the universality of a family of nontrivial partitions.
Necessary and sufficient conditions for the universality of a singleton comes very readily from Marin's work, but 
the case of multiple partitions is more challenging.  The latter requires taking a close 
look at certain problematical partitions, already identified as having special significance 
in Marin's analysis -- namely hooks and conjugates.  The problem with hooks and 
self-conjugate partitions is that their associated representations may be too ``small" to 
contain the needed special unitary algebras.  The problem with conjugate pairs is that 
the representation associated with one of each pair determines the other, which may be too 
constraining for universality.  We verify below that indeed hooks generally break universality.
On the other hand, we derive precise conditions under which conjugates
admit universality.

\subsection{Some general conditions}

First we derive various sufficient conditions for universality which
have wide applicability.  The first two are straightforward implications of the 
universality definitions and their properties introduced in the last section.
With these tools in hand, the remaining results of this subsection are essentially
corollaries of Marin's work on the algebra of transpositions.

We begin with a condition which, while not proved necessary for universality, is satisfied in every case
of universality considered in this thesis:
\begin{lemma}
\label{lem:suVsubthenu}
Given a family of partitions $\{\mu^{(i)}\}_{i=1}^N$ and nonempty $\cP\subset\cP_\mus$, suppose that
$\bigoplus_{\nu\in \cP}\su\left(V_\mus^\nu\right)$ is a Lie subalgebra of $\left(\bigoplus_{\nu\in \cP}\rho_\nu\right)(\g_n)\Pi$.
Then $\{\mu^{(i)}\}_{i=1}^N$ is universal on $\cP$.
\end{lemma}
\begin{proof}
We proceed by constructing a map witnessing the universality of $\{\mu^{(i)}\}_{i=1}^N$
as a direct sum over $\nu$ of maps to each $\rho_\nu(\g_n)$.
For each $\nu\in \cP$, define the map $\varphi_\nu:\su(\bigotimes_{i=1}^N\bmu^{(i)})\rightarrow\su\left(\bigotimes_{i=1}^N\bmu^{(i)}\otimes\Hom(\bigotimes_{i=1}^N\bmu^{(i)},\bnu)\right)$ 
by $\varphi_\nu(u)=u\otimes \one$.
Recalling the canonical isomorphism (Lemma~\ref{lem:canonicaliso}) 
\begin{eqnarray*}
\chi_\mus^\nu:\bigotimes_{i=1}^N\bmu^{(i)}\otimes\Hom(\bigotimes_{i=1}^N\bmu^{(i)},\bnu)&\rightarrow&V_\mus^\nu,\\
\psi\otimes\iota&\mapsto&\iota(\psi),
\end{eqnarray*}
define
$
\phi_\nu:\su(\bigotimes_{i=1}^N\bmu^{(i)}) \rightarrow \su\left(V_\mus^\nu\right)
$
by $\phi_\nu(u)=\chi_\mus^\nu\varphi_\nu(u)(\chi_\mus^\nu)^{-1}$.
Since by hypothesis
$\bigoplus_{\nu\in \cP}\su\left(V_\mus^\nu\right)\subset\left(\bigoplus_{\nu\in \cP}\rho_\nu\right)(\g_n)\Pi$,
we can further define
\begin{eqnarray*}
\tilde{\phi}:\su(\bigotimes_{i=1}^N\bmu^{(i)}) &\rightarrow& \left(\bigoplus_{\nu\in \cP}\rho_\nu\right)(\g_n)\Pi,\\
u &\mapsto& \bigoplus_{\nu\in \cP}\phi_\nu(u).
\end{eqnarray*}
Then we have $\tilde{\phi}(u)\iota(\psi)=\phi_\nu(u)\iota(\psi)=\iota(u\psi)$.
The existence of this map implies that for each $u\in\su(\bigotimes_{i=1}^N\bmu^{(i)})$, there exists $x_u\in\g_n$ such that
$(\bigoplus_{\nu\in \cP}\rho_\nu)(x_u)\Pi=\tilde{\phi}(u)$.
Therefore we can define
\begin{eqnarray*}
\phi:\su(\bigotimes_{i=1}^N\bmu^{(i)}) &\rightarrow& \left(\bigoplus_{\nu\in \cP}\rho_\nu\right)(\g_n)\\
u &\mapsto& \left(\bigoplus_{\nu\in \cP}\rho_\nu\right)(x_u)
\end{eqnarray*}
so that $\phi(u)\Pi=\tilde{\phi}(u)$.
Then noting that $\iota(\psi)=\Pi\iota(\psi)$, we have
\begin{eqnarray*}
\phi(u)\iota(\psi) &=& \phi(u)\Pi\iota(\psi)\\
&=& \tilde{\phi}(u)\iota(\psi)\\
&=& \iota(u\psi).
\end{eqnarray*}
Thus $\phi$ witnesses the universality of $\{\mu^{(i)}\}_{i=1}^N$ on $\cP$.
\end{proof}

Among the simplest cases to prove universal on a set $\cP$ are those for which $\cP$ 
excludes hooks, self-conjugate partitions and conjugate pairs.  Then,
according to Marin's Theorem,
the associated sum over $\nu$ of representations of $\g'_n$ equals
the corresponding sum over $\nu$ of $\sla(\bnu)$,
thus allowing plenty of room, so to speak, for the desired unitary algebra.
The following two lemmas formalize this idea.
\begin{lemma}
\label{lem:manysumuniversal}
Given nonempty $\cP\subset\cP_\mus$,
suppose that
$$\left(\bigoplus_{\nu\in P}\rho_\nu\right)(\g'_n)=\bigoplus_{\nu\in P}\sla(\bnu).$$
Then $\{\mu^{(i)}\}_{i=1}^N$ is universal on $P$.
\end{lemma}
\begin{proof}
For each $\nu\in\cP$, $V_{\mu^{(1)}\cdots\mu^{(N)}}^\nu$ is a subspace of $\bnu$. 
It follows that $\su(V_{\mu^{(1)}\cdots\mu^{(N)}}^\nu)$ is a subalgebra of
$\su(\bnu)$, and thus of $\sla(\bnu)$. 
Taking the direct sum, $\bigoplus_{\nu\in\cP}\su(V_{\mu^{(1)}\cdots\mu^{(N)}}^\nu)$
is a subalgebra of $\bigoplus_{\nu\in\cP}\sla(\bnu)$.
Then by hypothesis, we have $\bigoplus_{\nu\in\cP}\su(V_{\mu^{(1)}\cdots\mu^{(N)}}^\nu)\subset\left(\bigoplus_{\nu\in\cP}\rho_\nu\right)(\g'_n)$.  Therefore
by Lemma~\ref{lem:suVsubthenu}, $\{\mu^{(i)}\}_{i=1}^N$ is universal on $\cP$.
\end{proof}

\begin{lemma}
\label{lem:ifnohooksnorconjthenu}
Given a family of partitions $\{\mu^{(i)}\}_{i=1}^N$,
if
a nonempty
$P\subset\cP_\mus$
contains no hooks, self-conjugate partitions, nor conjugate pairs,
then $\{\mu^{(i)}\}_{i=1}^N$ is universal on $P$.
\end{lemma}
\begin{proof}
By Corollary~\ref{cor:marin} we have
$$\left(\bigoplus_{\nu\in P}\rho_\nu\right)(\g'_n)=\bigoplus_{\nu\in P}(\rho_\nu(\g'_n)),$$
and by Lemma~\ref{lem:marin} we have
$$\left(\bigoplus_{\nu\in P}\rho_\nu\right)(\g'_n)=\bigoplus_{\nu\in P}\sla(\bnu).$$
Therefore by Lemma~\ref{lem:manysumuniversal} $\{\mu^{(i)}\}_{i=1}^N$ is universal on $P$.
\end{proof}

The next lemma is a step towards generalizing the above idea.  
Suppose we have proved universality on hooks, conjugate pairs, or self-conjugates,
and wish to show universality on a larger set containing these.  The following
implication of Marin's Theorem indicates when universality is monotone:
\begin{lemma}
\label{lem:additiveu}
Given a family of partitions $\{\mu^{(i)}\}_{i=1}^N$, and a subset $\cP\subset\cP_\mus$, 
if $\{\mu^{(i)}\}_{i=1}^N$ is universal on the subset 
of hooks of $\cP$, and on each pair of conjugate proper partitions in $\cP$, 
and on each remaning singleton in $\cP$, 
then $\{\mu^{(i)}\}_{i=1}^N$ is universal on $\cP$.
\end{lemma}
\begin{proof}
Let $\Gamma=\{\nu\in\cP|\nu_2=1\}$,
$\cQ_{\neq}=\{\nu\in\cP\backslash\Gamma|\nu'\in\cP,\nu\neq\nu'\}$,
$\cQ_<=\{\nu\in\cQ_{\neq}|\nu'<\nu\}$,
and $\cR=\cP\backslash(\Gamma\cup\cQ_{\neq})$.
Universality on $\Gamma$ implies the existence of a witnessing function
$\phi_\Gamma:\su(\bigotimes_{i=1}^N\bmu^{(i)})\rightarrow\left(\bigoplus_{\nu\in\Gamma}\rho_\nu\right)(\g_n)$, universality on each conjugate pair implies for each $\nu\in\cQ_<$ the existence
of a witnessing function $\phi_{\nu\nu'}:\su(\bigotimes_{i=1}^N\bmu^{(i)})\rightarrow(\rho_\nu\oplus\rho_{\nu'})(\g_n)$, and universality on each remaining partition
implies for each $\nu\in\cR$ the existence of a witnessing function $\phi_\nu:\su(\bigotimes_{i=1}^N\bmu^{(i)})\rightarrow\rho_\nu(\g_n)$.
We wish to define a ``total" witnessing function.
This is made possible by
Corollary~\ref{cor:marin2}, which states
$$
\left(\bigoplus_{\nu\in\Gamma}\rho_\nu\right)(\g_n)\oplus\bigoplus_{\nu\in\cQ_<}(\rho_\nu\oplus\rho_{\nu'})(\g_n)\oplus\bigoplus_{\nu\in\cR}\rho_\nu(\g_n)=\left(\bigoplus_{\nu\in\cP}\rho_\nu\right)(\g_n).
$$
Therefore we can define the sum of the above witnessing functions in terms of the above equation 
to obtain
\begin{eqnarray*}
\phi:\su(\bigotimes_{i=1}^N\bmu^{(i)}) &\rightarrow& \left(\bigoplus_{\nu\in\cP}\rho_\nu\right)(\g_n)\\
u &\mapsto& \phi_\Gamma(u)\oplus\bigoplus_{\nu\in\cQ_<}\phi_{\nu\nu'}(u)\oplus\bigoplus_{\nu\in\cR}\phi_\nu(u).
\end{eqnarray*}
Then given $\nu\in\cP$, $\iota\in\Hom(\bigotimes_{i=1}^N\bmu^{(i)},\bnu)$, and $\psi\in\bigotimes_{i=1}^N\bmu^{(i)}$, 
we have $\phi(u)\iota(\psi)=\varphi(u)\iota(\psi)$, where $\varphi$ is one of $\phi_\Gamma$, 
$\phi_{\nu\nu'}$, or $\phi_\nu$ depending on whether $\nu$ is in $\Gamma$, $\cQ_{\neq}$, or $\cR$,
respectively.  Because each of the latter is a universality-witnessing function,
by definition we have $\varphi(u)\iota(\psi)=\iota(u\psi)$.
Thus $\phi$ witnesses the universality of $\{\mu^{(i)}\}_{i=1}^N$ on $\cP$.
\end{proof}

The following two theorems are important results in their own right, 
for their direct applicability
to quantum information.
The first, a consequence of
Marin's Lemma~\ref{lem:marin}, gives necessary and sufficient conditions for encoded universality on a 
single logical qudit.
\begin{theorem}
\label{thm:iff1partitionu}
A partition $\mu\vdash m$ is universal on itself if and only if 
one of the following holds:
\begin{enumerate}[label=(\roman{*})]
\item $\mu_2>1$ and $\mu\neq\mu'$,
\item $\mu=\ydiagram{2,2}$,
\item $\mu\in\{[m],[1^m],[m-1,1],[2,1^{m-2}]\}$.
\end{enumerate}
\end{theorem}
\begin{proof}
By Lemma~\ref{lem:marin}, we have $\rho_\mu(\g'_m)=\sla(\bmu)$.
Therefore by Lemma~\ref{lem:manysumuniversal}, $\mu$ is universal on $\mu$.
\end{proof}

The final theorem of this subsection states that every family of partitions
of two rows is 2-universal.  An equivalent statement was 
originally proved by Kempe et al. in \cite{PhysRevA.63.042307}.  It is also a near-immediate
consequence of Marin's work, particularly Lemma~\ref{lem:marindeq2}.
For the sake of thoroughness, we state and prove this result in the context of the
present formalism:
\begin{theorem}[Kempe--Bacon--Lidar--Whaley]
\label{thm:kempe}
Let $d\in\{1,2\}$ and let $\{\mu^{(i)}\}_{i=1}^N$ be a family of partitions
such that $\mu^{(i)\prime}_{\ 1}\leq d$.  Then $\{\mu^{(i)}\}_{i=1}^N$ is 
universal on any nonempty set $\cP\subset\cP_\mus^{(d)}$ and in particular, is d-universal.
\end{theorem}
\begin{proof}
If $d=1$ then $\su(\bigotimes_{i=1}^N\bmu^{(i)})=0$ and universality is 
trivially satisfied.  If $d=2$, it is immediate from the definitions 
that there are no nontrivial conjugate pairs $\nu$ and $\nu'$ in $\cP$ such that
$\nu'\neq\nu$ and $\nup_1\leq 2$, and there is at most one hook $\nu$ in $\cP$ such that 
$\nup_1\leq 2$.  Therefore by Corollary~\ref{cor:marin} we have 
$$
\left(\bigoplus_{\nu\in\cP}\rho_\nu\right)(\g'_{|\nu|})=\bigoplus_{\nu\in\cP}\rho_\nu(\g'_{|\nu|}).
$$
It furthermore follows from Lemma~\ref{lem:marindeq2} that 
$$
\left(\bigoplus_{\nu\in\cP}\rho_\nu\right)(\g'_{|\nu|})=\bigoplus_{\nu\in\cP}\sla(\bnu).
$$
Therefore by Lemma~\ref{lem:manysumuniversal}, $\{\mu^{(i)}\}_{i=1}^N$
is universal on $\cP$.
\end{proof}

The remainder of Section~4.2 may be considered an endeavor to generalize Theorem~\ref{thm:kempe}
for $d>2$, or alternatively to generalize Theorem~\ref{thm:iff1partitionu} for more than one partition. 

\subsection{Conjugates}

We have seen in the previous subsection 
that a set
$\cP\subset\cP_\mus$ that excludes hooks and conjugates
always admits universality of $\{\mu^{(i)}\}_{i=1}^N$.
Conjugates, however, generally do not.
Conjugate pairs may not admit universality because $\rho_\nu(\g_n)$ and
$\rho_{\nu'}(\g_n)$ depend on each other in such a way that they may not
simultaneously satisfy the intertwiner condition of Definition~\ref{def:strongu}. 
As for self-conjugate $\nu\in\cP$, in this case $\rho_\nu(\g_n)$ is typically a proper 
subalgebra of $\sla(\bnu)$ that may not admit a satisfactory image of 
$\su(\bigotimes_{i=1}^N\bmu^{(i)})$.

The key to finding cases in which conjugates admit universality is to consider how the alternating intertwiner $M_\nu$ between $\nu$ and $\nu'$ maps isotypical subspaces.
Our strategy for constructing the needed unitary operators on 
$V_\mus^\nu$ and on $V_\mus^{\nu'}$ is to use the 
spaces $V_\musp^\nu$ and $V_\musp^{\nu'}$ as scratch space,
in a sense to be made clear after we establish the relationship between these 
isotypical
subspaces imposed by $M_\nu$.
To proceed with the analysis and track how $M_\nu$ maps a given vector, we show
that the local Jucys--Murphy elements can identify the isotypical subspace to which
a given vector belongs:

\begin{lemma}
\label{lem:localJMiso}
Given a family of natural numbers $\{m_i\}_{i=1}^N$, a partition $\nu\vdash\sum_{i=1}^Nm_i$,
and a simultaneous eigenvector $e$ of the local Jucys--Murphy elements in $\cmplx\prod_{i=1}^N S_{m_i}$, 
the corresponding
local Jucys--Murphy eigenvalues uniquely determine the 
$\prod_{i=1}^N S_{m_i}$-isotypical subspace of $\nu$ that contains $e$.
\end{lemma}
\begin{proof}
Given particular $\{\mu^{(i)}\vdash m_i\}_{i=1}^N$, each simultaneous eigenvector in
$\bigotimes_{i=1}^N\bmu^{(i)}$ of the local Jucys--Murphy elements
corresponds to a product of tableaux, one from each $\mu^{(i)}$, by property
of the Jucys--Murphy elements.  Given such a product, the $\bigotimes_{i=1}^N\blambda^{(i)}$ space to which it belongs
is uniquely determined to be $\bigotimes_{i=1}^N\bmu^{(i)}$  by the shapes of the tableaux. 
In turn, each such product of tableaux is
uniquely determined by its local Jucys--Murphy eigenvalues.  Therefore, given a simultaneous eigenvector
in $\bigoplus_{\mu^{(1)}\vdash m_1,\ldots,\mu^{(N)}\vdash m_N}\bigotimes_{i=1}^N\bmu^{(i)}\otimes\Hom(\bigotimes_{i=1}^N\bmu^{(i)},\bnu)$
of the local Jucys--Murphy elements, its local Jucys--Murphy eigenvalues uniquely determine the
$\bigotimes_{i=1}^N\bmu^{(i)}\otimes\Hom(\bigotimes_{i=1}^N\bmu^{(i)},\bnu)$ subspace to which it belongs.

Recall that $\nu$ is partitioned into $\prod_{i=1}^NS_{m_i}$-isotypical subspaces:
$$
\nu=\bigoplus_{\mu^{(1)}\vdash m_1,\ldots,\mu^{(N)}\vdash m_N}V_\mus^\nu,
$$
onto which there exist canonical isomorphisms:
$$
\bigoplus_{\mu^{(1)},\ldots,\mu^{(N)}}\chi_{\mu^{(1)}\cdots\mu^{(N)}}^\nu:\bigoplus_{\mu^{(1)},\ldots,\mu^{(N)}}\bigotimes_{i=1}^N\bmu^{(i)}\otimes\Hom(\bigotimes_{i=1}^N\bmu^{(i)},\bnu)\rightarrow\bigoplus_{\mu^{(1)},\ldots,\mu^{(N)}}V_{\mu^{(1)}\cdots\mu^{(N)}}^\nu
$$
where $\mu^{(i)}\vdash m_i$.
Since 
$\bigoplus_{\mu^{(1)}\vdash m_1,\ldots,\mu^{(N)}\vdash m_N}\chi_{\mu^{(1)}\cdots\mu^{(N)}}^\nu$
is an $\prod_{i=1}^NS_{m_i}$-module isomorphism, isomorphic Jucys--Murphy eigenvectors
must have equal eigenvalues.  It follows that those eigenvalues also uniquely determine the
$\prod_{i=1}^NS_{m_i}$-isotypical subspace containing a given simultaneous eigenvector of the local
Jucys--Murphy elements.

\end{proof}

We are prepared to prove the crucial fact that $M_\nu V_\mus^\nu=V_\musp^{\nu'}$:
\begin{lemma}
\label{lem:MnuV}
Given a family of partitions $\{\mu^{(i)}\}_{i=1}^N$, and $\nu\vdash\sum_{i=1}^N|\mu^{(i)}|$, the alternating interwiner $M_\nu$ between $\nu$
and $\nu'$ also restricts to an isomorphism between the
$\bigotimes_{i=1}^N\bmu^{(i)}$-isotypical subspace $V_\mus^\nu$ of $\nu$
and the $\bigotimes_{i=1}^N\bmu^{(i)\prime}$-isotypical subspace $V_\musp^{\nu'}$
of $\nu'$.
\end{lemma}
\begin{proof}
Given any local Jucys--Murphy element $X$, we have $M_\nu^{-1} X M_\nu=-X$, by definition
of $M_\nu$ and the odd parity of $X$.  It follows that, for each
simultaneous eigenvector $e$ in $\nu$ of the local \JM elements, $M_\nu e$ is
also an eigenvector of the local \JM elements with eigenvalues being the
negation of those of $e$.
We also have that $e=\iota(\psi)$ for some $\iota\in\Hom(\bigotimes_{i=1}^N\bmu^{(i)},\bnu)$ and $\psi\in\bigotimes_{i=1}^N\bmu^{(i)}$. Then by equivariance of the intertwiner with respect to
the local {\JM} elements, $\psi$ must have the same local \JM
eigenvalues as $e$.
Then because conjugating tableaux in each $\mu^{(i)}$ negates their
\JM eigenvalues, and the products of these tableaux form the basis supporting
$\psi$, there exists $\psi'\in\bigotimes_{i=1}^N\bmu^{(i)\prime}$ with \JM eigenvalues
being the negation of those of $\psi$.  Thus for an intertwiner
$\iota'\in\Hom(\bigotimes_{i=1}^N\bmu^{(i)\prime},\bnu')$, $\iota'(\psi')$ in $V_\musp^{\nu'}$ has the same eigenvalues as that of $M_\nu e$.  By Lemma~\ref{lem:localJMiso} it follows that
$M_\nu e$ is in $V_{\lambda'\mu'}^{\nu'}$.
Thus we have $M_\nu V_\mus^\nu\subset V_\musp^{\nu'}$.
Further since $M_\nu $ is injective and $\dim(V_\mus^\nu)=\dim(V_\musp^{\nu'})$, we conclude that $M_\nu V_\mus^\nu=V_\musp^{\nu'}$.
\end{proof}

\begin{figure}
\begin{center}
\begin{tikzpicture}
\matrix (m) [
    matrix of nodes,
    row sep=0.4cm,
    column sep=1cm,
    every node/.style={text width=4.2cm, text depth=0.5cm},
    every node/.append code={
       \ifnum\pgfmatrixcurrentrow=1\relax
          \pgfkeysalso{align=center}
       \else
           \pgfkeysalso{align=center,draw}
       \fi
    }
]{
{$\nu\phantom{'}=\nu'$} \\
$V_{\mu^{(1)}\cdots\mu^{(N)}}^\nu=V_{\mu^{(1)}\cdots\mu^{(N)}}^{\nu'}$ \\
$V_{\mu^{(1)\prime}\cdots\mu^{(N)\prime}}^\nu=V_{\mu^{(1)\prime}\cdots\mu^{(N)\prime}}^{\nu'}$  \\
};
\draw[->] (m-2-1) -- (m-3-1);
\draw[->] (m-3-1) -- (m-2-1);
\node[rectangle,draw,fit=(m-1-1)(m-3-1)]{};
\end{tikzpicture}
\end{center}
\caption{Schematic diagram of $M_\nu$ mapping from isotypical subspaces of $\nu=\nu'$ to themselves.} 
\label{fig:muneqmupnueqnup}
\end{figure}
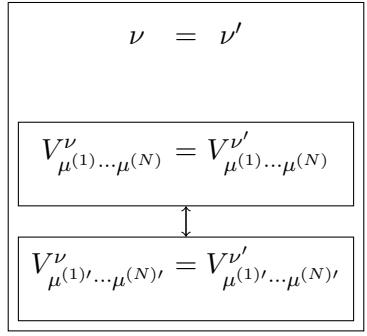

In the case of $\nu=\nu'$, $M_\nu$ is a vector space automorphism on $\nu$ and in particular
on the internal direct sum $V_\mus^\nu\oplus V_\musp^\nu$. 
Then $\rho_\nu(x)$ is constrained by Lemma~\ref{lem:isoMorphism} to be $M_\nu$-invariant,
which may prohibit some unitary operators on $\bnu$ and in particular on $V_\mus^\nu$. 
But if $\{\mu^{(i)}\}_{i=1}^N\neq\{\mu^{(i)\prime}\}_{i=1}^N$  
then we are saved because $M_\nu$ does not map
$V_\mus^\nu$ to itself (Figure~\ref{fig:muneqmupnueqnup}).  Instead, $M_\nu$ maps vectors in $V_\musp^\nu$ to $V_\mus^\nu$ and vice versa.  Again, the action on $V_\musp^\nu$ is for our purposes arbitrary.
This proves to be enough freedom to generate any desired unitary operator on $V_\mus^\nu$. 
\begin{lemma}
\label{lem:noncononselfcon}
Given a family of partitions $\{\mu^{(i)}\vdash m_i\}_{i=1}^N$, 
and proper partition $\nu\vdash \sum_{i=1}^Nm_i$
such that $\{\mu^{(i)}\}_{i=1}^N\neq\{\mu^{(i)\prime}\}_{i=1}^N$,
$\nu=\nu'$, and $\Hom(\bigotimes_{i=1}^N\bmu^{(i)},\bnu)\neq 0$,
we have that $\{\mu^{(i)}\}_{i=1}^N$ is universal on $\nu$.
\end{lemma}
\begin{proof}
By Lemma~\ref{lem:MnuV}, we have $M_\nu V_\mus^\nu=V_\musp^\nu$
and $M_\nu V_\musp^\nu=V_\mus^\nu$,
so $M_\nu $ fixes the internal direct sum $V_\mus^\nu\oplus V_\musp^\nu$. 
Furthermore by the orthogonality of distinct $\prod_{i=1}^NS_{m_i}$-isotypical
subspaces (Lemma~\ref{lem:isotypicalorthogonality}), $M_\nu V_\mus^\nu$ is orthogonal to $V_\mus^\nu$
and $M_\nu V_\musp^\nu$ is orthogonal to $V_\musp^\nu$.
Therefore on $V_\mus^\nu\oplus V_\musp^\nu$ we can write
$$
M_\nu =\left(\begin{array}{c|c} 0 & M_{12} \\\hline M_{21} & 0 \end{array}\right)
$$
where by properties of the alternating intertwiner, $M^\tp_{21} = M^{-1}_{21} = \pm M_{12}$,
with the sign depending on $\nu$. 
Let $u\in\su(V_\mus^\nu)$. 
Our goal is to establish that $u$ is in $\rho_\nu(\g_n)\Pi$.
Recalling that membership of any $X$ in $\rho_\nu(\g_n)$ requires
$M_\nu XM_\nu^{-1}=-X^\tp$, we wish to solve the following equation 
for $u'\in\su(V_\musp^\nu)$:
$$
\left(\begin{array}{c|c} 0 & M_{12} \\\hline M_{21} & 0 \end{array}\right)
\left(\begin{array}{c|c} u & 0 \\\hline 0 & u' \end{array}\right)
\left(\begin{array}{c|c} 0 & M^{-1}_{21} \\\hline M^{-1}_{12} & 0 \end{array}\right)
=-\left(\begin{array}{c|c} u^\tp & 0 \\\hline 0 & (u')^\tp \end{array}\right).
$$
This reduces to the system
\begin{eqnarray*}
M_{12}u'M^{-1}_{12} &=& -u^\tp \\
M_{21}uM^{-1}_{21} &=& -(u')^\tp,
\end{eqnarray*}
which admits the solution
$u'=-M^{-1}_{12}u^\tp M_{12}$.
Finally noting that $(u\oplus u')\Pi=u\oplus 0$,
we conclude that $\su(V_\mus^\nu)$ is a subalgebra of $\rho_\nu(\g_n)\Pi$
and, by Lemma~\ref{lem:suVsubthenu}, we have that $\{\mu^{(i)}\}_{i=1}^N$ is universal on $\nu$.
\end{proof}

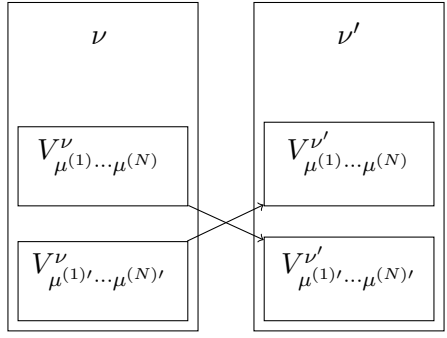
\begin{figure}
\begin{center}
\begin{tikzpicture}
\matrix (m) [
    matrix of nodes,
    row sep=0.4cm,
    column sep=1cm,
    every node/.style={text width=2cm, text depth=0.5cm},
    every node/.append code={
       \ifnum\pgfmatrixcurrentrow=1\relax
          \pgfkeysalso{align=center}
       \else
           \pgfkeysalso{align=center,draw}
       \fi
    }
]{
{$\nu\phantom{'}$} &
{$\nu'$} \\
$V_{\mu^{(1)}\cdots\mu^{(N)}}^\nu\phantom{'}$ &
$V_{\mu^{(1)}\cdots\mu^{(N)}}^{\nu'}$ \\
$V_{\mu^{(1)\prime}\cdots\mu^{(N)\prime}}^\nu\phantom{'}$ &
$V_{\mu^{(1)\prime}\cdots\mu^{(N)\prime}}^{\nu'}$ \\
};
\draw[->] (m-2-1) -- (m-3-2);
\draw[->] (m-3-1) -- (m-2-2);
\node[rectangle,draw,fit=(m-1-1)(m-3-1)]{};
\node[rectangle,draw,fit=(m-1-2)(m-3-2)]{};
\end{tikzpicture}
\end{center}
\caption{Schematic diagram of $M_\nu$ mapping from isotypical subspaces of $\nu$ to isotypical subspaces of $\nu'$.}
\label{fig:muneqmupnuneqnup}
\end{figure}

In the case of proper $\nu\neq\nu'$, since $\rho_\nu(\g_n)=\sla(\bnu)$, we can find
an $x\in\g_n$ to give any desired operator on $\nu$.  However, the resulting action on 
$\nu'$ is then constrained by $M_\nu$.  Fortunately we are not concerned with the action on
all of $\nu$ and $\nu'$, rather for universality of $\{\mu^{(i)}\}_{i=1}^N$ we are
only concerned with the action on the isotypical subspaces $V_\mus^\nu$
and $V_\mus^{\nu'}$.
Furthermore the action on these subspaces are not directly dependent on each other,
provided that $\{\mu^{(i)}\}_{i=1}^N\neq\{\mu^{(i)\prime}\}_{i=1}^N$ (Figure~\ref{fig:muneqmupnuneqnup}).
Instead, as proven above, the action on $V_\mus^{\nu'}$ is determined, via
$M_\nu$, by the action on $V_\musp^\nu$, and the action on 
$V_\mus^\nu$ is determined, via $M_{\nu'}=M_\nu^{-1}$, by the action on $V_\musp^{\nu'}$.
The actions on $V_\musp^\nu$ and on $V_\musp^{\nu'}$ are,
for our purposes, arbitrary, since they are unconstrained by the requirements of universality. 
Therefore, since we can find an $x\in\g_n$ to give any desired action on $\nu$, we can
in particular find an $x$ to give the desired action on $V_\mus^\nu$ while
simultaneously acting on $V_\musp^\nu$ in such a way that its transformation
by $M_\nu$ gives the desired action on $V_\mus^{\nu'}$:
\begin{lemma}
\label{lem:ifnoncononpaircon}
Given a family of partitions $\{\mu^{(i)}\}_{i=1}^N$ and a proper partition $\nu\in\cP_\mus$
such that $\{\mu^{(i)\prime}\}_{i=1}^N\neq\{\mu^{(i)}\}_{i=1}^N$, 
$\nu'\neq\nu$, and $\nu'\in\cP_\mus$, we have that $\{\mu^{(i)}\}_{i=1}^N$ is universal on $\{\nu,\nu'\}$.
\end{lemma}
\begin{proof}
By Lemma~\ref{lem:marin}, we have $\rho_\nu(\g_n)=\sla(\bnu)$.  Since $V_\mus^\nu\oplus V_\musp^\nu\subset\bnu$, 
it follows that $\su(V_\mus^\nu)\oplus\su(V_\musp^\nu)\oplus\bzero\subset\rho_\nu(\g_n)$ 
where $\bzero$ acts on the orthogonal complement of $V_\mus^\nu\oplus V_\musp^\nu$ in $\bnu$.
Thus for each $u\oplus u'\in\su(V_\mus^\nu)\oplus\su(V_\musp^\nu)$, there exists $x\in\g_n$
such that $\rho_\nu(x)=u\oplus u'\oplus 0$.
Then by Lemma~\ref{lem:isoMorphism}, we have
\begin{eqnarray*}
\rho_{\nu'}(x) &=& M_\nu^{-1}(u\oplus u'\oplus 0)M_\nu\\
&=& -u^\tp\oplus(u')^\tp\oplus 0.
\end{eqnarray*}
Since by Lemma~\ref{lem:MnuV} $M_\nu V_\mus^\nu=V_\musp^{\nu'}$ and $M_\nu V_\musp^\nu=V_\mus^{\nu'}$, we have that
$-u^\tp$ in the expression above fixes $V_\musp^{\nu'}$ in $\bnu'$, while
$-(u')^\tp$ in the expression above fixes $V_\mus^{\nu'}$ in $\bnu'$.
Therefore $(\rho_\nu\oplus\rho_{\nu'})(x)\Pi=u\oplus(-u')^\tp$.
Since $u'$ can be chosen independently of $u$, such that $(-u')^\tp$ is any element in $\su(V_\mus^{\nu'})$,
we conclude that 
$$\su(V_\mus^\nu)\oplus\su(V_\mus^{\nu'})\subset(\rho_\nu+\rho_{\nu'})(\g_n)\Pi.$$
Therefore by Lemma~\ref{lem:suVsubthenu}, we have that $\{\mu^{(i)}\}_{i=1}^N$ is universal on $\{\nu,\nu'\}$.
\end{proof}

Combining this argument with the previous result on self-conjugate $\nu$, we derive a 
condition for universality on a more general set of partitions:

\begin{lemma}
\label{lem:muneqmupuniversal}
Given a family of partitions $\{\mu^{(i)}\}_{i=1}^N$ such that 
$\{\mu^{(i)}\}_{i=1}^N\neq\{\mu^{(i)\prime}\}_{i=1}^N$,
and a set of proper partitions $\cP\subset\cP_\mus$, we have that
$\{\mu^{(i)}\}_{i=1}^N$ is universal on $\cP$.
\end{lemma}
\begin{proof}
By Lemma~\ref{lem:noncononselfcon}, $\{\mu^{(i)}\}_{i=1}^N$ is universal on every self-conjugate partition in $\cP$,
and by Lemma~\ref{lem:ifnoncononpaircon} $\{\mu^{(i)}\}_{i=1}^N$ is universal on every pair of conjugate partitions in $\cP$.
Therefore by Lemma~\ref{lem:additiveu}, we have that $\{\mu^{(i)}\}_{i=1}^N$ is universal on $\cP$.
\end{proof}

\begin{figure}
\begin{center}
\begin{tikzpicture}
\matrix (m) [
    matrix of nodes,
    row sep=0.4cm,
    column sep=1cm,
    every node/.style={text width=4cm, text depth=0.5cm},
    every node/.append code={
       \ifnum\pgfmatrixcurrentrow=1\relax
          \pgfkeysalso{align=center}
       \else
           \pgfkeysalso{align=center,draw}
       \fi
    }
]{
{$\nu\phantom{'}=\nu'$} \\
$V_{\mu^{(1)}\cdots\mu^{(N)}}^\nu=V_{\mu^{(1)\prime}\cdots\mu^{(N)\prime}}^{\nu'}$ \\
|[draw=none]|  \\
};
\draw [->] (m-2-1.south)arc(-160:160:1.65);
\node[rectangle,draw,fit=(m-1-1)(m-3-1)]{};
\end{tikzpicture}
\end{center}
\caption{Schematic diagram of $M_\nu$ mapping from the isotypical subspace $V_\mus^\nu=V_{\musp}^{\nu'}$ of $\nu=\nu'$ to itself.} 
\label{fig:mueqmupnueqnup}
\end{figure}

However the strategies above fail when $\{\mu^{(i)}\}_{i=1}^N=\{\mu^{(i)\prime}\}_{i=1}^N$.  For example if we also have that
$\nu=\nu'$, then $M_\nu$ is an automorphism on $V_\mus^\nu$ (Figure~\ref{fig:mueqmupnueqnup}).
The condition that $\rho_\nu(x)$ be $M_\nu$-invariant on $V_\mus^\nu$ is 
then too constraining to allow unitary operators, in general.
\begin{lemma}
\label{lem:mueqmupnueqnup}
Given a family of self-conjugate partitions $\{\mu^{(i)}\}_{i=1}^N$ such that 
$\dim(\bigotimes_{i=1}^N\bmu^{(i)})>2$, 
and a self-conjugate, proper $\nu$, it follows that 
$\{\mu^{(i)}\}_{i=1}^N$ is not universal on $\nu$.
\end{lemma}
\begin{proof}

Universality implies that for all $u\in\su(\bigotimes_{i=1}^N\bmu^{(i)})$, we have that
$\chi_\mus^\nu(u\otimes\one)(\chi_\mus^\nu)^{-1}\in\rho_\nu(\g_n)$.
Then Lemma~\ref{lem:isoMorphism} implies that
$$M_\nu \chi_\mus^\nu(u\otimes\one)(\chi_\mus^\nu)^{-1}M_\nu^{-1}=-\chi_\mus^\nu(u^\tp\otimes\one)(\chi_\mus^\nu)^{-1},$$
where we have used that $\chi_\mus^\nu$ is an orthogonal transformation
(by Lemma~\ref{lem:orthogonalchi}).
Choose $u=\diag(-1,-1,2,0,\ldots,0)$, where the number of trailing zeros
equals $\dim(\bigotimes_{i=1}^N\bmu^{(i)})-3$.
By the above relation, $u$ is similar to its negation.
But similar matrices have the same eigenvalues while $u$ and $-u$ do not:
Contradiction.  Therefore, $\{\mu^{(i)}\}_{i=1}^N$ is not universal on $\nu$.
\end{proof}

\begin{figure}
\begin{center}
\begin{tikzpicture}
\matrix (m) [
    matrix of nodes,
    row sep=0.4cm,
    column sep=1cm,
    every node/.style={text width=4cm, text depth=0.5cm},
    every node/.append code={
       \ifnum\pgfmatrixcurrentrow=1\relax
          \pgfkeysalso{align=center}
       \else
           \pgfkeysalso{align=center,draw}
       \fi
    }
]{
{$\nu\phantom{'}$} &
{$\nu'$} \\
$V_{\mu^{(1)}\cdots\mu^{(N)}}^\nu=V_{\mu^{(1)\prime}\cdots\mu^{(N)\prime}}^\nu$ &
$V_{\mu^{(1)}\cdots\mu^{(N)}}^{\nu'}=V_{\mu^{(1)\prime}\cdots\mu^{(N)\prime}}^{\nu'}$ \\
|[draw=none]| & |[draw=none]| \\
};
\draw[->] (m-2-1) -- (m-2-2);
\node[rectangle,draw,fit=(m-1-1)(m-3-1)]{};
\node[rectangle,draw,fit=(m-1-2)(m-3-2)]{};
\end{tikzpicture}
\end{center}
\caption{Schematic diagram of $M_\nu$ mapping from the isotypical subspace $V_\mus^\nu=V_\musp^\nu$ of $\nu$ to the isotypical subspace
$V_\mus^{\nu'}=V_\musp^{\nu'}$ of $\nu'$.}
\label{fig:mueqmupnuneqnup}
\end{figure}

In the case that $\{\mu^{(i)}\}_{i=1}^N=\{\mu^{(i)\prime}\}_{i=1}^N$ and $\nu\neq\nu'$, the action on $V_\mus^{\nu'}=V_\musp^{\nu'}$ is directly determined, via $M_\nu$, by the action on $V_\mus^\nu=V_\musp^\nu$ (Figure~\ref{fig:mueqmupnuneqnup}).
This proves again to be too constraining for both $\rho_\nu(x)$ and
$\rho_{\nu'}(x)$ to simultaneously satisfy the intertwining condition required for 
universality:
\begin{lemma}
\label{lem:mueqmupnuneqnup}
Given a family of self-conjugate partitions $\{\mu^{(i)}\}_{i=1}^N$
such that $\dim(\bigotimes_{i=1}^N\bmu^{(i)})>2$,
let $\cP$ be a set of partitions such that for some $\nu\in \cP$
where $\nu\neq\nu'$, we also have $\nu'\in \cP$.  Then $\{\mu^{(i)}\}_{i=1}^N$ is not
universal on $\cP$.
\end{lemma}
\begin{proof}
Assume that $\{\mu^{(i)}\}_{i=1}^N$ is universal on $\cP$.
Let $u$ be an arbitrary member of $\su(\bigotimes_{i=1}^N\bmu^{(i)})$.
(We choose a specific one later to arrive at a contradiction.)
By definition of universality there exists a witnessing function
$$\phi:\su(\bigotimes_{i=1}^N\bmu^{(i)})\rightarrow\left(\bigoplus_{\nu\in \cP}\rho_\nu\right)(\g_n)$$
and we can choose
$\iota\in\Hom(\bigotimes_{i=1}^N\bmu^{(i)},\bnu)$ and $\iota'\in\Hom(\bigotimes_{i=1}^N\bmu^{(i)},\bnu')$
such that for all $\psi\in\bigotimes_{i=1}^N\bmu^{(i)}$, we have
\begin{eqnarray*}
\phi(u)\iota(\psi) &=& \iota(u\psi),\\
\phi(u)\iota'(\psi) &=& \iota'(u\psi).
\end{eqnarray*}
Therefore there exists $x\in\g_n$ such that 
\begin{eqnarray*}
\rho_\nu(x)\iota(\psi) &=& \iota(u\psi)\\
\rho_{\nu'}(x)\iota'(\psi) &=& \iota'(u\psi)
\end{eqnarray*}
from which follows
\begin{eqnarray*}
u &=& \iota^{-1}\rho_\nu(x)\iota\\
u &=& (\iota')^{-1}\rho_{\nu'}(x)\iota'.
\end{eqnarray*}
By Lemma~\ref{lem:isoMorphism}, we have $M_\nu^{-1}\rho_{\nu'}(x)M_\nu=-\rho_\nu^\tp(x)$.  This implies 
that $u$ and $-u^{\tp}$ are related by a similarity transformation.
Let $u=\diag(-1,-1,2,0,\ldots,0)$, where the number of trailing zeros
equals $\dim(\bigotimes_{i=1}^N\bmu^{(i)})-3$.  Then $u$ and $-u^\tp$ have
unequal sets of eigenvalues and so cannot be similar: Contradiction.
\end{proof}

The above cases of non-universality assume $\dim(\bigotimes_{i=1}^N\bmu^{(i)})>2$.
For $\dim(\bigotimes_{i=1}^N\bmu^{(i)})\leq 2$, universality is in fact unavoidable.
Since the only self-conjugate, 2-dimensional representations are given by $\ydiagram{2,1}$
and $\ydiagram{2,2}$, and since the only self-conjugate, 1-dimensional representation is
$\ydiagram{1}$, $\dim(\bigotimes_{i=1}^N\bmu^{(i)})\leq 2$ implies that at most one of
the $\mu^{(i)}$, say $\mu^{(j)}$, equals $\ydiagram{2,1}$ or $\ydiagram{2,2}$, with
every other $\mu^{(i)}$ being $\ydiagram{1}$.
Then the problem of constructing a universality-witnessing map on $\su(\bigotimes_{i=1}^N\bmu^{(i)})$ reduces to the problem of constructing a universality-witnessing map on $\su(\bmu^{(j)})$, which is guaranteed 
because $\mu^{(j)}$ is universal on itself.

\subsection{Hooks}
Alongside self-conjugates and conjugate pairs, hooks play a special role in the representation theory of transpositions.  Of course we are particularly interested in the question of whether hooks
admit universality.  
We have already seen in the case of a single partition (Theorem~\ref{thm:iff1partitionu}) that
shallow hooks are universal on themselves, while deep hooks are not.
In this section we consider whether hooks admit universality of multiple partitions. 
We begin by determining which hooks admit nontrivial intertwiners from 
a product of representations.
\begin{lemma}
\label{lem:2hooks}
Given partitions $\lambda$ and $\mu$,
$\{\nu|c_{\lambda\mu}^{\nu}>0\}$ contains an improper partition if and only if $\lambda$ and $\mu$ are improper.
In that case, $\{\nu|c_{\lambda\mu}^{\nu}>0\}$ contains exactly two improper partitions:
$$\nu=[\lambda_1+\mu_1,1^{|\lambda|+|\mu|-\lambda_1-\mu_1}],$$
and
$$\nu=[\lambda_1+\mu_1-1,1^{|\lambda|+\mu|-\lambda_1-\mu_1+1}].$$
Moreover, we have $c_{\lambda\mu}^\nu=1$ in both cases.
\end{lemma}
\begin{proof}\ \\
\noindent $\Rightarrow$
We prove the contrapositive. Suppose $\lambda$ or $\mu$ is proper.
Then $\lambda_2>1$ or $\mu_2>1$.
Then by Lemma~\ref{lem:nugeqmax}, for all $\nu\vdash|\lambda|+|\mu|$
such that $c_{\lambda\mu}^\nu>0$,
$\nu_2>1$ and therefore no such $\nu$ is improper.\\
\noindent $\Leftarrow$
Assume $\lambda$ and $\mu$ are improper.
Let $\nu\vdash|\lambda|+|\mu|$ be improper
such that $\nu_i\geq\max(\lambda_i,\mu_i)$.
Following the Littlewood--Richardson rules for determining $c_{\lambda\mu}^\nu$,
consider the skew diagram $\nu\backslash\lambda$,
which we note is the disjoint union of
$[\nu_1-\lambda_1]$ and $[1^{\nup_1-\lambdap_1}]$.
As in Definition~\ref{def:LRcoefficient}, its weight is determined by $\mu$.
In particular, it must contain $\mu_1$ 1's.
Further, $[\nu_1-\lambda_1]$ must contain only 1's in order to form 
a Littlewood--Richardson sequence as required by Definition~\ref{def:LRtableau},
and $[1^{\nup_1-\lambdap_1}]$ must contain no more than one 1 to
satisfy the definition of a semistandard tableau.
It follows that $c_{\lambda\mu}^\nu=1$ if and only if
$\nu=[\lambda_1+\mu_1,1^{|\lambda|+|\mu|-\lambda_1-\mu_1}]$ or
$\nu=[\lambda_1+\mu_1-1,1^{|\lambda|+\mu|-\lambda_1-\mu_1+1}]$.
\end{proof}
\noindent A visual demonstration of the above proof may provide some clarity.
\begin{exmp}
Let $\lambda=\mu=\ydiagram{2,1}$.  Then the only possible hooks that can result from
their product are $\nu=\ydiagram{4,1,1}$,
\setydiagrameq
$$
\ydiagram{4,1,1} - \ydiagram{2,1}  = \ydiagram{2+2,0,1}  \mapsto \ytableaushort{\none\none 11,\none,2},\\
$$
\setydiagramtext
and $\nu=\ydiagram{3,1,1,1}$,
\setydiagrameq
$$
\ydiagram{3,1,1,1}  - \ydiagram{2,1} = \ydiagram{2+1,0,1,1}  \mapsto\ \ytableaushort{\none\none 1,\none,1,2},
$$
\setydiagramtext
where above we have illustrated $\nu\backslash\lambda$ and the only permissible Littlewood--Richardson 
tableaux.
\end{exmp}
\noindent From this we can conclude that shallow hooks do not admit nontrivial intertwiners
from products of nontrivial representations.
\begin{lemma}
Given partitions $\lambda$ and $\mu$,
if $\dim(\blambda)>1$ and $\dim(\bmu)>1$, then $\{\nu|c_{\lambda\mu}^{\nu}>0\}$
does not contain the hooks $[n-1,1]$ or $[2,1^{n-2}]$.
\end{lemma}
\begin{proof}
The one-dimensional representations ruled out by hypothesis are
associated with single-row or single-column partitions of the form $[m]$ and $[1^m]$.
Thus we have that $2\leq\lambda_1\leq|\lambda|-1$ and $2\leq\mu_1\leq|\mu|-1$.  
So it follows that $3<\lambda_1+\mu_1<n-1$ and thus $2<\lambda_1+\mu_1-1<n-2$.
Then by Lemma~\ref{lem:2hooks}, if $\nu\in\{\nu|c_{\lambda\mu}^\nu>0\}$
is a hook then either $3<\nu_1<n-1$ or $2<\nu_1<n-2$.  In either case, we have
$\nu_1\neq n-1$ and $\nu_1\neq 2$.  Thus $[n-1,1]$ and $[2,1^{n-2}]$ are excluded.
\end{proof}

So shallow hooks cannot admit universality of multiple nontrivial partitions.
On the other hand, in all but a few cases, the universality of a family of multiple nontrivial
hooks on
deep hooks can be ruled out by dimensional arguments.  Essentially, given such a
family
$\{\mu^{(i)}\}_{i=1}^N$, the dimension of the relevant computational space
$\dim(\bigotimes_{i=1}^N\bmu^{(i)})$ is too large relative to
that of the space $V_{[n-1,1]}$, where $n=\sum_{i=1}^N|\mu^{(i)}|$.
To apply this argument to a family of arbitrarily many partitions,
first we need this following generalization of Lemma~\ref{lem:2hooks}.
\begin{lemma}
\label{lem:hookfamilies}
Given families of partitions $\{\mu^{(i)}\}_{i=1}^N$ and $\{\nu^{(i)}\}_{i=1}^N$ such that $\prod_{i=1}^{N-1}c_{\nu^{(i)}\mu^{(i+1)}}^{\nu^{(i+1)}(d)}>0$ and $\nu^{(N)}$ is a hook,
we have that every $\mu^{(i)}$ and $\nu^{(i)}$ is improper.
\end{lemma}
\begin{proof}
By Lemma~\ref{lem:2hooks}, since $\nu^{(N)}$ is a hook and $c_{\nu^{(N-1)}\mu^{(N)}}^{\nu^{(N)}(d)}>0$,
it follows that $\nu^{(N-1)}$ and $\mu^{(N)}$ are improper.
Similarly, since $\nu^{(N-1)}$ is improper and $c_{\nu^{(N-2)},\mu^{(N-1)}}^{\nu^{(N-1)}}>0$, it follows that $\nu^{(N-2)}$ and $\mu^{(N-1)}$
are improper.  Iterating this argument, every $\mu^{(i)}$ and $\nu^{(i)}$ is also improper.
\end{proof}
\noindent The promised dimensional argument 
requires some arithmetic,
particularly the following inequality between a product and a sum as succinctly proved in \cite{Ecker}. 
\begin{lemma}
\label{lem:prodgtsum}
Given integers $N>2$ and $m_i>1$ for all $1\leq i\leq N-1$, we have $\prod_{i=1}^{N-1}m_i\geq\sum_{i=1}^{N-1}m_i$.
\end{lemma}
\begin{proof}
Clearly the result holds for $N=3$, which we take as our base case,
and proceed to prove that it holds for $N>3$ by induction.  
Suppose the result is true for $N-1>2$. 
Then we have:
$$m_1+m_2+\ldots+m_{N-2}+m_{N-1}\leq m_1m_2\cdots m_{N-2}+m_{N-1}.$$
Applying the base case to the right hand side, we conclude that 
$$m_1+m_2+\ldots+m_{N-1}\leq m_1m_2\cdots m_{N-1}.$$ 
\end{proof}
Finally we obtain:
\begin{lemma}
\label{lem:hookdimsmall}
Given a family of partitions $\{\mu^{(i)}\}_{i=1}^N$ such that 
$N>1$ and for all $i$ $\dim(\bmu^{(i)})>1$, suppose 
that $\{\mu^{(i)}\}_{i=1}^N$ is not one of $(\ydiagram{2,1},\ydiagram{2,1})$,
$(\ydiagram{2,1},\ydiagram{3,1})$, $(\ydiagram{2,1},\ydiagram{2,1,1})$,
or $(\ydiagram{2,1},\ydiagram{2,1},\ydiagram{2,1})$, up to reordering.
Then $\{\mu^{(i)}\}_{i=1}^N$ is not universal on any hook.
\end{lemma}
\begin{proof}
Suppose $\{\mu^{(i)}\}_{i=1}^N$ is universal on a hook $\nu\vdash n\equiv\sum_{i=1}^N|\mu^{(i)}|$.  Then by Lemma~\ref{lem:hookfamilies}, $\{\mu^{(i)}\}_{i=1}^N$ must consist of
hooks.  
We show that this implies $\dim(\bigotimes_{i=1}^N\bmu^{(i)})>n-1$.
It is convenient to separately consider three cases.
\begin{enumerate}[leftmargin=0pt,label={}]
\item Case $N=2$:
Let $\ell=|\mu^{(1)}|$ and $m=|\mu^{(2)}|$.
By hypothesis, we have $\ell\geq 3$ and $m\geq 5$. 
Then $\ell>2\left(\frac{m-1}{m-2}\right)$, because for $m\geq 5$ the right hand side
is less than 3.  Rearranging, we have $(\ell-1)(m-1)>\ell+m-1$.
So we have by Lemma~\ref{lem:hookdimension},
\begin{eqnarray*}
\dim(\bmu^{(1)}\otimes\bmu^{(2)}) &\geq& (\ell-1)(m-1)\\
&>& \ell+m-1\\
&>& n-1.
\end{eqnarray*}
\item Case $N>2, \mu^{(i)}=\ydiagram{2,1}$ for all $1\leq i\leq N$:
Since $\dim(V_\subydiagramtext{2,1})=2$, we have $\dim(\bigotimes_{i=1}^N\bmu^{(i)})=2^N$.
Excluding $N=3$ by hypothesis, and noting that $n=3N$, we therefore have
$\dim(\bigotimes_{i=1}^N\bmu^{(i)})>n-1$.
\item Case $N>2, \mu^{(j)}\neq\ydiagram{2,1}$ for some $1\leq j\leq N$:
Let $m_i=|\mu^{(i)}|-1$ and without loss of generality, since order does not matter, consider $j=N$. 
Therefore we have $m_N>2$. 
By Lemma~\ref{lem:hookdimension}, it follows that
$$
\dim\left(\bigotimes_{i=1}^N\bmu^{(i)}\right) \geq \prod_{i=1}^Nm_i.
$$
Then using that $\prod_{i=1}^{N-1}m_i \geq \sum_{i=1}^{N-1}m_i>N-1$ by Lemma~\ref{lem:prodgtsum},
and that $\prod_{i=1}^{N-1}m_i\geq 4$ while $m_N/(m_N-2)\leq 3$,
we obtain
\begin{eqnarray*}
\dim\left(\bigotimes_{i=1}^N\bmu^{(i)}\right) &\geq& \prod_{i=1}^{N-1}m_i+(m_N-2)\prod_{i=1}^{N-1}m_i+\prod_{i=1}^{N-1}m_i\\
&>& \sum_{i=1}^{N-1}m_i+m_N+N-1\\
&>& \sum_{i=1}^Nm_i+N-1\\
&>& \sum_{i=1}^N|\mu^{(i)}|-1\\
&>& n-1.
\end{eqnarray*}
\end{enumerate}
If $\{\mu^{(i)}\}_{i=1}^N$ is universal on the hook $\nu\vdash n$,
then by Lemma~\ref{lem:monomorphism}, $\su(\bigotimes_{i=1}^N\bmu^{(i)})$ is isomorphic to a 
subalgebra of $\Pi\rho_\nu(\g'_n)\Pi$, where $\Pi$ is the projection onto
$V_{\mu^{(1)}\cdots\mu^{(N)}}^\nu$.  Complexifying, we have $\sla(\bigotimes_{i=1}^N\bmu^{(i)})\subset\Pi\rho_\nu(\g'_n)\Pi$.  Thus it follows that 
$\dim(\rho_\nu(\g'_n))\geq\dim(\sla(\bigotimes_{i=1}^N\bmu^{(i)}))$.
Then by the above result, we have 
\begin{eqnarray*}
\dim(\rho_\nu(\g'_n)) &>& \dim(\sla(n-1,\cmplx))\\
&>& \dim(\sla(V_{[n-1,1]})).
\end{eqnarray*}
  But by
Lemma~\ref{lem:marin}, we have $\rho_\nu(\g'_n)\cong\sla(V_{[n-1,1]})$:
Contradiction.

\end{proof}

Of the four cases that remain, two of them can be ruled out by considering associated representations of the special unitary algebra:
\begin{lemma}
\label{lem:smallhooks}
There is no hook on which $(\ydiagram{2,1},\ydiagram{3,1})$ or $(\ydiagram{2,1},\ydiagram{2,1,1})$ is universal.
\end{lemma}
\begin{proof}
Let the pair of partitions $(\lambda,\mu)$ be one of $(\ydiagram{2,1},\ydiagram{3,1})$ or $(\ydiagram{2,1},\ydiagram{2,1,1})$.
The only hooks $\nu$ for which $c_{\lambda\mu}^\nu>0$ are $\nu=\ydiagram{5,1,1}$ for
$(\lambda,\mu)=(\ydiagram{2,1},\ydiagram{3,1})$,
$\nu=\ydiagram{4,1,1,1}$ for $(\lambda,\mu)=(\ydiagram{2,1},\ydiagram{3,1})$ or $(\ydiagram{2,1},\ydiagram{2,1,1})$,
and $\nu=\ydiagram{3,1,1,1,1}$ for $(\lambda,\mu)=(\ydiagram{2,1},\ydiagram{2,1,1})$.
We proceed by showing that each of the above $\bnu$ can be identified with an irreducible representation
of $\sla(6,\cmplx)$, and then prove that universality of $(\lambda,\mu)$ on $\nu$ implies that
each such irreducible representation of $\sla(6,\cmplx)$ has a proper subrepresentation of $\sla(6,\cmplx)$: a contradiction.
First, by Lemma~\ref{lem:marin}, we have $\rho_\nu(\g_7)\cong\rho_\subydiagramtext{6,1}(\g_7)\cong\sla(6,\cmplx)$,
where $\ydiagram{6,1}$ can be identified with the standard 6-dimensional irreducible representation
of $\sla(6,\cmplx)$.
Since each $\nu$ above is an exterior power of $\ydiagram{6,1}$, it is also an irreducible
representation of $\sla(6,\cmplx)$.
Thus there exists an irreducible representation $\Phi_\nu:\sla(6,\cmplx)\rightarrow\rho_\nu(\g_n)$
that is also an isomorphism.
This representation has dimension 15 for $\nu=\ydiagram{5,1,1}$
and $\nu=\ydiagram{3,1,1,1,1}$, and 20 for $\nu=\ydiagram{4,1,1,1}$.

Assume $(\lambda,\mu)$
is universal on $\nu\in\left\{\ydiagram{5,1,1},\ydiagram{4,1,1,1},\ydiagram{3,1,1,1,1}\right\}$.
By definition this implies the existence of a witnessing function
$\phi_\nu:\su(\blambda\otimes\bmu)\rightarrow\rho_\nu(\g_7)$ such that for all
$u\in\su(\blambda\otimes\bmu)$, $\iota\in\Hom(\blambda\otimes\bmu,\bnu)$,
and $\psi\in\blambda\otimes\bmu$, $\phi_\nu(u)\iota(\psi)=\iota(u\psi)$.
The latter condition implies that members of the image of $\phi_\nu$ fix the isotypical subspace $V_{\lambda\mu}^\nu$.
By Lemma~\ref{lem:monomorphism}, $\phi_\nu$ is injective.
By complexification we obtain the map $\phi_\nu^\cmplx:\sla(\blambda\otimes\bmu)\rightarrow\rho_\nu(\g_7)$, the image of which again fixes
$V_{\lambda\mu}^\nu$, and again it is injective.
Because $\rho_\nu(\g_7)=\Phi_\nu(\sla(6,\cmplx))$, we have $\phi_\nu^\cmplx(\sla(\blambda\otimes\bmu))\subset\Phi_\nu(\sla(6,\cmplx))$.
We also have $\dim(\blambda\otimes\bmu)=6$ from which follows $\sla(\blambda\otimes\bmu)\cong\sla(6,\cmplx)$ and $\Span(\phi_\nu^\cmplx(\sla(\blambda\otimes\bmu)))\cong\sla(6,\cmplx)$.  Recalling
that $\Phi_\nu(\sla(6,\cmplx))\cong\sla(6,\cmplx)$, 
we must have $\Span(\phi_\nu^\cmplx(\sla(\blambda\otimes\bmu)))=\Phi_\nu(\sla(6,\cmplx))$.  It follows that the subspace $V_{\lambda\mu}^\nu$ of $\nu$
is invariant with respect to the action of $\sla(6,\cmplx)$ as determined by $\Phi_\nu$.
However Lemma~\ref{lem:2hooks} implies $c_{\lambda\mu}^\nu=1$ and thus $\dim(V_{\lambda\mu}^\nu)=\dim(\blambda\otimes\bmu)=6$,
while $\dim(\bnu)>6$.  Thus the representation $(\Phi_\nu,\nu)$ of $\sla(6,\cmplx)$ has
a proper subrepresentation, in contradiction with its irreducibility.
\end{proof}

So, when the goal is to assemble partitions that  admit universality, 
hooks are generally to be avoided.
This fact motivates conditions by which to exclude hooks from $\cP_\mus^{(d)}$:
\begin{lemma}
\label{lem:nohooks}
Given a family of multiple nontrivial partitions $\{\mu^{(i)}\}_{i=1}^N$ of at most $d$ rows,
we have that $\cP_\mus^{(d)}$ contains a hook if and only if $\max\{\mu^{(i)}_2\}= 1$
and $\sum_{i=1}^N\mu^{(i)\prime}_1<N+d$.
\end{lemma}
\begin{proof}
If every $\mu^{(i)}$ is a hook, then by Lemma~\ref{lem:2hooks}, for each hook $\nu^{(i)}$ there are
exactly two hooks $\nu^{(i+1)}$ such that $c_{\nu^{(i)}\mu^{(i+1)}}^{\nu^{(i+1)}}>0$,
one satisfying $\nu^{(i+1)\prime}_1=\nu^{(i)\prime}+\mu^{(i+1)\prime}_1-1$
and the other satisfying $\nu^{(i+1)\prime}_1=\nu^{(i)\prime}+\mu^{(i+1)\prime}_1$.
Therefore if $\nu^{(1)}=\mu^{(1)}$, by iteration there exists a hook $\nu^{(N)}$ such that
$\prod_{i=1}^{N-1}c_{\nu^{(i)}\mu^{(i+1)}}^{\nu^{(i+1)}(d)}>0$ satisfying
$\nu^{(N)\prime}_1=\sum_{i=1}^N\mu^{(i)\prime}_1-N+1$.
So if $\sum_{i=1}^N\mu^{(i)\prime}_1<N+d$, we have $\sum_{i=1}^N\mu^{(i)\prime}_1-N+1\leq d$,
thus $\nu^{(N)\prime}_1\leq d$ and in turn $\nu^{(N)}\in\cP_\mus^{(d)}$.

Conversely if any $\mu^{(i)}$ is not a hook then by Lemma~\ref{lem:hookfamilies}, $\cP_\mus^{(d)}$ contains
no hooks.  Or if $\sum_{i=1}^N\mu^{(i)\prime}_1\geq N+d$, then any hook $\nu$ such that
$c_\mus^\nu>0$ must satisfy $\nup_1\geq d+1$ and thus $\nu^{(N)}\not\in\cP_\mus^{(d)}$.
\end{proof}

\subsection{Multiple nontrivial partitions}

In this section we consider a family of multiple nontrivial partitions,
and in particular a family of two such partitions,
with the aim of deriving 
necessary and sufficient conditions for universality. 
It already follows 
from Lemma~\ref{lem:muneqmupuniversal} that if at least one of the partitions is proper, and at least one of the partitions is not self-conjugate, then universality is guaranteed.
We now examine the other cases -- when all partitions are hooks, and when all
partitions are self-conjugate  -- in some detail.

To begin with, in the last section we found that if the partitions are hooks 
then
universality on hooks is ruled out in all but two cases.
Below we show that whenever these remaining two cases map by intertwiner to hooks, they fail to be $d$-universal.  We conclude that hooks generally contraindicate universality: 
\begin{lemma}
\label{lem:ifhooknotu}
Given a family $\{\mu^{(i)}\}_{i=1}^N$ of partitions such that for all $i$ $\dim(\bmu^{(i)})>1$ and $N>1$,
if $\cP_\mus^{(d)}$ contains a hook then $\{\mu^{(i)}\}_{i=1}^N$ is not $d$-universal.
\end{lemma}
\begin{proof}
If $\cP_\mus^{(d)}$ contains a hook then by Lemma~\ref{lem:hookfamilies}, $\{\mu^{(i)}\}_{i=1}^N$
consists entirely of hooks. 
By Lemma~\ref{lem:hookdimsmall}, unless $\{\mu^{(i)}\}_{i=1}^N$ is one of $(\ydiagram{2,1},\ydiagram{3,1})$,
$(\ydiagram{2,1},\ydiagram{2,1,1})$, $(\ydiagram{2,1},\ydiagram{2,1})$,
or $(\ydiagram{2,1},\ydiagram{2,1},\ydiagram{2,1})$,
then $\{\mu^{(i)}\}_{i=1}^N$ is not $d$-universal.
If $\{\mu^{(i)}\}_{i=1}^N$ equals $(\ydiagram{2,1},\ydiagram{3,1})$ or
$(\ydiagram{2,1},\ydiagram{2,1,1})$, then by Lemma~\ref{lem:smallhooks} $\{\mu^{(i)}\}_{i=1}^N$ is not $d$-universal.
For the remaining two cases, note that by Lemma~\ref{lem:nohooks}, $N+d>\sum_{i=1}^N\mu^{(i)\prime}_1\geq 2N$, 
hence $d>N\geq 2$.
Therefore if $\{\mu^{(i)}\}_{i=1}^N=(\ydiagram{2,1},\ydiagram{2,1})$, 
it follows by the Littlewood--Richardson rules that $\cP_\mus^{(d)}$ contains $\ydiagram{3,2,1}$, 
and if $\{\mu^{(i)}\}_{i=1}^N=(\ydiagram{2,1},\ydiagram{2,1},\ydiagram{2,1})$ 
then it follows by the Littlewood--Richardson rules that $\cP_\mus^{(d)}$ contains $\ydiagram{3,3,3}$.  
In either case, by Lemma~\ref{lem:mueqmupnueqnup}, $\{\mu^{(i)}\}_{i=1}^N$ is not $d$-universal.
\end{proof}

Combining the above implications of hooks on universality, with previously derived implications of conjugates
on universality, we obtain necessary and sufficient conditions:
\begin{theorem}
\label{thm:iffumulti}
A family of nontrivial partitions $\{\mu^{(i)}\}_{i=1}^N$ where $N>1$ is $d$-universal if and only if
\begin{enumerate}[label=(\roman{*})]
\item $P_\mus^{(d)}$ has no hooks, and
\item if $\{\mu^{(i)\prime}\}_{i=1}^N=\{\mu^{(i)}\}_{i=1}^N$ then $P_\mus^{(d)}$ has no conjugates.
\end{enumerate}
\end{theorem}
\begin{proof}
Assume these conditions are true.  If $\{\mu^{(i)\prime}\}_{i=1}^N\neq\{\mu^{(i)}\}_{i=1}^N$, then since $\cP_\mus^{(d)}$ 
contains no hooks it admits the universality of $\{\mu^{(i)}\}_{i=1}^N$ by Lemma~\ref{lem:muneqmupuniversal}.
If $\{\mu^{(i)\prime}\}_{i=1}^N=\{\mu^{(i)}\}_{i=1}^N$, then $\cP_\mus^{(d)}$ contains neither hooks nor conjugates and 
therefore admits the universality of $\{\mu^{(i)}\}_{i=1}^N$ by Lemma~\ref{lem:ifnohooksnorconjthenu}.

Conversely, if the first condition fails then $\cP_\mus^{(d)}$ contains a hook and therefore fails to admit the universality 
of $\{\mu^{(i)}\}_{i=1}^N$ by Lemma~\ref{lem:ifhooknotu}.
If the second condition fails, then $\{\mu^{(i)\prime}\}_{i=1}^N=\{\mu^{(i)}\}_{i=1}^N$ and $\cP_\mus^{(d)}$ contains a 
self-conjugate partition or a pair of conjugate partitions.  Furthermore the existence of multiple nontrivial partitions in 
$\{\mu^{(i)}\}_{i=1}^N$  implies $\dim(\bigotimes_{i=1}^N\bmu^{(i)})>2$.  Therefore if $\cP_\mus^{(d)}$ contains a 
self-conjugate partition it fails to admit the universality of $\{\mu^{(i)}\}_{i=1}^N$ by Lemma~\ref{lem:mueqmupnueqnup}, and if 
$\cP_\mus^{(d)}$ contains a pair of conjugate pairs then it fails to admit the universality of $\{\mu^{(i)}\}_{i=1}^N$ by 
Lemma~\ref{lem:mueqmupnuneqnup}. 
\end{proof}

Application of the above conditions for universality, as presented, is typically not efficient, 
because the conditions require 
checking whether Littlewood--Richardson coefficients vanish for a number of partitions that grows
faster than a polynomial in the total number of $d$-state systems.  A shortcut exists for the first condition 
in the form of Lemma~\ref{lem:nohooks}, which gives simple conditions on the partitions in $\{\mu^{(i)}\}_{i=1}^N$ 
to rule out hooks in $\cP_\mus^{(d)}$.  However we do not have a comparable shortcut for the second condition, 
whereby one can rule out conjugates in $\cP_\mus^{(d)}$, given a family of self-conjugate partitions.

Fortunately the problem of efficiently determining universality becomes tractable in the case of two partitions.  
In this case we can 
derive necessary and sufficient conditions, on the pair
$(\lambda,\mu)=(\lambda',\mu')$,
for the set $\cP_{\lambda\mu}^{(d)}$ to be devoid of conjugates.
First we simply derive a sufficient condition for the first row of every partition in $\cP_{\lambda\mu}^{(d)}$ 
to be longer than $d$.  This condition is obviously sufficient to exclude conjugates from 
$\cP_{\lambda\mu}^{(d)}$, since they do not ``fit" its definition that every column be no longer than $d$.  
It may not be so obvious, however, that this condition is also necessary.  
We prove it to be thus by explicitly producing a conjugate pair in $\cP_{\lambda\mu}^{(d)}$,
under the assumption that the aforementioned condition is violated.
This construction
makes use of Knutson and Tao's seminal result on the relationship between 
Littlewood--Richardson coefficients and the eigenvalues of Hermitian matrices.

\begin{lemma}
\label{lem:iffnoconjugates}
Given self-conjugate partitions $\lambda$ and $\mu$, 
$\cP_{\lambda\mu}^{(d)}=\{\nu|c_{\lambda\mu}^{\nu(d)}>0\}$ contains no conjugates
if and only if $\max_{i+j=1+d}\{\lambda_i+\mu_j\}>d$.
\end{lemma}
\begin{proof}
By Weyl's inequalities in Lemma~\ref{lem:Weyl}, given any partition $\nu$ such that 
$c_{\lambda\mu}^\nu>0$, we have
$$\max_{i+j=1+d}{\lambda_i+\mu_j}\leq\nu_1.$$
Therefore
if $\max_{i+j=1+d}\{\lambda_i+\mu_j\}>d$ then 
for any partition $\nu$ such that $c_{\lambda\mu}^\nu>0$, and
in particular for all $\nu\in \cP_{\lambda\mu}^{(d)}$, we have $\nu_1>d$.
Then $(\nu')'_{\ 1}>d$, from which it follows by definition that
$c_{\lambda\mu}^{\nu'(d)}=0$.  Thus we have $\nu'\not\in \cP_{\lambda\mu}^{(d)}$.

Conversely, assume $\max_{i+j=1+d}\{\lambda_i+\mu_j\}\leq d$. We show that $\cP_{\lambda,mu}^{(d)}$ has a conjugate pair, as follows.
Let $\nu=\sort(\lambda_1+\mu_d,\lambda_2+\mu_{d-1},\ldots,\lambda_d+\mu_1)$,
where here ``sort" denotes the operation of arranging the sequence
into weakly decreasing order. 
Thus for all $i$, we have $\nu_i\leq d$.
By Lemma~\ref{lem:KnutsonTao}, it follows that $c_{\lambda\mu}^\nu>0$
because there exist $d\times d$ Hermitian 
matrices $A$ and $B$ such that
$A$, $B$, and $A+B$ have eigenvalues $\{\lambda_i\}_{i=1}^d$,
$\{\mu_i\}_{i=1}^d$ and $\{\nu_i\}_{i=1}^d$
respectively; namely:
\begin{eqnarray*}
A &=& \diag(\lambda_1,\lambda_2,\ldots,\lambda_d)\\
B &=& \diag(\mu_d,\mu_{d-1},\ldots,\mu_1).
\end{eqnarray*}
Furthermore we have $\nu\in \cP_{\lambda\mu}^{(d)}$ since $\nup_1\leq d$. 
Now by the Littlewood--Richardson identity of Lemma~\ref{lem:lrsymmetries}, we have $c_{\lambda'\mu'}^{\nu'}=c_{\lambda\mu}^\nu$ and thus, since $\lambda$ and $\mu$ are self-conjugate, 
$c_{\lambda\mu}^{\nu'}>0$.
Finally since by 
construction $\nu_1\leq d$, we have
$(\nu')'_{\ 1}\leq d$ so $c_{\lambda\mu}^{\nu'(d)}>0$ and $\nu'\in \cP_{\lambda\mu}^{(d)}$.
\end{proof}

It seems somewhat of a coincidence that the case for which conjugate partitions in the set $\cP_{\lambda\mu}^{(d)}$ prohibit universality -- that of a self-conjugate pair $(\lambda,\mu)=(\lambda',\mu')$ -- is also a case for 
which there exist arithmetically simple necessary and sufficient conditions on $(\lambda,\mu)$ 
by which to exclude such conjugate partitions from $\cP_{\lambda\mu}^{(d)}$.
When $(\lambda,\mu)\neq(\lambda',\mu')$, the conjugate-excluding condition 
$\max_{i+j=1+d}\{\lambda_i+\mu_j\}>d$ of Lemma~\ref{lem:iffnoconjugates}
is no longer a necessary one, as can be verified by the counterexample 
$\lambda=\mu=\ydiagram{4,1}$ with $d=4$.
In this more general case it does not appear possible to derive, from Weyl's inequalities, 
necessary and sufficient conditions by which to exclude conjugates. 
So we are fortunate that such conditions are not needed for universality.

With all of the pieces now in place, it is rather straightforward to derive the promised 
necessary and sufficient conditions for universality.  What follows are two conditions 
which can be simply stated in terms of shapes, as in Theorem~\ref{thm:iffumulti}:  Given partitions $\lambda$ and $\mu$, 
each of which has more than one row and more than one column but no more than $d$ rows, 
$(\lambda,\mu)$ is universal on $\cP_{\lambda\mu}^{(d)}$ if and only if 
\begin{enumerate}[label=(\roman{*})]
\item $\cP_{\lambda\mu}^{(d)}$ is devoid of hooks, and
\item if $\lambda$ and $\mu$ are both self-conjugate then $\cP_{\lambda\mu}^{(d)}$ is also devoid of 
conjugates.  
\end{enumerate}
The purpose of stating the theorem below in arithmetic terms is to 
compactly express the conditions on $\lambda$ and $\mu$ without reference to $\cP_{\lambda\mu}^{(d)}$, 
thus obviating the inefficient task of calculating Littlewood--Richardson 
coefficients when applying the theorem. 

\begin{theorem}
\label{thm:iff2partitionsu}
Given partitions $\lambda$ and $\mu$ such that $\dim(\blambda)>1$ and
$\dim(\bmu)>1$, $(\lambda,\mu)$ is $d$-universal if and only if the following conditions are satisfied:
\begin{enumerate}[label=(\roman{*})]
\item if $\max\{\lambda_2,\mu_2\}=1$ then $\lambdap_1+\mup_1>d+1$; 
\item if $(\lambda,\mu)=(\lambda',\mu')$ then $\max_{i+j=1+d}\{\lambda_i+\mu_j\}>d$.
\end{enumerate}
\end{theorem}
\begin{proof}
We prove each direction of the biconditional in turn:

\noindent  $\Rightarrow:$  We prove the contrapositive.  Assume (i) or (ii) is false.
If (i) is false, then $\max\{\lambda_2,\mu_2\}=1$ and $\lambdap_1+\mup_1\leq d+1$.  
Then by Lemma~\ref{lem:nohooks}, $\cP_{\lambda\mu}^{(d)}$ contains a hook,
and therefore by Lemma~\ref{lem:ifhooknotu} $\cP_{\lambda\mu}^{(d)}$ does not admit the universality of $(\lambda,\mu)$.

If (ii) is false then $(\lambda,\mu)=(\lambda',\mu')$ and $\max_{i+j=1+d}\{\lambda_i+\mu_j\}\leq d$.  Then by Lemma~\ref{lem:iffnoconjugates}, $\cP_{\lambda\mu}^{(d)}$ contains self-conjugate
partitions or conjugate pairs, and so by Lemma~\ref{lem:mueqmupnueqnup} or Lemma~\ref{lem:mueqmupnuneqnup}, $\cP_{\lambda\mu}^{(d)}$ does not
admit universality.

\noindent $\Leftarrow:$  Assume (i) and (ii) are true.
By (i) and Lemma~\ref{lem:nohooks}, $\cP_{\lambda\mu}^{(d)}$ does not contain any hooks.
By (ii), either $(\lambda,\mu)\neq(\lambda',\mu')$, 
or $(\lambda,\mu)=(\lambda',\mu')$
and $\max_{i+j=1+d}\{\lambda_i+\mu_j\}>d$.  In the former case, $\cP_{\lambda\mu}^{(d)}$
admits universality by Lemma~\ref{lem:muneqmupuniversal}.  In the latter case, $\cP_{\lambda\mu}^{(d)}$ has no
conjugates by Lemma~\ref{lem:iffnoconjugates}, and thus admits universality by Lemma~\ref{lem:ifnohooksnorconjthenu}.
\end{proof}

\subsection{Multitudinous nontrivial partitions}

Having obtained arithmetically simple necessary and sufficient conditions for the universality of a single logical qudit (Theorem~\ref{thm:iff1partitionu}) and two logical qudits (Theorem~\ref{thm:iff2partitionsu}),
it is naturally of interest to ask whether we can derive similarly simple necessary and 
sufficient conditions for the universality of arbitrarily many logical qudits.
First recall that by Lemma~\ref{lem:ifhooknotu}, if $\cP_\mus^{(d)}$ contains a hook then it does not admit
universality of a family $\{\mu^{(i)}\}_{i=1}^N$ of multiple nontrivial partitions.
This justifies generalizing the first condition of Theorem~\ref{thm:iff2partitionsu} from two 
partitions to arbitrarily many partitions, which can be done using the hook-avoiding 
conditions given by Lemma~\ref{lem:nohooks}.

Generalizing the second condition of Theorem~\ref{thm:iff2partitionsu} is more challenging.
One approach is to generalize Lemma~\ref{lem:iffnoconjugates} by iteratively applying Weyl's
inequalities, resulting in an expression of nested maxima.  Such an approach
indeed yields sufficient conditions for universality, but not necessary 
conditions.

For example, consider the generalized condition of Lemma~\ref{lem:iffnoconjugates} for three partitions:  
$$\max_{j+m=1+d}\{\max_{k+\ell=j+d}\{\kappa_k+\lambda_\ell\}+\mu_m\}>d.$$
Let $\kappa=\ydiagram{2,2}$, $\lambda=\ydiagram{2,1}$, and $\mu=\ydiagram{2,1}$, with $d=3$.  In that case the left hand side above equals $d$, thus failing
the inequality.  Yet, by direct computation, 
$\cP_{\kappa\lambda\mu}^{(3)}$
is in
fact devoid of conjugates and hooks, and thus $(\ydiagram{2,2},\ydiagram{2,1},\ydiagram{2,1})$
is 3-universal. 

It appears, then, that the Weyl inequalities are not enough to obtain 
necessary and sufficient conditions to exclude conjugates.  This suggests
utilization of a larger subset of the Horn inequalities.  But whether
an efficient algorithm to determine universality can be obtained in this way,
or at all, is left as an open problem.

We can, however, readily obtain useful sufficient conditions for the universality of arbitrarily many logical qudits.  
In particular, for sufficiently many
logical qudits of dimension greater than 1, universality is guaranteed.
To prove this, 
we need to show that
a sufficiently large family $\{\mu^{(i)}\}_{i=1}^N$ of partitions rules out 
conjugates in $\cP_\mus^{(d)}$.  First we show that a sufficiently large partition in $\cP_\mus^{(d)}$ 
rules out its own conjugate in $\cP_\mus^{(d)}$.  Clearly if the number $n$ being partitioned is too large relative to the bound $d$ on the number of rows, then there is not enough ``room" 
for both a partition and its conjugate in $\cP_\mus^{(d)}$.  This can also be understood as a 
simple corollary of Weyl's inequality on two partitions:
\begin{lemma}
\label{cor:ifd2noconjugates}
Given partitions $\lambda$ and $\mu$,
if $|\lambda|+|\mu|>d^2$ and $c_{\lambda\mu}^{\nu(d)}>0$ then we have
$c_{\lambda\mu}^{\nu'(d)}=0$.
\end{lemma}
\begin{proof}
Suppose that $\max\limits_{i+j=1+d}\{\lambda_i+\mu_j\}\leq d$.
Then we have
\begin{eqnarray*}
\lambda_1+\mu_d &\leq&  d\\
\lambda_2+\mu_{d-1} &\leq&  d\\
&\vdots& \\
\lambda_d+\mu_1 &\leq&  d.
\end{eqnarray*}
Summing the above yields $|\lambda|+|\mu|\leq d^2$.
Contrapositively, if $|\lambda|+|\mu|>d^2$ then
$\max\limits_{i+j=1+d}\{\lambda_i+\mu_j\}>d$, 
and in turn 
by Weyl's inequality we have $\nu_1>d$.  Therefore it follows that
$(\nu')'_{\ 1}>d$, and by definition $c_{\lambda\mu}^{\nu'(d)}=0$.
\end{proof}

The above result obviously generalizes to arbitrarily many partitions:
\begin{lemma}
\label{lem:manynoconj}
Let $\{\mu^{(i)}\vdash m_i\}_{i=1}^N$ be a family of partitions satisfying
$\mup_1\leq d$.
If $\sum_{i=1}^Nm_i>d^2$ then $P=\{\nu\vdash\sum_{i=1}^Nm_i|\nup_1\leq d,\Hom(\bigotimes_{i=1}^N\bmu^{(i)},\bnu)\neq\{0\}\}$
contains no conjugates.
\end{lemma}
\begin{proof}
By Lemma~\ref{lem:frobenius}, we have
$$P=\{\nu|c_{\mu^{(1)}\mu^{(2)}}^{\nu^{(2)}}c_{\nu^{(2)}\mu^{(3)}}^{\nu^{(3)}}\cdots c_{\nu^{(N-2)}\mu^{(N-1)}}^{\nu^{(N-1)}}c_{\nu^{(N-1)}\mu^{(N)}}^{\nu}>0\}.$$
Given
$$c_{\mu^{(1)}\mu^{(2)}}^{\nu^{(2)}}c_{\nu^{(2)}\mu^{(3)}}^{\nu^{(3)}}\cdots c_{\nu^{(N-2)}\mu^{(N-1)}}^{\nu^{(N-1)}}c_{\nu^{(N-1)}\mu^{(N)}}^{\nu}>0,$$
by the Littlewood--Richardson rules we have $|\nu^{(N-1)}|=\sum_{i=1}^{N-1}|\mu^{(i)}|$ 
Thus by hypothesis we have
$$|\nu^{(N-1)}|+|\mu^{(N)}|>d^2.$$
Then by Lemma~\ref{cor:ifd2noconjugates} we have
$c_{\nu^{(N-1)}\mu^{(N)}}^{\nu'(d)}=0$.
It follows that $P$ contains no conjugates.
\end{proof}

Finally we arrive at the result that universality is inevitable 
as the number of logical qudits is scaled up:
\begin{theorem}
Given a family of nontrivial partitions $\{\mu^{(i)}\}_{i=1}^N$ of at most $d$ rows, 
if $\max\{\mu^{(i)}_2\}>1$ or $\sum_{i=1}^N\mu^{(i)\prime}_1\geq N+d$ and $\sum_{i=1}^N|\mu^{(i)}|>d^2$
then $\{\mu^{(i)}\}_{i=1}^N$ is $d$-universal. 
\label{thm:ifnohookslargenu}
\end{theorem}
\begin{proof}
By Lemma~\ref{lem:nohooks}, $\cP_\mus^{(d)}$ contains no hooks, and by 
Lemma~\ref{lem:manynoconj}, $\cP_\mus^{(d)}$ contains no conjugates.
Therefore by Lemma~\ref{lem:ifnohooksnorconjthenu}, $\{\mu^{(i)}\}_{i=1}^N$ is universal on $\cP_\mus^{(d)}$, and
thus by definition $d$-universal.
\end{proof}
\noindent This theorem has a corollary which makes the inevitability more obvious,
and may be useful when only the number of logical qudits is known:
\begin{corollary}
\label{cor:ifmanyu}
Given a family of nontrivial partitions $\{\mu^{(i)}\}_{i=1}^N$ of at most $d$ rows, 
if $N>d^2/3$ then $\{\mu^{(i)}\}_{i=1}^N$ is $d$-universal.
\end{corollary}
\begin{proof}
If $d=2$ then $\{\mu^{(i)}\}_{i=1}^N$ is universal on $\cP_\mus^{(d)}$ by Theorem~\ref{thm:kempe}.
If $d>2$ then we have $N>d$.
Since each $\mu^{(i)}$ is nontrivial, we have $|\mu^{(i)}|\geq 3$, and thus
\begin{eqnarray*}
\sum_{i=1}^N|\mu^{(i)}| &\geq& 3N>N+d,\\
&>& 3(d^2/3),\\
&>& d^2.
\end{eqnarray*}
\end{proof}

\subsection{Characterizing universal families}

The problem of finding necessary and sufficient conditions for the universality of a family of
nontrivial partitions may be usefully reframed by the observation that universal families of
multiple nontrivial partitions are upward closed:
\begin{theorem}
\label{thm:upwardclosed}
Given a family $\{\mu^{(i)}\}_{i=1}^N$ of nontrivial partitions, if any subfamily containing at least two
partitions is $d$-universal, then $\{\mu^{(i)}\}_{i=1}^N$ is also
$d$-universal.
\end{theorem}
\begin{proof}
Without loss of generality, for some $2\leq k\leq N$ let $\{\mu^{(i)}\}_{i=1}^k$ be a subfamily of
$\{\mu^{(i)}\}_{i=1}^N$ that is universal on $\cP_\musk^{(d)}$.
Then by Lemma~\ref{lem:ifhooknotu}, $\cP_\mus^{(d)}$ cannot contain any hooks.
Thus if $\{\mu^{(i)}\}_{i=1}^k\neq\{\mu^{(i)\prime}\}_{i=1}^k$ then universality of
$\{\mu^{(i)}\}_{i=1}^N$ on $\cP_\mus^{(d)}$ follows immediately from Lemma~\ref{lem:muneqmupuniversal}.

Assume that $\{\mu^{(i)}\}_{i=1}^k=\{\mu^{(i)\prime}\}_{i=1}^k$.  By Lemma~\ref{lem:mueqmupnueqnup} and Lemma~\ref{lem:mueqmupnuneqnup}, $\cP_\musk^{(d)}$ has no
conjugates.  But by symmetry of the Littlewood--Richardson coefficients, $\cP_\musk$ necessarily
contains the conjugate of all its partitions.  Consistency thus demands that if $\nu\in\cP_\musk^{(d)}\subset\cP_\musk$, then $\nup_1>d$.  Then by Lemma~\ref{lem:nugeqmax} we also have that $\nup_1>d$ for all
$\nu\in\cP_\mus^{(d)}$.  So $\cP_\mus^{(d)}$ contains no conjugates and by Lemma~\ref{lem:ifnohooksnorconjthenu} $\{\mu^{(i)}\}_{i=1}^N$ is again $d$-universal.
\end{proof}

Thus the problem of characterizing universal families of multiple nontrivial partitions reduces to the problem of characterizing 
minimal such families.  For that purpose, it is convenient to consider self-conjugate and non-self-conjugate universal families 
separately. The latter can be characterized as follows:
\begin{theorem}
\label{thm:iffunonselfcon}
A non-self-conjugate family of multiple nontrivial partitions $\{\mu^{(i)}\}_{i=1}^N$ of at most $d$ rows is $d$-universal if and only if $\max\{\mu^{(i)\prime}_1\}>1$, or $\sum_{i=1}^N\mu^{(i)\prime}_1\geq N+d$.
\end{theorem}
\begin{proof}
Assume $\max\{\mu^{(i)\prime}_1\}=1$ and $\sum_{i=1}^N\mu^{(i)\prime}_1<N+d$.
Then by Lemma~\ref{lem:nohooks}, $\cP_\mus^{(d)}$ contains a hook,
and by Lemma~\ref{lem:ifhooknotu}, $\{\mu^{(i)}\}_{i=1}^N$ is not $d$-universal.

Conversely, assume $\max\{\mu^{(i)\prime}_1\}>1$, or $\sum_{i=1}^N\mu^{(i)\prime}_1\geq N+d$.
Then by Lemma~\ref{lem:nohooks}, $\cP_\mus^{(d)}$ contains no hooks.
In turn, by Lemma~\ref{lem:muneqmupuniversal}, $\{\mu^{(i)}\}_{i=1}^N$ is universal on $\cP_\mus^{(d)}$.
\end{proof}

Clearly for each $d$ there is an infinite number of minimal universal families of nontrivial partitions containing at least one 
non-self-conjugate partition, since that partition can have arbitrarily large rows.
But since the non-self-conjugate universal families of nontrivial partitions are already efficiently characterized by Theorem~\ref{thm:iffunonselfcon}, 
we turn our attention now to
the set of minimal families of multiple, nontrivial, self-conjugate partitions that are $d$-universal, which we denote $F^{(d)}$.
Conveniently, we have
\begin{prop}
For any $d>0$, $F^{(d)}$ is finite. 
\end{prop}
\begin{proof}
Since the size of each self-conjugate partition of at most $d$ rows is bounded by $d^2$,
the number of such partitions is finite.  Furthermore the families in $F^{(d)}$
must be bounded since if they were not, each such arbitrarily large family would have an arbitrarily
large, nonuniversal subfamily, contradicting the effective bound of $d^2/3$ on nonuniversal
families given by Corollary~\ref{cor:ifmanyu}.
\end{proof}

\begin{exmp}
Consider $d=2$.
By Theorem~\ref{thm:kempe}, all families of partitions are 2-universal.
It follows that the minimal universal families of multiple self-conjugate nontrivial partitions
include $(\ydiagram{2,1},\ydiagram{2,1})$, $(\ydiagram{2,1},\ydiagram{2,2})$,
and $(\ydiagram{2,2},\ydiagram{2,2})$.  Since every
family of multiple self-conjugate nontrivial partitions contains one or more of these three subfamilies, they are the only such minimal families.
Therefore we have
$$
F^{(2)}=\{(\ydiagram{2,1},\ydiagram{2,1}), (\ydiagram{2,1},\ydiagram{2,2}), (\ydiagram{2,2},\ydiagram{2,2})\}.
$$
\end{exmp}

\begin{exmp}
Consider $d=3$.
By Lemma~\ref{lem:ifhooknotu} and Theorem~\ref{thm:iff2partitionsu}, the only family in $F^{(3)}$ of at most nine total cells is $(\ydiagram{2,2},\ydiagram{2,2})$.
All families in $F^{(3)}$ with more than nine total
cells are found to be 3-universal by Theorem~\ref{thm:ifnohookslargenu}.
To find these families, note there are only six nontrivial, self-conjugate 
partitions of at most three rows.  Since each nontrivial partition has at least three cells, all families of at least four nontrivial partitions are 3-universal.  Therefore all families of at least five nontrivial partitions contain a 3-universal subfamily, and are thus not minimal. Thus we need only look for families of at most four partitions, containing a total of more than nine cells, but lacking any subfamily containing more than nine cells.  By such accounting we finally obtain:
\begin{eqnarray*}
F^{(3)} &=& \{(\ydiagram{2,2},\ydiagram{2,2}),(\ydiagram{2,1},\ydiagram{3,3,2}),(\ydiagram{2,1},\ydiagram{3,3,3}),(\ydiagram{2,2},\ydiagram{3,2,1}),(\ydiagram{2,2},\ydiagram{3,3,3}),(\ydiagram{3,1,1},\ydiagram{3,1,1}),(\ydiagram{3,1,1},\ydiagram{3,2,1}),(\ydiagram{3,1,1},\ydiagram{3,3,2}),\\
&& (\ydiagram{3,1,1},\ydiagram{3,3,3}),(\ydiagram{3,2,1},\ydiagram{3,2,1}),(\ydiagram{3,2,1},\ydiagram{3,3,2}),(\ydiagram{3,2,1},\ydiagram{3,3,3}),(\ydiagram{3,3,2},\ydiagram{3,3,2}),(\ydiagram{3,3,2},\ydiagram{3,3,3}),(\ydiagram{3,3,3},\ydiagram{3,3,3}),(\ydiagram{2,1},\ydiagram{2,1},\ydiagram{2,2})\\
&& (\ydiagram{2,1},\ydiagram{2,1},\ydiagram{3,1,1}),(\ydiagram{2,1},\ydiagram{2,1},\ydiagram{3,2,1}),(\ydiagram{2,1},\ydiagram{2,2},\ydiagram{2,2}),(\ydiagram{2,1},\ydiagram{2,2},\ydiagram{3,1,1}),(\ydiagram{2,2},\ydiagram{2,2},\ydiagram{2,2}),(\ydiagram{2,2},\ydiagram{2,2},\ydiagram{3,1,1}),(\ydiagram{2,1},\ydiagram{2,1},\ydiagram{2,1},\ydiagram{2,1})\}.
\end{eqnarray*}
\end{exmp}

\begin{exmp}
Consider $d=4$.
The two-partition families in $F^{(4)}$ can be determined by Theorem~\ref{thm:iff2partitionsu},
and many of the larger families in $F^{(4)}$ can be determined by Theorem~\ref{thm:ifnohookslargenu}.
The only remaining families in $F^{(4)}$ are found by direct computation to be
$$
\left(\ydiagram{4,1,1,1},\ydiagram{2,2},\ydiagram{2,2}\right),\left(\ydiagram{4,2,1,1},\ydiagram{2,2},\ydiagram{2,2}\right),\left(\ydiagram{3,3,2},\ydiagram{2,2},\ydiagram{2,2}\right).
$$
\end{exmp}

\section{Redefining universality}

In this section we prove the equivalence of weak universality with strong universality.
Clearly strong universality implies weak universality, since a strong
universality-witnessing map is just a special case of a weak universality-witnessing map.
So all cases previously shown to be strongly universal are also weakly universal.
The task remaining is to show that in all other cases weak universality implies strong
universality.  These cases are, specifically:
\begin{itemize}
\item Any family of self-conjugate partitions on a 
self-conjugate partition.
\item Any family of self-conjugate partitions on a pair
of unequal conjugate partitions.
\item Any family of partitions on a set of hooks.
\end{itemize}

In order to tackle the cases involving conjugates, it is helpful to decompose the alternating
intertwiner $M_\nu$ as follows.
\begin{lemma}
\label{lem:Mnufactors}
Given a family of self-conjugate partitions $\{\mu^{(i)}\}_{i=1}^N$, the alternating interwiner
$M_\nu:\nu\rightarrow\nu'$ restricted to $V_\mus^\nu$ has the form
$\chimnp(\bigotimes_{i=1}^N M_{\bmu^{(i)}}\otimes M_{\Hom})(\chimn)^{-1}$,
where $M_{\Hom}$ is an isomorphism from $\Hom(\bigotimes_{i=1}^N\bmu^{(i)},\bnu)$ to $\Hom(\bigotimes_{i=1}^N\bmu^{(i)},\bnu')$.
\end{lemma}
\begin{proof}
Given any $\iota\in\Hom(\bigotimes_{i=1}^N\bmu^{(i)},\bnu)$, define $\iota'\in\Hom(\bigotimes_{i=1}^N\bmu^{(i)},\bnu')$ so that for any $\psi\in\bigotimes_{i=1}^N\bmu^{(i)}$, we have
$$\iota'(\psi)=\chimnp(\bigotimes_{i=1}^N M_{\mu^{(i)}}^{-1}\otimes\mathbb{1})(\chimnp)^{-1}M_\nu\iota(\psi).$$
To verify $\iota'$ is an $\prod_{i=1}^N S_{m_i}$-intertwiner, where $m_i=|\mu^{(i)}|$,
let
$s\in \prod_{i=1}^N S_{m_i}$.  Then we have
\begin{eqnarray*}
\iota'(s.\psi) &=& \chimnp(\bigotimes_{i=1}^N M_{\mu^{(i)}}^{-1}\otimes\mathbb{1})(\chimnp)^{-1}M_\nu\iota(s.\psi)\\
&=& \chimnp(\bigotimes_{i=1}^N M_{\mu^{(i)}}^{-1}\otimes\mathbb{1})(\chimnp)^{-1}M_\nu s.\iota(\psi)\\
&=& \chimnp(\bigotimes_{i=1}^N M_{\mu^{(i)}}^{-1}\otimes\mathbb{1})(\chimnp)^{-1}\epsilon(s)s.M_\nu\iota(\psi)\\
&=& \epsilon(s)s.\chimnp(\bigotimes_{i=1}^N M_{\mu^{(i)}}^{-1}\otimes\mathbb{1})(\chimnp)^{-1}\epsilon(s)M_\nu\iota(\psi)\\
&=& s.\chimnp(\bigotimes_{i=1}^N M_{\mu^{(i)}}^{-1}\otimes\mathbb{1})(\chimnp)^{-1}M_\nu\iota(\psi)\\
&=& s.\iota'(\psi).
\end{eqnarray*}
Now define $M_{\Hom}$ as the transformation of intertwiners taking $\iota$ to $\iota'$, as given above.
Then
rearranging terms in the
definition of $\iota'(\psi)$ we obtain
$$
\chimnp\bigotimes_{i=1}^N(M_{\mu^{(i)}}\otimes M_{\Hom})\psi\otimes\iota=M_\nu\chimn\psi\otimes\iota.
$$  
Since this holds for any $\iota\in\Hom(\bigotimes_{i=1}^N\bmu^{(i)},\bnu)$ and $\psi\in\bigotimes_{i=1}^N\bmu^{(i)}$, it follows that $M_\nu$ restricted to $V_\mus^\nu$ 
must factor as advertised.
\end{proof}
\noindent With this expression for $M_\nu$ in hand, we can show that weak universality is equivalent to strong 
universality when all partitions are self-conjugate.
\begin{lemma}
\label{lem:weakeqstrong4selfconj}
Given a family of self-conjugate partitions $\{\mu^{(i)}\}_{i=1}^N$ and a self-conjugate,
proper partition $\nu$, it follows that
$\{\mu^{(i)}\}_{i=1}^N$ is weakly universal on $\nu$
if and only if it is strongly universal on $\nu$.
\end{lemma}
\begin{proof}
Assume $\{\mu^{(i)}\}_{i=1}^N$ is weakly universal on $\nu$.
Then by definition, for all $u\in\su(\bigotimes_{i=1}^N\bmu^{(i)})$ 
there exists $t\in\End(\Hom(\bigotimes_{i=1}^N\bmu^{(i)},\bnu))$
such that 
$\chimn(u\otimes\mathbb{1}+\mathbb{1}\otimes t)(\chimn)^{-1}\in\rho_\nu(\g_n)$.
By Lemma~\ref{lem:osp}, membership in $\rho_\nu(\g_n)$ implies that
$$M_\nu\chimn(u\otimes\one+\one\otimes t)(\chimn)^{-1}M_\nu^{-1}=-\chimn(u^\tp\otimes\one+\one\otimes t^\tp)(\chimn)^{-1},$$
where we have used the fact that $\chi_\mus^\nu$ is an orthogonal transformation
(by Lemma~\ref{lem:orthogonalchi}).
Using the factorization $M_\nu=\bigotimes_{i=1}^N M_{\mu^{(i)}}\otimes M_{\Hom}$ from Lemma~\ref{lem:Mnufactors},
the above can be rewritten as
$$u'\otimes\one+\one\otimes t'=-u^\tp\otimes\one+\one\otimes t^\tp$$
where $u'=\bigotimes_{i=1}^N M_{\mu^{(i)}} u\bigotimes_{i=1}^N M_{\mu^{(i)}}^{-1}$
and $t'=M_{\Hom} tM_{\Hom}^{-1}$. 
Rearranging terms, we have
$$(u'+u^\tp)\otimes\one=-\one\otimes(t'+t^\tp).$$
This can only be satisfied if both sides are scalar.  
The tracelessness of the left hand side further requires 
that both sides vanish.  Thus $u'=-u^\tp$, which implies $\chimn(u\otimes\one)(\chimn)^{-1}\in\rho_\nu(\g_n)$ and thus strong universality.

\end{proof}
\noindent A similar argument applies to conjugate pairs. 
\begin{lemma}
\label{lem:weakeqstrong4conjpairs}
Given a family of self-conjugate partitions $\{\mu^{(i)}\}_{i=1}^N$ and a partition $\nu$ such that $\nu\neq\nu'$,
if $\{\mu^{(i)}\}_{i=1}^N$ is weakly universal
on $\{\nu,\nu'\}$, then 
it is strongly universal on $\{\nu,\nu'\}$.
\end{lemma}
\begin{proof}
Assume $\{\mu^{(i)}\}_{i=1}^N$ is weakly universal on $\{\nu,\nu'\}$.
Then by definition for all $u\in\su(\bigotimes_{i=1}^N\bmu^{(i)})$, there exists
$x\in\g_n$, $t\in\End(\Hom(\bigotimes_{i=1}^N\bmu^{(i)},\bnu))$, and $t'\in\End(\Hom(\bigotimes_{i=1}^N\bmu^{(i)},\bnu'))$ such that
\begin{eqnarray*}
(\chimn)^{-1}\rho_\nu(x)\chimn &=& u\otimes\one+\one\otimes t\\
(\chimnp)^{-1}\rho_{\nu'}(x)\chimnp &=& u\otimes\one+\one\otimes t'.
\end{eqnarray*}
However by the alternating interwiner $M_\nu$ between $\nu$ and $\nu'$ given by Lemma~\ref{lem:isoMorphism}, and its factorization given by Lemma~\ref{lem:Mnufactors}, we have
\begin{eqnarray*}
(\chimnp)^{-1}\rho_{\nu'}(x)\chimnp &=& -(\chimnp)^{-1}M_\nu\rho_\nu(x)^\tp M_\nu^{-1}\chimnp\\
&=& -M_\mus\otimes M_{\Hom}(u^\tp\otimes\one+\one\otimes t^\tp)M_\mus^{-1}\otimes M_{\Hom}^{-1}\\
&=& -M_\mus u^\tp M_\mus^{-1}\otimes\one-\one\otimes M_{\Hom} t^\tp M_{\Hom}^{-1},
\end{eqnarray*}
where $M_\mus\equiv \bigotimes_{i=1}^N M_{\mu^{(i)}}$.
Weak universality thus implies
$$
u\otimes\one+\one\otimes t'=-M_\mus u^\tp M_\mus^{-1}\otimes\one-\one\otimes M_{\Hom} t^\tp M_{\Hom}^{-1},
$$
or, rearranging,
$$
(u+M_\mus u^\tp M_\mus^{-1})\otimes\one=-\one\otimes(t'+M_{\Hom} t^\tp M_{\Hom}^{-1}).
$$
This is only satisfied if both sides are scalar.  The tracelessness of the
left hand side further implies $u=-M_\mus u^\tp M_\mus^{-1}$.
This implies that $u$ and $-u^\tp$ are related by a similarity transformation.
If $u=\diag(1,1,-2,0,\ldots,0)$, we therefore arrive a contradiction.
Thus if $\dim(\bigotimes_{i=1}^N\bmu^{(i)})>2$ weak universality must fail just as, by 
Lemma~\ref{lem:mueqmupnuneqnup}, strong universality fails.
On the other hand if $\dim(\bigotimes_{i=1}^N\bmu^{(i)})\leq 2$ then $\{\mu^{(i)}\}_{i=1}^N$
consists of at most one partition equal to $\ydiagram{2,1}$ or $\ydiagram{2,2}$
and the rest equal to $\ydiagram{1}$.  In this case strong universality holds, because it holds for
$\ydiagram{2,1}$ and $\ydiagram{2,2}$ by Theorem~\ref{thm:iff1partitionu}.

\end{proof}

To prove weak-strong equivalence for hooks, we first prove it for a broader
class of partitions.  The following proof makes use of three essential facts:
any representation of $\g_n$ is closed under the commutator,
the special unitary algebra is its own derived algebra,
and the only traceless one-dimensional operator is 0.
\begin{lemma}
\label{lem:weakeqstrong4ceq1}
If a family of partitions $\{\mu^{(i)}\}_{i=1}^N$ is weakly universal
on 
$$\cP\subset\{\nu|\dim(\Hom(\bigotimes_{i=1}^N\bmu^{(i)},\bnu))=1\},$$
then $\{\mu^{(i)}\}_{i=1}^N$ is strongly universal on $\cP$.
\end{lemma}
\begin{proof}
Let $\varphi:\su(\bigotimes_{i=1}^N\bmu^{(i)})\rightarrow(\bigoplus_{\nu\in\cP}\rho_\nu)(\g_n)$
be a weak-universality witnessing function, $\h$ its image,
and $\h'=[\h,\h]$.
Then $\h'\subset\bigoplus_{\nu\in\cP}\rho_\nu(\g_n)$ and thus
$\h'\Pi\subset\bigoplus_{\nu\in\cP}\rho_\nu(\g_n)\Pi$,
where $\Pi$ is the projection onto $\bigoplus_{\nu\in\cP}V_\mus^\nu$.  By Lemma~\ref{lem:canonicalweak} and the hypothesis that $\Hom(\bigotimes_{i=1}^N\bmu^{(i)},\bnu)=\cmplx$, it follows that for all $u\in\su(\bigotimes_{i=1}^N\bmu^{(i)})$, 
$\bigoplus_{\nu\in\cP}\chi_\mus^\nu(u\otimes\one+\one\otimes\one t_{\nu,u})(\chi_\mus^\nu)^{-1}\in\h\Pi$ for some $t_{\nu,u}\in\cmplx$. Taking the commutator of both sides, it follows that 
$$\left\{\bigoplus_{\nu\in\cP}\chi_\mus^\nu(u\otimes\one)(\chi_\mus^\nu)^{-1}|u\in\su(\bigotimes_{i=1}^N\bmu^{(i)})\right\}\subset\h'\Pi,$$
where we have used that $\h$ commutes with $\Pi$, since it fixes the isotypical subspaces, and
$\su(\bigotimes_{i=1}^N\bmu^{(i)})$ is its own commutator.  So for all $u\in\su(\bigotimes_{i=1}^N\bmu^{(i)})$, we have
$$\bigoplus_{\nu\in\cP}\chi_\mus^\nu(u\otimes\one)(\chi_\mus^\nu)^{-1}\in\bigoplus_{\nu\in\cP}\rho_\nu(\g_n)\Pi.$$
Therefore by Lemma~\ref{lem:canonicalstrong}, $\{\mu^{(i)}\}_{i=1}^N$ is strongly universal on $\cP$.
\end{proof}
\noindent The equivalence of weak and strong universality for hooks then follows.
\begin{lemma}
\label{lem:weakeqstrong4hooks}
If a set of partitions $\{\mu^{(i)}\}_{i=1}^N$ is weakly universal on
a set of hooks $\cP\subset\cP_\mus$ then $\{\mu^{(i)}\}_{i=1}^N$ is strongly universal on $\cP$.
\end{lemma}
\begin{proof}
By Lemma~\ref{lem:hookfamilies}, every $\mu^{(i)}$ is improper, as well as every 
$\nu^{(i)}$ in every nonvanishing product of Littlewood--Richardson 
coefficients of the form 
$c_{\mu^{(1)}\mu^{(2)}}^{\nu^{(2)}}c_{\nu^{(2)}\mu^{(3)}}^{\nu^{(3)}}\cdots c_{\nu^{(N-2)}\mu^{(N-1)}}^{\nu^{(N-1)}}c_{\nu^{(N-1)}\mu^{(N)}}^{\nu}$.
Then by Lemma~\ref{lem:2hooks}, this product equals 1.  In turn by Lemma~\ref{lem:LRres}, we have
$$\cP\subset\{\nu|\dim(\Hom(\bigotimes_{i=1}^N\bmu^{(i)},\bnu))=1\}.$$
Finally by Lemma~\ref{lem:weakeqstrong4ceq1}, $\{\mu^{(i)}\}_{i=1}^N$ is strongly universal on $\cP$.
\end{proof}

Proving the equivalence of weak and strong universality is now just a matter of 
putting the above pieces together.
\begin{theorem}
\label{thm:weakeqstrong}
Weak universality is equivalent to strong universality.
\end{theorem}
\begin{proof}
Given a family of partitions $\{\mu^{(i)}\}_{i=1}^N$, 
and a set $\cP\subset\cP_\mus$, 
assume $\{\mu^{(i)}\}_{i=1}^N$ is weakly universal on $\cP$.
By Lemma~\ref{lem:weakeqstrong4hooks}, $\{\mu^{(i)}\}_{i=1}^N$ is strongly universal on the hooks in $\cP$.
If $\{\mu^{(i)}\}_{i=1}^N=\{\mu^{(i)\prime}\}_{i=1}^N$, 
then by Lemma~\ref{lem:weakeqstrong4selfconj}, $\{\mu^{(i)}\}_{i=1}^N$ is
also strongly universal on each self-conjugate partition in $\cP$.
Furthermore by Lemma~\ref{lem:ifnohooksnorconjthenu}, $\{\mu^{(i)}\}_{i=1}^N$
is strongly universal on each proper $\nu\in\cP$ such that $\nu\neq\nu'$
and so by Lemma~\ref{lem:weakeqstrong4conjpairs}, $\{\mu^{(i)}\}_{i=1}^N$ is also strongly universal on each proper
conjugate pair in $\cP$.
If on the other hand $\{\mu^{(i)}\}_{i=1}^N\neq\{\mu^{(i)\prime}\}_{i=1}^N$ then by Lemma~\ref{lem:muneqmupuniversal}, 
$\{\mu^{(i)}\}_{i=1}^N$ is strongly universal on all the proper partitions in $\cP$.
In either case it follows by Lemma~\ref{lem:additiveu} that $\{\mu^{(i)}\}_{i=1}^N$ is strongly 
universal on $\cP$.

Conversely if $\{\mu^{(i)}\}_{i=1}^N$ is strongly universal on $\cP$, then
clearly $\{\mu^{(i)}\}_{i=1}^N$ is weakly universal on $\cP$.

\end{proof}

\section{Engineering universality}

We have seen that universality is common and, with increasingly many nontrivial partitions, 
inevitable.
If universality is desired then one can either choose appropriate partitions, for example
by Theorem~\ref{thm:iff2partitionsu}, or choose sufficiently many partitions, in accordance with Corollary~\ref{cor:ifmanyu}.
If these are not options and one is bound to a particular family of partitions, 
there are still other ways to force universality.

\subsection{Ancillae}

Ancillae are auxiliary qudits that facilitate a computation in some way, even though their 
initial and final states are independent of the computation.  Similarly, a partition may be 
added to a family of partitions in order to make it universal, even though we are only 
interested in unitary operations on the space associated with the original family.  Since in 
its quantum application the cells of this auxiliary partition correspond to qudits, they may 
be described as ancillae.  We find that at most five ancillae are required to ensure universality.
\begin{theorem}
\label{thm:five}
Given any family of partitions $\{\mu^{(i)}\}_{i=1}^N$,
there exists a partition $\mu^{(N+1)}$ of size no greater than five
such that $\{\mu^{(i)}\}_{i=1}^{N+1}$ is universal on any nonempty set
$\cP\subset\cP_{\mu^{(1)}\cdots\mu^{(N+1)}}$. 
\end{theorem}
\begin{proof}
Let $\mu^{(N+1)}=\ydiagram{3,2}$.
Since we have $\mu^{(N+1)\prime}\neq\mu^{(N+1)}$, it follows that
$\{\mu^{(i)\prime}\}_{i=1}^{N+1}\neq\{\mu^{(i)}\}_{i=1}^{N+1}$.
Since furthermore $\mu^{(N+1)}$ is proper, it follows from Lemma~\ref{lem:hookfamilies} that
$\cP$ contains only proper partitions.
Therefore the conditions of Lemma~\ref{lem:muneqmupuniversal} are fulfilled,
implying the universality of $\{\mu^{(i)}\}_{i=1}^{N+1}$.
\end{proof}

In many cases fewer ancillae are needed.  For example, if hooks are not a concern, 
then adding the partition $\ydiagram{2}$ to the family suffices to render the family 
non-self-conjugate and thereby achieve universality.  
On the other hand if the family consists entirely of 
hooks, but is not self-conjugate, then adding $\ydiagram{2,2}$ is sufficient for universality.

\subsection{Highest weight states}

The Cartan subalgebra of a finite dimensional semisimple Lie algebra is a maximal abelian 
subalgebra whose elements are semisimple (i.e. their matrix forms in the adjoint representations are 
simultaneously diagonalizable). 
Within a given irreducible representation, the eigenvalues of a canonical set of
generators of the Cartan subalgebra are called weights.  A highest weight vector is a
simultaneous eigenvector of the Cartan generators such that each of the associated 
eigenvalues is greater than or equal to the corresponding eigenvalue for every other eigenvector.

Barnum et al. have shown, in quantum physical terms, that there always exists a Hamiltonian 
in a given Lie algebra such that its ground state is equal to a highest weight vector in a 
given irreducible representation \cite[Theorem~14]{PhysRevA.68.032308}.  This implies that, in principle, a physical system can be
engineered in such a way that its lowest energy state corresponds to a highest weight vector.
For example there exists a Hamiltonian in $\su(d)$ whose ground state equals a 
highest weight vector in $\bigotimes_{i=1}^NU^{(d)}_{\mu^{(i)}}$.
If the vectors in this representation correspond to spin states, then the needed Hamiltonian 
might be physically implemented with the aid of a uniform magnetic field, with a highest 
weight vector corresponding to aligned spins.  By such means, qudits can be physically 
prepared in an initial state corresponding to a highest weight vector.

The purpose of giving a particular vector of a representation such preferred status is to 
correspondingly give preferred status to the irreducible representation to which that vector maps by 
intertwiner.  As the representations relevant to qudits correspond to products of special unitary 
representations with symmetric group representations, both determined by the same partition,
constraining the former simultaneously constrains the latter.  The end result is that, when 
considering the universality of a family of partitions $\{\mu^{(i)}\}_{i=1}^N$  on 
$\cP_\mus^{(d)}$, the physically relevant 
subset is reduced to a single partition.  We show, furthermore, that this partition always 
admits universality of a family of multiple nontrivial partitions.  The key to this outcome 
is the following.

\begin{lemma}[Cartan]
Let $d\in\nats$.
Given partitions $\lambda$ and $\mu$ such that $\lambdap_1\leq d$ and
$\mup_1\leq d$,
a highest weight vector of the $\SU(d)$ representation
$U^{(d)}_\lambda\otimes U^{(d)}_\mu$ maps via the $\SU(d)$-equivariant isomorphism
$U^{(d)}_\lambda\otimes U^{(d)}_\mu\rightarrow\bigoplus_\nu c_{\lambda\mu}^\nu U^{(d)}_\nu$
to a highest weight vector of $U^{(d)}_{\lambda+\mu}$.
\end{lemma}

\noindent The above lemma is a fundamental property of the ``Cartan product", or highest
weight component of the tensor product.  In other words, the tensor product of highest weight
vectors of two irreducible representations $U^{(d)}_\lambda$ and $U^{(d)}_\mu$ always maps by $\SU(d)$-intertwiner into the Cartan product of
the irreducible representations, $U^{(d)}_{\lambda+\mu}$ \cite{eastwood2005}, \cite[Exercise~15.24]{fulton1991representation}, \cite[Proposition~3]{procesi2007lie}.

As with the isomorphism between the product and sum of unitary group 
representations (Lemma~\ref{lem:LRunitary}), the above mapping between highest weight vectors 
likewise generalizes to a product of arbitrarily many vectors by iteration.
So a highest weight vector of $\bigotimes_{i=1}^NU^{(d)}_{\mu^{(i)}}$
maps to a highest weight vector of $U^{(d)}_{\mu^{(1)}+\cdots+\mu^{(N)}}$.
It follows that a highest weight vector of a product of qudits,
$\bigotimes_{i=1}^k(V_{\mu^{(i)}}\otimes U^{(d)}_{\mu^{(i)}})$,
maps to a highest weight vector of 
$\bigotimes_{i=1}^kV_{\mu^{(i)}}\otimes U^{(d)}_{\mu^{(1)}+\cdots+\mu^{(N)}}$.
The latter in turn maps by a $S_{m_1}\times\cdots S_{m_N}\times\SU(d)$-intertwiner to $V_{\mu^{(1)}+\cdots+\mu^{(N)}}\otimes U^{(d)}_{\mu^{(1)}+\cdots+\mu^{(N)}}$
in the Hilbert space of the full physical system, as discussed in Chapter 3.
So in such an application a highest weight vector of the unitary group 
representation determines the relevant representation of the symmetric 
groups as well.  The utility of such control over the 
representation is made apparent by the following.
\begin{theorem}
\label{thm:highestweight}
Given a family of partitions $\{\mu^{(i)}\}_{i=1}^N$ such that
$N>1$ and $\dim(\bmu^{(i)})>1$ for all $1\leq i\leq N$, 
it follows that
$\{\mu^{(i)}\}_{i=1}^N$ is universal on $\sum_{i=1}^N\mu^{(i)}$.
\end{theorem}
\begin{proof}
Let $\nu=\sum_{i=1}^N\mu^{(i)}$.
It follows from Lemma~\ref{lem:LRCartan} that $c_\mus^\nu>0$ and thus there exists 
a nontrivial intertwiner from $\bigotimes_{i=1}^N\bmu^{(i)}$ to $\nu$.

Since $\dim(\bmu^{(i)})>1$ it follows that $\mu^{(i)}_2\geq 1$, and since
furthermore $N>1$, it follows that $\nu_2=\sum_{i=1}^N\mu^{(i)}_2>1$.
So $\nu$ is proper.

If furthermore we have $\{\mu^{(i)\prime}\}_{i=1}^N\neq\{\mu^{(i)}\}_{i=1}^N$,
then by Lemma~\ref{lem:muneqmupuniversal} $\{\mu^{(i)}\}_{i=1}^N$ is universal on $\nu$.
On the other hand if $\{\mu^{(i)\prime}\}_{i=1}^N=\{\mu^{(i)}\}_{i=1}^N$,
we have that $\nu_1=\sum_{i=1}^N\mu^{(i)}_1$ while 
$$\nup_1=\max_{1\leq i\leq N}\{\mu^{(i)\prime}_1\}=\max_{1\leq i\leq N}\{\mu^{(i)}_1\}.$$  
Since $N>1$, the sum is strictly greater than the maximum and
so $\nu_1>\nup_1$.  Thus $\nu'\neq\nu$ and by Lemma~\ref{lem:ifnohooksnorconjthenu} $\{\mu^{(i)}\}_{i=1}^N$ is again universal on $\nu$.
\end{proof}

\begin{exmp}
The pair $(\ydiagram{2,1},\ydiagram{2,1})$ is not $d$-universal for $d>2$ because
it is not universal on $\ydiagram{3,2,1}$, nor on $\left\{\ydiagram{3,3},\ydiagram{2,2,2}\right\}$.
But constrained to a highest weight vector, the only physically relevant partition
in 
$\cP_{\subydiagramtext{2,1}\,\subydiagramtext{2,1}}^{(d)}$
is
\setydiagrameq
$$
\ydiagram{2,1}+\ydiagram{2,1}=\ydiagram{4,2},
$$
\setydiagramtext
on which $(\ydiagram{2,1},\ydiagram{2,1})$ is universal.
\end{exmp}

\newpage

\chapter{Conclusions}

To recap, we have given a mathematical definition of exchange-only universality, applicable to any number of encoded qudits of any dimension.  We have shown this is a reasonable definition by proving the equivalence with a physically motivated
alternative definition.  We have furthermore conveniently expressed the definition in terms of partitions that label the qudits,
through their identification with irreducible representations of the symmetric group.

Having established universality as a property of partitions, we proceeded to derive conditions on the partitions for  universality to be possible on the associated qudits.  Most notably we derived necessary and sufficient conditions for universality on any number of qudits.  In particular we derived efficiently decidable, 
necessary and sufficient conditions for universality on one or two qudits, in terms of simple conditions on the associated partitions.  We also showed that universal families of qudits have an upward closed property, and that for sufficiently many qudits, universality is guaranteed.  Further we showed how to physically engineer universality when it is lacking, for example by simply adding at most five ancillae.

Mathematically, a cornerstone of all of the above results was supplied by Marin's work on the irreducible representations of the algebra of transpositions.  Indeed, Theorem~\ref{thm:iff1partitionu} giving necessary and sufficient conditions for the universality of a single qudit, and Theorem~\ref{thm:ifnohookslargenu} implying universality of sufficiently many qudits, both
follow quite readily from Marin's work.  The remaining seven theorems derived here
(Theorems~\ref{thm:iffumulti}, \ref{thm:iff2partitionsu}, \ref{thm:upwardclosed}, \ref{thm:iffunonselfcon}, \ref{thm:weakeqstrong}, \ref{thm:five}, \ref{thm:highestweight})
were made possible by certain key results original to this thesis,
 a principal one being that proper non-self-conjugate partitions ensure universality, while self-conjugate partitions only enable universality under certain conditions.  This was arrived at by careful examination of the isotypical decomposition of the representation induced by a tensor product of qudits.  Another key result provided the conditions under which universality is possible on hooks.  The latter was derived from dimension counting, considerations about irreducible Lie-algebra representations, and the Littlewood--Richardson rules.

Two of the above theorems stand out for requiring longer and more intricate proofs than the others.
Such may be an indication of their significance, or that simpler proofs were overlooked.  In either case, they seem worthy of special attention.   One is the equivalence of weak and strong universality, Theorem~\ref{thm:weakeqstrong}, which resulted from the accumulation of a 
variety of other results, including analysis of the alternating intertwiner that Marin found to relate conjugate representations of the algebra of transpositions, application of the Littlewood--Richardson rules, as well as properties of Lie algebras.
The other theorem of note, Theorem~\ref{thm:iff2partitionsu}, gave necessary and sufficient conditions for the universality of two qudits, in arithmetic terms of the partitions.  The proof of this theorem made critical use of the Weyl inequalities for the Littlewood--Richardson coefficients, as well as the work of Klyachko, Knutson and Tao relating these coefficients to eigenvalues of Hermitian matrices.

We conclude with some open problems.  
Chief among them is to complete the characterization of universal families of partitions.
Since we succeeded at characterizing universal families of nontrivial partitions,
what remains is to allow for trivial partitions, specifically in families that otherwise contain only hooks.
Universal families excluding hooks are already characterized by lemmas in subsections~4.2.1--4.2.2, 
which allow for inclusion of trivial partitions.  

Although trivial partitions label one-dimensional representations and are thus not useful as qudits by themselves, 
they could play the role of ancillae if they make a nonuniversal family universal.
A nonuniversal family of hooks can often be assessed as such by the dimension of the induced representation labeled by the larger hook that results from the Littlewood--Richardson rules.
This can be remedied by adding trivial partitions to the family, which increase the dimension of the induced representation without increasing the dimension of the relevant isotypical subspace.  
The challenge is to prove whether the intertwiner condition of a universality-witnessing map can
be satisfied in this way, which is not at all clear.

\begin{exmp}
It is not known whether $\left(\ydiagram{3,1,1},\ydiagram{4}\right)$ is $d$-universal, for $d\geq 3$.
\end{exmp}

Another open problem is that of completing the efficient characterization of universal families of partitions, or at least of nontrivial partitions.  For the latter, what remains is to efficiently characterize universal families of more than two self-conjugate partitions.  The goal here is to find necessary and sufficient conditions for $d$-universality, preferably in a simple arithmetic form, that can be determined by a polynomial-time algorithm.  By the discussion in subsection~4.2.5 what is needed, more specifically, is an algorithm to determine whether $\cP_\mus^{(d)}$ contains conjugates, given $\{\mu^{(i)}\}_{i=1}^N=\{\mu^{(i)\prime}\}_{i=1}^N$.  
As mentioned in subsection~4.2.5, the Weyl inequalities alone
are inadequate for this task, suggesting the use of more Horn inequalities.  
Fortunately, determining whether a given Littlewood--Richardson coefficient is nonvanishing can still theoretically be accomplished in polynomial time \cite{LRcomplexity}.
But the number of Littlewood--Richardson coefficients that must be checked grows superpolynomially in the sum of the sizes of the partitions in the family, thus motivating the search for a more efficient procedure.

Finally, once a family of partitions has been proven universal, there remains the task of explicitly
constructing the universality-witnessing function.  This is an important step towards physical applications,
which require explicit construction of given unitary operators from exchange-interactions.  Significant progress
has already been made for the 2-universal $(\ydiagram{2,1},\ydiagram{2,1})$ \cite{divincenzo:qc2000b,PhysRevA.99.042331,DBLP:journals/qic/FongW11} and $(\ydiagram{2,2},\ydiagram{2,2})$ \cite{Hsieh2003} cases.  
For the purpose of practical exchange-only quantum computation, there is more work to be done on these and other cases.

\bibliographystyle{plain}
\bibliography{Thesis}

\end{document}